\documentclass[ss,noshowframe,preprint]{imsart}

\RequirePackage{amsthm,amsmath,amsfonts,amssymb}
\RequirePackage[authoryear]{natbib}
\RequirePackage[colorlinks,citecolor=blue,urlcolor=blue]{hyperref}
\RequirePackage{graphicx}
\usepackage{tikz}
\usetikzlibrary{shapes.geometric, arrows.meta, positioning, fit, backgrounds, decorations.pathreplacing}
\usepackage{mathtools}
\usepackage{array}
\usepackage{etoolbox}
\usepackage{dsfont}
\usepackage{booktabs}
\usepackage{algorithm}
\usepackage{algpseudocode}
\usepackage{cleveref}

\startlocaldefs

\theoremstyle{plain}

\newtheorem{theorem}{Theorem}[section]
\newtheorem{lemma}[theorem]{Lemma}
\newtheorem{proposition}[theorem]{Proposition}

\crefname{lemma}{lemma}{lemmas}
\crefname{proposition}{proposition}{propositions}

\theoremstyle{definition}
\newtheorem{definition}{Definition}
\newtheorem{example}{Example}
\newcommand{\exampleendsymbol}{\hfill\scalebox{1.3}{$\diamond$}}
\AtEndEnvironment{example}{%
  \exampleendsymbol
  \vspace{0.5\baselineskip}
}

\theoremstyle{remark}

\newtheorem{remark}{Remark}

\newcommand{\inlinedef}[1]{\textit{#1}}

\DeclareMathOperator*{\argmin}{argmin}
\newcommand{\Exp}[1]{\exp\mathopen{}\left\{#1\right\}\mathclose{}}
\newcommand*{\norm}[1]{\left\lVert#1\right\rVert}

\newcommand{\E}{\mathbb{E}}
\newcommand{\R}{\mathbb{R}}
\newcommand{\Var}{\mathrm{Var}}
\newcommand{\Cov}{\mathrm{Cov}}
\newcommand{\law}{\mathrm{law}}
\newcommand{\given}{\mid}
\newcommand{\iid}{\mathrm{iid}}
\newcommand{\simiid}{\overset{\iid}{\sim}}
\newcommand{\indicator}{\mathds{1}}
\newcommand{\Def}{\coloneqq}

\newcommand{\Gaussian}{\mathcal{N}}
\newcommand{\lognormal}{\mathrm{Lognormal}}
\newcommand{\GP}{\mathcal{GP}}
\newcommand{\GaussianQuantile}{q}
\newcommand{\BigO}{\mathcal{O}}
\newcommand{\cst}{\varsigma}
\let\oldd\d\renewcommand\d{\relax\ifmmode\mathrm{d}\else\oldd\fi}

\newcommand{\Par}{\theta}
\newcommand{\ParTrue}{\Par_{\star}}
\newcommand{\ParBatch}{\boldsymbol{\Par}}
\newcommand{\parSpace}{\Theta}
\newcommand{\dimPar}{D}
 
\newcommand{\obs}{y_o}
\newcommand{\noise}{\epsilon} 
\newcommand{\covNoise}{\Sigma} 
\newcommand{\obsSpace}{\mathbb{Y}}
\newcommand{\dimObs}{P}

\newcommand{\fwd}{\mathcal{G}}
\newcommand{\lik}{\mathsf{L}}
\newcommand{\priorDens}{\pi_0}
\newcommand{\normCst}{Z}
\newcommand{\postDens}{\pi}
\newcommand{\jointDens}{\widetilde{\postDens}}

\newcommand{\simObsProcess}{\varphi}
\newcommand{\simObsDist}{\eta}
\newcommand{\simObsSpace}{\Phi}
\newcommand{\trainData}{\mathcal{D}}
\newcommand{\Ndesign}{N}
\newcommand{\Em}[2][\Ndesign]{{#2}_{#1}}
\newcommand{\simObsPredProcess}{\Em{\simObsProcess}}
\newcommand{\simObsPredDist}{\Em{\eta}}
\newcommand{\target}{\mathsf{f}}
\newcommand{\targetEm}[1][\Ndesign]{f_{#1}}
\newcommand{\targetTraj}{f}
\newcommand{\targetRange}{\mathbb{F}}
\newcommand{\emMean}[1][\Ndesign]{\mu_{#1}}
\newcommand{\kerMat}{K}
\newcommand{\emSD}[1][\Ndesign]{s_{#1}}
\newcommand{\emVar}[1][\Ndesign]{s^2_{#1}}
\newcommand{\emKer}[1][\Ndesign]{k_{#1}}
\newcommand{\emDist}[1][\Ndesign]{\nu_{#1}}
\newcommand{\emE}{\E_{\emDist}}
\newcommand{\quantileProb}{\alpha}
\newcommand{\emQ}{\mathbb{Q}^\quantileProb}
\newcommand{\jointEm}[1][\Ndesign]{\Em[#1]{\jointDens}}
\newcommand{\postEm}{\Em{\postDens}}
\newcommand{\likEm}{\Em{\lik}}
\newcommand{\normCstEm}{\Em{\normCst}}
\newcommand{\idxBasis}{r}
\newcommand{\dimBasis}{R}
\newcommand{\basisVec}{\psi}

\newcommand{\ep}{\mathrm{ep}}
\newcommand{\eup}{\mathrm{eup}}
\newcommand{\postApproxMean}{\postEm^{\mathrm{mean}}}
\newcommand{\postApproxEP}{\postEm^{\ep}}
\newcommand{\postApproxEUP}{\postEm^{\eup}}
\newcommand{\postApproxQuantile}{\postEm^{\quantileProb}}
\newcommand{\jointApproxMean}{\jointEm^{\mathrm{mean}}}
\newcommand{\jointApproxEUP}{\jointEm^{\eup}}

\newcommand{\Nbatch}{B}
\newcommand{\batchResponse}{\boldsymbol{\simObsProcess}}
\newcommand{\designResponse}{\boldsymbol{\simObsProcess}_{\Ndesign}}
\newcommand{\Naugment}{\Ndesign + \Nbatch}
\newcommand{\designIdx}{n}
\newcommand{\design}{\boldsymbol{\Par}_{\Ndesign}}
\newcommand{\acq}{\mathcal{A}}
\newcommand{\ParBatchAug}{\ParBatch_{\Naugment}}
\newcommand{\ParBatchOpt}{\boldsymbol{\ParBatch}_{\star}}
\newcommand{\surmetric}{\mathcal{H}}
\newcommand{\CovComb}[1][\Ndesign]{C_{#1}}
\newcommand{\weightdens}{\rho}
\newcommand{\varInflation}{\tau}

\newcommand{\loss}{\mathcal{L}}
\newcommand{\qDens}{q}
\newcommand{\qDensOpt}{q_\star}
\newcommand{\qFunc}{\widetilde{q}}
\newcommand{\qFuncOpt}{\widetilde{q}_\star}
\newcommand{\qSpace}{\mathcal{Q}}
\newcommand{\qFuncSpace}{\widetilde{\mathcal{Q}}}
\newcommand{\risk}{\mathcal{R}}
\newcommand{\bayesrisk}[1][\Ndesign]{\risk_{#1}}

\newcommand{\sampleIndex}{k}
\newcommand{\NSample}{K}

\newcommand{\accProbMH}{\alpha}

\newcommand{\approxPost}{\widehat{\postDens}}
\newcommand{\approxLik}{\widehat{\lik}}
\newcommand{\numMCSamp}{M}
\newcommand{\mcSampIdx}{m}
\newcommand{\latent}{z}
\newcommand{\latentDens}{p_z}
\newcommand{\predObs}{y}
\newcommand{\smoothingKernel}{\kappa}
\newcommand{\summaryStat}{\mathcal{S}}
\newcommand{\summaryDim}{S}
\newcommand{\bandwidth}{\epsilon}
\newcommand{\SLMean}{m}
\newcommand{\SLCov}{C}
\newcommand{\finiteDimPar}{\phi}
\newcommand{\condDensEst}[1][\finiteDimPar]{q_{#1}}

\newcommand{\Time}{t}
\newcommand{\state}{x}
\newcommand{\odeRHS}{F}
\newcommand{\timeStart}{\Time_0}
\newcommand{\timeEnd}{\Time_1}
\newcommand{\stateIC}{\state_{\circ}}
\newcommand{\solutionOp}{\mathcal{M}}
\newcommand{\obsOp}{\mathcal{H}}

\newcommand{\dimState}{S}
\newcommand{\stateIndex}{s}

\newcommand{\NTimeStep}{K}
\newcommand{\indexState}[2][\stateIndex]{{#2}^{({#1})}}
\newcommand{\indexTime}[2][\timeIndex]{{#2}_{#1}}

\endlocaldefs

\begin{document}
\begin{frontmatter}
\title{Surrogate-Based Bayesian Inference: Uncertainty Quantification and Active Learning}
\runtitle{Surrogate-Based Bayesian Inference}

\begin{aug}

\author[CDS]{\fnms{Andrew Gerard}~\snm{Roberts}\ead[label=e1]{arober@bu.edu}},
\author[EE]{\fnms{Michael C.}~\snm{Dietze}\ead[label=e2]{dietze@bu.edu}}, \\
\and
\author[CDS,MS]{\fnms{Jonathan H.}~\snm{Huggins}\ead[label=e3]{huggins@bu.edu}}

\address[CDS]{Faculty of Computing and Data Sciences,
Boston University\printead[presep={,\ }]{e1}}

\address[EE]{Department of Earth and Environment, 
Boston University\printead[presep={,\ }]{e2}}

\address[MS]{Department of Mathematics and Statistics, 
Boston University\printead[presep={,\ }]{e3}}
\end{aug}

\begin{abstract}
Surrogate models---also called emulators---are widely used to facilitate Bayesian 
inference in settings where computational costs preclude the use of 
standard posterior inference algorithms. 
Their deployment is now standard practice across many
scientific domains.
However, integrating surrogates in statistical analyses
introduces unique challenges that complicate established
Bayesian workflow principles. 
While significant progress has been 
made in addressing these issues,
the relevant developments are scattered across several distinct research communities, with different emphases and perspective.
We present a unifying review 
that synthesizes the literature into a coherent framework, aiming to 
benefit both practitioners and methods developers. We place particular emphasis 
on propagating surrogate uncertainty and sequentially refining emulators via 
active learning, two key components of a robust surrogate-based Bayesian workflow.
\end{abstract}



\end{frontmatter}

\tableofcontents

\section{Introduction} \label{sec:intro}
In fields ranging from climate science to astrophysics, the scientific process increasingly relies on complex, 
computationally expensive computer simulations \citep{ESM_modeling_2pt0,ESMReview,VarmaBlackHole2019,npeAstrophysics}. 
Many such computer models contain parameters with values that cannot be specified
a priori or directly estimated experimentally. Therefore, a critical task in many scientific workflows 
is to infer parameter values that are consistent with observed data. Bayesian inference offers 
a principled framework for parameter estimation, providing mechanisms to 
quantify uncertainty and balance prior information with observed data \citep{gelmanBDA}.
However, posterior approximation algorithms such as Markov chain Monte Carlo (MCMC; \citet{mcmcHandbook})
and Approximate Bayesian Computation (ABC; \citet{ABCPrimer})
require running the simulator tens or hundreds of thousands of times.
When a single simulation takes minutes or hours, this cost renders these standard approaches intractable.
To overcome this barrier, a rich literature has developed around \textit{surrogate modeling} (i.e., \textit{emulation}). 
The core idea is to replace the expensive simulator with a computationally-cheap statistical approximation, trained
on a sparse set of simulator runs \citep{KOH,gramacy2020surrogates}. The emulator is then used in place 
of the expensive model, enabling the application of otherwise intractable inference algorithms.

While the premise is straightforward, deploying surrogates within a Bayesian framework presents 
distinct methodological and conceptual challenges.
In broader statistical practice, 
established principles of \textit{Bayesian workflow} help practitioners navigate the 
complexities of iterative model building 
\citep{BayesianWorkflow,AmortizedBayesianWorkflow,VisBayesianWorkflow}. 
However, the incorporation of surrogates introduces complexities that 
extend beyond standard considerations. To bridge this gap, we consolidate a 
wide body of research in order to lay out the principles
underlying a robust workflow for surrogate-based Bayesian inference.

Foremost among the complexities faced in this setting 
is the sheer modeling challenge of 
fitting an accurate emulator with sparse training data. 
The quantities typically targeted by emulators, such as
simulator outputs or the log-likelihood surface, tend to be high-dimensional or nonstationary, 
which poses a challenge for standard predictive models \citep{nonstationaryGP}. This difficulty is compounded 
by the fact that the posterior distribution often concentrates in a small, unknown subset 
of the prior support \citep{Li_2014,murrayNPE,PCEBIP}. 
The sparse set of simulator training data may completely miss this ``needle in a haystack.'' An emulator
that is accurate in a global, prior-averaged sense may nonetheless fail to provide a reasonable
approximation in the high-density region \citep{SinsbeckNowak}. Even if these issues are adequately addressed,
the surrogate will invariably be an imperfect representation of the true model. Posterior estimates derived 
using the surrogate will inherit these imperfections, potentially leading to biased and overconfident 
inference \citep{BurknerSurrogate,RobertsUncProp}.

Despite remaining challenges, significant progress has been made in mitigating these concerns.
A driving feature of these advancements is the use of emulators equipped with a notion of predictive 
uncertainty; i.e., models that ``know what they don't know'' \citep{epistemicNN}. The use 
of uncertainty-aware surrogates offers several benefits and opportunities. Emulator uncertainty
can be propagated through the Bayesian workflow, enabling well-calibrated posterior inference that 
properly accounts for the additional source of uncertainty \citep{BilionisBayesSurrogates}.
Acknowledging the full scope of uncertainty is often crucial for downstream decision-making
and scientific inquiry \citep{DecisionMakingUncertainty}. 
Moreover, probabilistic surrogates can be used to inform 
the efficient allocation of computational resources. As noted above, passively sampling
simulator evaluation locations often yields training data that is poorly 
suited to posterior estimation \citep{Kandasamy2017}. When constrained by a limited computational budget, 
it is therefore crucial to carefully design the simulation schedule with respect to 
the explicit goal of approximating the posterior \citep{wang2018adaptive,Surer2023sequential}.
An uncertainty-aware emulator facilitates the identification of parameter locations with 
the greatest potential to improve the posterior estimate. This idea forms the
basis of \textit{active learning} algorithms that sequentially allocate simulator runs
optimized for posterior refinement \citep{KandasamyActiveLearning2015,VehtariParallelGP,weightedIVAR}.

\subsection{Overview and Paper Organization}
The topics highlighted above are highly interdisciplinary, spanning the statistical, 
applied mathematics, engineering, and machine learning literatures. 
In addition, the use of surrogates in practical applications has become 
standard across many scientific domains. By synthesizing these widespread 
developments, we aim to provide a unified treatment that benefits both practitioners and 
methods developers alike. We begin by providing a high-level overview of the paper 
structure.

To make the setting 
concrete, consider the task of inferring parameters $\Par$ using observations 
$\obs$ under a Bayesian model $p(\Par, \predObs) = \priorDens(\Par)p(\predObs \given \Par)$.
We assume that characterizing the posterior density 
$\postDens(\Par) = p(\Par \given \obs)$ requires evaluations of an expensive 
function $\target(\Par)$, which typically implies running a computer simulation.
The simulator may be deterministic---in which case we directly observe the value
of $\target(\Par)$ when the function is queried---or stochastic, such that we 
only observe $\simObsProcess(\Par)$, a noisy and potentially indirect observation
of $\target(\Par)$. \Cref{sec:background} motivates this setup, describing the
challenges imposed by expensive models in both the deterministic and noisy
(i.e., simulation-based/likelihood-free) settings.
Since $\target$ is assumed to be the computational bottleneck,
this review focuses on methods that make efficient use of simulator runs
by training a predictive model on a sparse set of evaluations
$\{\Par_\designIdx, \simObsProcess_\designIdx\}_{\designIdx=1}^{\Ndesign}$,
where $\simObsProcess_\designIdx = \simObsProcess(\Par_\designIdx)$.
\Cref{fig:workflow} summarizes the
high-level workflow that guides the remainder of this article,
which we now outline.

\begin{figure}
\centering
\begin{tikzpicture}[
    font=\sffamily,
    base/.style={draw, thick, rounded corners, align=center, fill=white,
                 minimum height=3.5em, minimum width=4.5cm},
    small_base/.style={base, minimum width=2.1cm, font=\small},
    expensive/.style={base, fill=gray!15, line width=1.5pt},
    outer_box/.style={draw, dashed, inner sep=10pt, inner ysep=18pt, rounded corners},
    outer_label/.style={font=\bfseries\sffamily, anchor=north west,
                        xshift=5pt, yshift=-5pt},
    arrow/.style={-Stealth, thick}
]


\node[small_base] (define_map)
  {Define Target \\ $\Par \mapsto \target(\Par)$};

\node[small_base, right=0.3cm of define_map] (init_design)
  {Prior Design \\ $\{\Par_\designIdx\} \sim \priorDens$};

\node[expensive, below=0.8cm of define_map, xshift=1.2cm] (init_sim)
  {Initial Simulations \\
   $\trainData \gets \{\Par_\designIdx, \target(\Par_\designIdx)\}$};

\begin{scope}[on background layer]
  \node[outer_box, fill=green!5,
        fit=(define_map) (init_design) (init_sim)] (stage1) {};
\end{scope}
\node[outer_label] at (stage1.north west) {Initial Design};


\node[base, below=1.0cm of stage1.south] (fit)
  {Fit Surrogate \\ $\widehat{\target} \approx \target$};

\begin{scope}[on background layer]
  \node[outer_box, fill=violet!5, fit=(fit)] (stage_fit) {};
\end{scope}
\node[outer_label] at (stage_fit.north west) {Surrogate Model};


\node[base, below=1.0cm of stage_fit.south] (post)
  {Approximate Posterior \\
   $\approxPost \approx \postDens$};

\begin{scope}[on background layer]
  \node[outer_box, fill=orange!5, fit=(post)] (stage3) {};
\end{scope}
\node[outer_label] at (stage3.north west) {Uncertainty Propagation};



\node[base, right=1.0cm of init_design] (choose)
  {Choose New Inputs \\ $\{\Par_i^\star\} \gets \argmin \acq(\{\Par_i\})$};

\node[expensive, below=0.8cm of choose] (new_sim)
  {New Simulations \\
   $\trainData_{\star} \gets
   \{\Par_\designIdx^{\star}, \target(\Par_\designIdx^{\star})\}$};

\node[base, below=0.8cm of new_sim] (update)
  {Augment Design \\ $\trainData \gets \mathcal{D} \cup \mathcal{D}_{\star}$};


\begin{scope}[on background layer]
  \node[outer_box, fill=blue!5,
        fit=(choose) (new_sim) (update)] (stage2) {};
\end{scope}
\node[outer_label] at (stage2.north west) {Active Learning};


\draw[arrow] (define_map.south) -- (define_map.south |- init_sim.north);
\draw[arrow] (init_design.south) -- (init_design.south |- init_sim.north);

\draw[arrow] (init_sim.south) -- (fit.north);

\draw[arrow] ([yshift=5pt]fit.east) -- ++(0.5,0) |- (choose.west);

\draw[arrow] (choose.south) -- (new_sim.north);
\draw[arrow] (new_sim.south) -- (update.north);

\draw[arrow] (update.west) -- ++(-0.25,0) |- ([yshift=-5pt]fit.east);

\draw[arrow] (fit.south) -- (stage3.north);

\draw[arrow, dashed] (stage3.east) -| (stage2.south);

\end{tikzpicture}
\caption{(Modular Surrogate-Based Bayesian Workflow) For simplicity, the diagram shows the
case where noiseless target evaluations $\target(\Par)$ are directly available, but the
workflow also applies to the case where queries yield noisy or indirect information
(see \Cref{sec:reg-emulators}). The broad workflow stages consist
of (1) \textbf{Initial design}: generate an initial set of training data by 
running the simulator and evaluating the function targeted for emulation; 
(2) \textbf{Surrogate training}: fit a probabilistic surrogate to the training set; 
(3) \textbf{Uncertainty propagation}: construct an uncertainty-aware posterior
approximation using the emulator. The final component, \textbf{active learning},
describes the iterative process of augmenting the training data
and updating the emulator. The current approximate posterior is sometimes used
in informing the selection of new query points (dashed arrow). 
The gray, shaded boxes highlight the steps that require calls to the expensive simulator.}
\label{fig:workflow}
\end{figure}
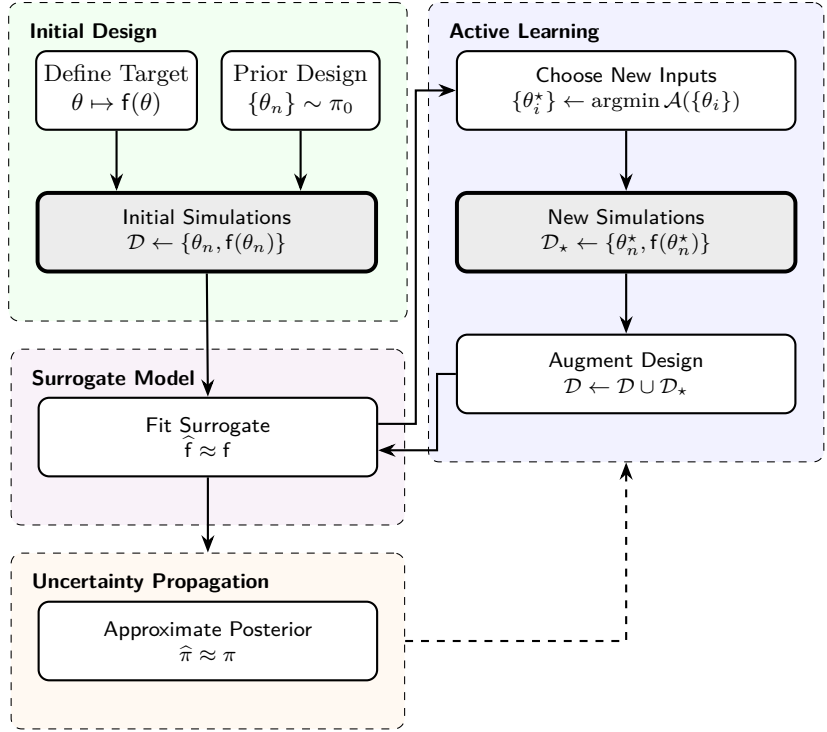

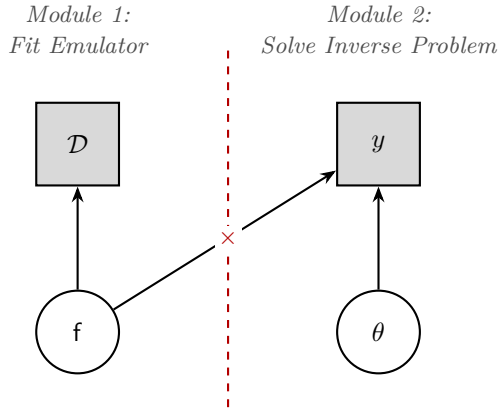
\begin{figure}
\centering
\begin{tikzpicture}[
    param/.style={circle, draw, thick, minimum size=1.1cm, inner sep=0pt},
    data/.style={rectangle, draw, thick, fill=gray!30, minimum size=1.1cm, inner sep=0pt},
    arrow/.style={-{Stealth[length=6pt]}, thick},
    cutarrow/.style={-{Stealth[length=6pt]}, thick},
    every node/.style={font=\normalsize},
  ]

  \node[param] (f) at (0, 0) {$\target$};
  \node[param] (theta) at (4, 0) {$\Par$};

  \node[data] (DN) at (0, 2.5) {$\trainData$};
  \node[data] (y) at (4, 2.5) {$\predObs$};

  \draw[arrow] (f) -- (DN);

  \draw[arrow] (theta) -- (y);

  \draw[cutarrow] (f) -- (y);

  \draw[dashed, thick, red!70!black]
    (2, -1) -- (2, 3.8);

  \node[fill=white, inner sep=1.5pt, font=\small\bfseries, text=red!70!black]
    at (2, 1.25) {$\times$};

  \node[anchor=south, font=\small\itshape, text=black!70, align=center]
    at (0, 3.6) {Module 1:\\[1pt] Fit Emulator};

  \node[anchor=south, font=\small\itshape, text=black!70, align=center]
    at (4, 3.6) {Module 2:\\[1pt] Solve Inverse Problem};

\end{tikzpicture}
\caption{Graphical representation of a joint Bayesian model, 
with a ``cut'' (red dashed line) that prevents feedback from $\predObs$ to $\target$.
Shaded squares are observed variables  
(data $\predObs$ and simulator queries $\trainData$), while circles 
are unobserved variables (parameters $\Par$ and target map $\target$). The joint
Bayesian approach constructs a probability distribution over all quantities, then
computes $p(\Par, \target \given \predObs, \trainData)$. The cut model 
severs feedback so that $\predObs$ does not affect inference for $\target$.}
\label{fig:cut}
\end{figure}

\paragraph{Emulator Target and Initial Design.}
When introducing an emulator within an existing Bayesian model, one of the first 
decisions is choosing which quantity to emulate. 
\Cref{sec:reg-emulators,sec:practical-em-target} discuss
the choice of \inlinedef{emulator target} in depth.
In general, we consider emulators that seek to learn a target map of the form 
$\Par \mapsto \target(\Par)$. This encompasses popular strategies such as 
\inlinedef{forward model emulation} (predicting the outputs of an underlying 
computer model), \inlinedef{log-density emulation} (directly learning
the log-likelihood or log-posterior surface), and 
\inlinedef{conditional likelihood estimation} (fitting a parameterized likelihood).
Once the target is selected, the initial emulator training set 
$\trainData \Def \{\Par_\designIdx, \simObsProcess_\designIdx\}_{\designIdx=1}^{\Ndesign}$
is typically constructed by sampling input points from the prior
$\{\Par_\designIdx\}_{\designIdx=1}^{\Ndesign} \sim \priorDens$ 
(potentially with a space-filling adjustment to prevent clumping) and then querying
the target at the sampled points. We follow the statistical design of experiments 
literature by referring to the input set $\{\Par_\designIdx\}_{\designIdx=1}^{\Ndesign}$ 
as the \inlinedef{design} \citep{initDesignReview,gramacy2020surrogates}.
The first stage (green box) in \Cref{fig:workflow} highlights the choice of 
target and the initial design construction.

\paragraph{Fitting the Emulator.}
The second stage (pink box) entails fitting a predictive model to the design data $\trainData$.
A variety of different model classes are commonly used for this task, including
Gaussian processes, polynomial chaos expansions, and neural networks 
\citep{StuartTeck2,BayesianPCE2,BayesOptNN}. To provide a unifying viewpoint, 
we avoid focusing on a particular model. Instead, we adopt a generic probabilistic 
perspective that encompasses any regression or interpolation model that,
in addition to point estimates, provides some notion of predictive uncertainty.
This is encoded by a probability distribution $p(\target \given \trainData)$
that summarizes the uncertainty in the target due to limited training 
data $\trainData$. \Cref{sec:prob-emulators} provides a general perspective
on probabilistic emulators, while \Cref{sec:model-classes,sec:finite-vs-inf}
highlight practical considerations for particular model classes.

\paragraph{Uncertainty Propagation.}
The trained surrogate is then deployed to estimate the posterior 
distribution $\postDens$ (yellow box). This can represent either the 
final output of the inference workflow or an intermediate approximation
used to guide the selection of additional design points. 
Our framework is rooted in a \textit{modular} workflow, in which 
the emulator predictive distribution $p(\target \given \trainData)$ and the 
parameter posterior $p(\Par \given \obs)$ are learned in two distinct
steps \citep{BurknerSurrogate}. Modular 
Bayesian workflows have been shown to improve robustness to 
misspecification in the likelihood $p(\obs \given \Par)$, in addition 
to offering computational advantages \citep{modularization,moduleModels}.
Given that the two inference stages are not 
coupled through an overarching probability model, uncertainty propagation 
is not immediately provided by standard Bayesian conditioning and marginalization.
In fact, several different uncertainty-aware posterior estimates are routinely
used in the literature, which approximate $p(\Par \given \obs)$ while accounting
for the uncertainty in $p(\target \given \trainData)$ \citep{BurknerSurrogate,RobertsUncProp}.
\Cref{sec:post-approx} reviews these approaches and their tradeoffs.

\begin{remark}
Our use of the word ``modular'' refers to the two-step computational workflow,
similar to its use in \citet{BurknerSurrogate,BurknerTwoStep}. This implies
that the emulator predictive distribution $p(\target \given \trainData)$
is obtained without referencing $\obs$, and then the uncertainty is
propagated when performing inference for $\Par$. This should not be 
confused with its usage in literature on cutting feedback in 
probabilistic models, in which case ``modular'' implies 
that information never
flows from $\obs$ to $\target$ during inference \citep{PlummerCut}. 
These usages are \textit{not} equivalent. A two-step computational workflow can still 
introduce feedback from $\obs$ to $\target$ during the second stage,
depending on how uncertainty is propagated from the first stage.
A graphical representation of cutting feedback is shown in \Cref{fig:cut}.
\end{remark}

\paragraph{Active Learning.} 
The final component of the workflow (purple box) consists of an iterative 
procedure for refining the surrogate. This represents a sequential algorithm
for constructing the design, in which the current surrogate is used to 
identify the next batch of simulator inputs. The dashed arrow in 
\Cref{fig:workflow} emphasizes that an intermediate posterior approximation
can be used to inform the design selection. In our context, we adopt a 
broad use of the term \inlinedef{active learning} to refer to algorithms 
that select new design points with respect to the specific goal of 
posterior estimation. This includes strategies that
optimize goal-oriented acquisition functions (analogous to Bayesian 
optimization; \Cref{sec:design-optimization}), iterative sampling
of approximate posteriors (\Cref{sec:surrogate-guided-sampling}),
tempering strategies (\Cref{sec:tempering}), and MCMC-guided
design (\Cref{sec:mcmc-design}).

\section{Bayesian Inference with Expensive Models} \label{sec:background}
This paper concerns Bayesian inference problems in which evaluation of the 
likelihood (or more generally the unnormalized posterior density) 
incurs significant computational expense. This situation is exceedingly common 
in large-scale scientific and engineering applications, where the goal is 
to infer latent parameters related to observations through a complex computer model. 
When the computer model is stochastic, the challenge is compounded: computational
costs remain prohibitive, and the density itself is often mathematically intractable 
to evaluate. While many simulation-based (or likelihood-free) methods address this 
intractability by constructing an approximate likelihood, they do not avoid the 
computational burden. Evaluating the likelihood approximation still relies on 
repeated forward simulations of the underlying model, implying that the core 
problem of expensive density evaluations remains.
This section introduces the challenges faced in both the deterministic and 
stochastic settings, motivating the emulator-based methods considered in the 
remainder of the paper.

\subsection{The Bayesian Inference Setting}
Consider the task of inferring unknown parameters 
$\Par \in \parSpace \subseteq \R^{\dimPar}$ using noisy observations 
$\obs \in \obsSpace \subseteq \R^{\dimObs}$. The Bayesian approach to this
problem is to define a joint probability distribution 
$p(\Par, \predObs) = \priorDens(\Par)p(\predObs \given \Par)$ over 
the parameters and data. The prior $\priorDens$ is a probability 
distribution over $\parSpace$ and the conditional $p(\predObs \given \Par)$
defines the likelihood 
\begin{equation}
\lik(\Par) \Def p(\obs \given \Par),
\end{equation}
a function of $\Par$ dependent on the fixed data realization $\obs$.
The central task of Bayesian inference is to characterize the posterior distribution
\begin{align}
&\postDens(\Par) \Def p(\Par \given \obs) = 
\frac{1}{\normCst} \priorDens(\Par) \lik(\Par),
&&\normCst \Def \int_{\parSpace} \priorDens(\Par) \lik(\Par) \d\Par,
\end{align}
which encodes the uncertainty in the parameters after observing the data. 
The quantities $\lik$, $\postDens$, and $\normCst$ all depend on $\obs$,
but we suppress this in the notation as the data realization is assumed
constant throughout most of this paper.
Typically, the integral defining $\normCst$ is 
intractable, so posterior inference algorithms
rely on pointwise evaluations of the \textit{unnormalized} posterior density $\jointDens(\Par) \Def \priorDens(\Par)\lik(\Par)$. A standard approach is to sample the posterior using a Markov chain 
Monte Carlo (MCMC) algorithm, often requiring $10^5 - 10^7$ serial evaluations of $\jointDens(\Par)$. 
While not a problem when $\jointDens(\Par)$ is cheap to compute, this sequential computation 
renders MCMC infeasible when the cost of a density evaluation is high.

\subsection{Deterministic Bayesian Inverse Problems} \label{sec:bip}
The inference challenge posed by expensive posterior densities commonly arises in the Bayesian approach to 
inverse problems \citep{Stuart_BIP}. In this setting, the observation $\obs$ typically corresponds to an 
observable quantity associated with a complex physical system. The goal is to recover latent parameters of interest
$\Par$ that gave rise to the observed data. The bulk of the modeling effort typically consists of constructing a
mechanistic model $\fwd: \parSpace \to \obsSpace$ encoding domain knowledge, which describes the 
process by which parameters produce observed quantities. We refer to $\fwd$ as the \textit{forward model}.
The task of solving the inverse problem entails
inverting the relationship to identify parameter values consistent with a particular observation. This parameter recovery 
is commonly cast as a problem of Bayesian inference with a Gaussian noise model
\begin{equation}
\predObs = \fwd(\Par) + \noise, \qquad
\noise \sim \Gaussian(0, \covNoise), \label{eq:inv_prob_Gaussian} 
\end{equation}
implying the likelihood
\begin{equation}
\lik(\Par) = \det(2\pi\covNoise)^{-\dimObs/2} \Exp{-\frac{1}{2} \norm{\obs - \fwd(\Par)}^{2}_{\covNoise}}.
\label{eq:Gaussian-likelihood}
\end{equation}
We utilize the weighted norm notation $\norm{x}^{2}_{A} \Def x^\top A^{-1}x$, where $A$ is a positive definite matrix.
Computing $\lik(\Par)$ (and thus $\jointDens(\Par)$) requires the forward model evaluation $\fwd(\Par)$,
which may involve running an expensive computer simulation (e.g., a numerical differential equation solver). Hence, in Bayesian
inverse problems the posterior density $\jointDens(\Par)$ is often a computationally expensive, and potentially black box 
(i.e., difficult or impossible to differentiate), function.

\begin{example}[Parameter Estimation for ODEs]
\label{ex:ode}
In modeling the Earth system, complex dynamical models are used to simulate trajectories
of various physical processes, such as temperature trends and carbon fluxes
 \citep{ESM_modeling_2pt0,paramLSM,idealizedGCM,FATES_CES,CLMBayesianCalibration,FerEmulation}.
These models often feature empirical parameters with unknown values that must be estimated from data.
As a simple illustrative example, we consider the problem of parameter estimation for an ordinary 
differential equation (ODE). While this example is simplified for clarity, the setup is structurally similar to many 
inverse problems faced in practice, such as parameter estimation for land surface models \citep{paramLSM}. 
Consider a parameter-dependent initial value problem 
\begin{align}
\frac{\d}{\d\Time} \state(\Time, \Par) &= \odeRHS(\state(\Time, \Par), \Par), &&x(\timeStart, \Par) = \stateIC, \label{ode_ivp}
\end{align}
describing the time evolution of $\dimState$ state variables 
$\state(\Time, \Par) \Def \left\{\indexState{\state}(\Time, \Par)\right\}_{\stateIndex=1}^{\dimState}$
with dynamics depending on $\Par \in \parSpace$. As our focus will be on estimating these parameters 
from observations, we consider the parameter-to-state map
\begin{align}
\Par &\mapsto \left\{\state(\Time, \Par) :  \Time \in [\timeStart, \timeEnd] \right\}.
\end{align}
In practice, the solution is typically approximated via a numerical discretization of the form 
\begin{align}
\solutionOp: \Par &\mapsto \left[\indexTime[0]{\state}(\Par), \dots, \indexTime[\NTimeStep]{\state}(\Par) \right]^\top. \label{eq:ode-solution-op}
\end{align}
Here, $\solutionOp: \parSpace \to \R^{\NTimeStep \times \dimState}$ represents the map induced by a numerical solver, 
and $\indexTime[0]{\state}(\Par), \dots, \indexTime[\NTimeStep]{\state}(\Par)$ are approximations of the state 
values $\state(\Time, \Par)$ at a finite set of time points in $[\timeStart, \timeEnd]$. 
Going forward, we focus on the discrete-time operator $\solutionOp$, neglecting discretization error
for simplicity. Finally, suppose we have observed data $\obs \in \obsSpace$ that we model as a
noise-corrupted function of the state trajectory. This is formalized by the definition of an observation operator 
$\obsOp: \R^{\NTimeStep \times \dimState} \to \obsSpace$ mapping from the state trajectory to a 
$\dimObs$-dimensional observable quantity. Assuming the data generating process
\begin{align}
&\predObs = (\obsOp \circ \solutionOp)(\ParTrue) + \noise, 
&&\noise \sim \Gaussian(0, \covNoise) \label{ode_inv_prob} 
\end{align}
for some ``true'' parameter value $\ParTrue \in \parSpace$, we see that this problem is of the form in 
\Cref{eq:inv_prob_Gaussian} with forward model $\fwd \Def \obsOp \circ \solutionOp$.
In this case, the computational
cost of evaluating $\postDens(\Par)$ can be high, owing to the dependence on the numerical solver $\solutionOp(\Par)$.
\end{example}

\subsection{Stochastic Bayesian Inverse Problems} \label{sec:sbi}
In many applications, the forward model may be inherently noisy \citep{stochasticComputerModels}. 
In this case, multiple evaluations
of $\fwd(\Par)$ at the same input may yield different outputs, and the simulator can 
be viewed as defining a family of conditional distributions $p(\fwd(\Par) \mid \Par)$ 
indexed by $\Par$. Let $\latent$ denote the latent variables that make the simulator 
noisy, such that $\fwd(\Par) = \fwd(\Par; \latent)$ is deterministic once a value for
$\latent$ is fixed. The simulator thus implicitly defines the data-generating process by
\begin{equation}
p(\predObs \given \Par) = 
\int p(\predObs \given \Par, \latent) \latentDens(\latent \given \Par) \d\latent.
\label{eq:marg-lik}
\end{equation}
Since this integral is typically intractable, evaluations of the likelihood 
$\lik(\Par) = p(\obs \given \Par)$ are unavailable. We apply the broad label 
\inlinedef{simulation-based inference (SBI)} to methods that
circumvent this difficulty by relying only on samples from the 
simulator, rather than likelihood evaluations \citep{frontierSBI,HartigStochasticReview}.
The term \textit{likelihood-free inference} is often used synonymously. 
In particular, we consider SBI to include purely black box methods 
which operate only on samples from $p(\predObs \given \Par)$,
as well as those that utilize samples from $\latentDens(\latent \given \Par)$
and exploit some tractable structure in $p(\predObs \given \Par, \latent)$.
Such approaches vary widely, with their applicability depending on the 
simulator cost and computational budget, as well as the presence of 
structure in \Cref{eq:marg-lik} that can be exploited. 
We describe several categories of methods below, each of which can be understood
as a means of constructing a tractable likelihood approximation 
$\approxLik(\Par) \approx p(\obs \given \Par)$
using information from simulator runs. The resulting (approximate) unnormalized
posterior density is thus subject to even more severe computational limitations
as compared to $\jointDens(\Par)$ in the deterministic setting, due to the fact
that computing $\approxLik(\Par)$ typically requires multiple replicate
simulations at the same input $\Par$.
Our goal here is not to provide a comprehensive survey of SBI, but rather 
to highlight classes of methods that are amenable to emulation, in the 
sense made precise in \Cref{sec:surrogates}.

\subsubsection{Pseudo-Marginal Methods} \label{sec:pseudo-marg}
In the case that $\lik(\Par)$ is intractable but the density
$p(\predObs \given \Par, \latent)$ can be evaluated, then an unbiased 
likelihood estimate is given by
\begin{align}
\overline{\lik}^\numMCSamp(\Par)
\Def \frac{1}{\numMCSamp} \sum_{\mcSampIdx=1}^{\numMCSamp} p(\obs \given \Par, \latent_\mcSampIdx),
&&\latent_\mcSampIdx \simiid \latentDens(\cdot \given \Par),
\end{align}
such that $\E_{\latent}[\overline{\lik}^\numMCSamp(\Par)] = \lik(\Par)$,
where the expectation is respect to the simulated $\{z_\mcSampIdx\}$.
When $\latentDens$ cannot be directly sampled, importance sampling and 
sequential Monte Carlo methods are commonly employed \citep{ParticleMCMC}.
It turns out that this is sufficient
to construct MCMC algorithms that exactly target the 
true posterior $\postDens$ \citep{pseudoMarginalMCMC}. However, 
the performance of these \inlinedef{pseudo-marginal} methods
degrades when the variance of the likelihood estimate is high,
which is often the case when the latent space is high-dimensional 
or the proposal distribution for $\latent$ is poor \citep{corrPM}.
The practical application of pseudo-marginal
samplers often requires many simulator calls at each value
of $\Par$ in order to control the estimator variance, which
is intractable for very expensive computer models. 
\inlinedef{Monte Carlo within Metropolis (MCwM)}, a closely related
algorithm, can exhibit better mixing with smaller values of $\numMCSamp$,
but implicitly targets an approximation 
to the true posterior \citep{noisyMCMC,stabilityNoisyMH}.

\subsubsection{Classical Likelihood-Free Methods} \label{sec:classical-lfi}
We next describe \textit{Approximate Bayesian Computation} and
\textit{Synthetic Likelihood}, two classical statistical frameworks for conducting SBI.
These approaches assume only the ability to generate samples 
$\predObs \sim p(\predObs \given \Par)$ from the generative 
model defined in \Cref{eq:marg-lik}.

\paragraph{Approximate Bayesian Computation.} 
Approximate Bayesian Computation (ABC) refers to a broad class of methods that 
target approximate posteriors of the form 
$\approxPost(\Par) \propto \priorDens(\Par) \approxLik(\Par)$, where
\begin{equation}
\approxLik(\Par) =
\int \smoothingKernel_{\bandwidth}(\summaryStat(\predObs), \summaryStat(\obs)) p(\predObs \given \Par) \d\predObs.
\label{eq:abc}
\end{equation}
The likelihood is thus approximated non-parametrically by specifying some smoothing 
kernel $\smoothingKernel_{\bandwidth}$ with bandwidth $\bandwidth$ and a choice of summary 
statistics $\summaryStat: \obsSpace \to \R^{\summaryDim}$ \citep{ABCPrimer,ABCApproxLik}.
This replaces the intractable likelihood with a smoothed version that 
weights simulated data $\predObs$ by its ``closeness'' to $\obs$. For example, if 
$\smoothingKernel_{\bandwidth}(s, s^\prime) = \indicator(\norm{s - s^\prime} < \epsilon)$,
then $\approxPost(\Par) \propto \priorDens(\Par) p(\norm{\summaryStat(\predObs) - \summaryStat(\obs)} < \epsilon \given \Par)$.
Since the integral in \Cref{eq:abc} is typically intractable, the likelihood approximation is replaced
with the noisy estimate
\begin{align}
&\overline{\lik}^\numMCSamp(\Par) \Def 
\frac{1}{\numMCSamp} \sum_{\mcSampIdx=1}^{\numMCSamp}
\smoothingKernel_{\bandwidth}(\summaryStat(\predObs_\mcSampIdx)), \summaryStat(\obs)),
&&\predObs_\mcSampIdx \simiid p(\predObs \given \Par).
\end{align}
As in the pseudo-marginal setting, $\numMCSamp$ simulator runs are required to evaluate
the likelihood approximation at a point $\Par$. Since the estimator is unbiased 
with respect to the likelihood in \Cref{eq:abc}, a pseudo-marginal algorithm can be 
used to target $\approxPost$ (which is itself an approximation of $\postDens$).

\paragraph{Synthetic Likelihood.}
As an alternative to the non-parametric ABC approximation, synthetic likelihood (SL)
methods instead opt for the parametric model
\begin{equation}
\approxLik(\Par) =
\Gaussian(\summaryStat(\obs) \given \SLMean(\Par), \SLCov(\Par)),
\label{eq:SL}
\end{equation}
resulting from a Gaussian assumption for the conditional 
$p(\summaryStat(\predObs) \given \Par)$ \citep{WoodSL,FrazierSL,PriceSL}.
Similar to ABC, the approximate likelihood is replaced with a noisy estimate 
\begin{align}
&\overline{\lik}^\numMCSamp(\Par) \Def 
\Gaussian(\summaryStat(\obs) \given \SLMean_\numMCSamp(\Par), \SLCov_\numMCSamp(\Par)),
&&\predObs_\mcSampIdx \simiid p(\predObs \given \Par),
\label{eq:SL-noisy}
\end{align}
where $\SLMean_{\numMCSamp}(\Par)$ and $\SLCov_{\numMCSamp}(\Par)$ denote the sample mean 
and covariance of $\{\summaryStat(\predObs_\mcSampIdx)\}_{\mcSampIdx=1}^{\numMCSamp}$.
Pseudo-marginal and MCwM schemes can again be employed for inference.
Unlike ABC, the Monte Carlo likelihood estimator is biased since
$\E[\Gaussian(\summaryStat(\obs) \given \SLMean_\numMCSamp(\Par), \SLCov_\numMCSamp(\Par))] \neq 
\Gaussian(\summaryStat(\obs) \given \SLMean(\Par), \SLCov(\Par))$. 
This implies that the target distribution of a pseudo-marginal sampler is not 
equal to $\approxPost$ and in general depends on $\numMCSamp$ \citep{PriceSL}.

\subsubsection{Neural Density Estimation} \label{sec:cond-dens-est}
The success of modern machine learning methods has inspired the use of 
more flexible conditional density estimators than those used in 
the classical likelihood-free approaches \citep{frontierSBI}.
Various methods have been proposed to approximate the family of conditional
densities $\{p(\predObs \given \Par)\}_{\Par \in \parSpace}$ by fitting 
a predictive model that maps 
\begin{equation}
\Par \mapsto \condDensEst(\predObs \given \Par) \approx p(\predObs \given \Par).
\end{equation}
The model $\condDensEst$ is typically taken to be a neural network, with 
parameters $\finiteDimPar$ estimated by minimizing the loss
\begin{align}
&\loss(\finiteDimPar) = 
-\sum_{\mcSampIdx=1}^{\numMCSamp} \log \condDensEst(\predObs_\mcSampIdx \given \Par_\mcSampIdx),
&&\predObs_\mcSampIdx \simiid p(\predObs \given \Par_\mcSampIdx),
\end{align}
for a set of chosen points $\{\Par_\mcSampIdx\}_{\mcSampIdx=1}^{\numMCSamp}$.
Once the model is trained, $\condDensEst(\predObs \given \Par)$ can be evaluated at 
any pair $(\Par, \predObs)$ without running additional simulations. Thus, unlike
the previous methods, evaluations of the approximate likelihood at new $\Par$ do
not incur additional simulator cost \citep{NeuralAmortizedReview}. 
Inference on the induced posterior approximation
$\approxPost(\Par) \propto \priorDens(\Par) \condDensEst(\obs \given \Par)$ is 
then carried out by standard methods such as MCMC. Similar workflows 
are used by related methods such as neural ratio estimation \citep{CranmerNRE}.
We note that the MCMC step can be avoided by directly approximating the reverse conditional
$\predObs \mapsto \condDensEst(\Par \given \predObs) \approx p(\Par \given \predObs)$, 
but such methods fall outside the scope of this paper \citep{murrayNPE}. 
A limitation of all neural density
estimation strategies is that they typically require a large number of simulations
to generate sufficient training data, limiting their applicability to expensive
models \citep{BurknerAmortizedSurrogate}.

\section{Emulators in Bayesian Inference} \label{sec:surrogates}
Given the inference bottleneck imposed by computationally expensive models, 
many methods have been proposed to approximate the posterior using only a small set of 
simulator calls \citep{KOH,dimRedPolyChaos,Higdon_2015,randMAP}. 
Many such approaches consist of replacing a computationally-limiting component of the 
model with a cheap \inlinedef{surrogate}, thus enabling 
the use of standard inference schemes on the resulting approximate 
model \citep{gramacy2020surrogates}. In this section, we introduce the 
particular sub-class of surrogate models that form the focus of this article.

\subsection{Regression-Based Emulators} \label{sec:reg-emulators}
While there is a vast literature on surrogates designed for specific model structures, 
we instead focus on the broadly applicable strategy of learning regression-based 
emulators from black-box model evaluations. Even within this class of surrogates, 
the regression modeling setup can vary widely, with a key decision being the choice 
of response variable to approximate with the regression model 
\citep{HigdonBasis,Surer2023sequential,scalarization}.

\begin{definition}[Emulator Target]
We define the \inlinedef{emulator target} as the deterministic
map $\target: \parSpace \to \targetRange$ that the surrogate seeks to learn, where
$\targetRange$ is the \inlinedef{target space} in which the emulated quantity lives.
\end{definition}

When relevant, we will augment the notation from \Cref{sec:background} to emphasize dependence
on the target; e.g., $\postDens(\Par; \target)$, $\lik(\Par; \target)$, $\normCst(\target)$, etc. 
We make the assumption throughout that $\jointDens(\Par; \target)$ depends on $\Par$ and $\target$ 
only as a function of $\target(\Par)$.
The target map should be chosen to alleviate the computational bottleneck; hence,
the density $\jointDens(\Par; \target)$ is cheap to compute once the expensive computation 
$\target(\Par)$ is performed.
We consider regression methods that approximate the target using training data 
$\{\Par_\designIdx, \simObsProcess_\designIdx\}_{\designIdx=1}^{\Ndesign}$
generated by evaluating $\target$ at a set of \inlinedef{design points} 
$\design \Def \{\Par_{\designIdx}\}_{\designIdx=1}^{\Ndesign}$.
The training responses $\simObsProcess_\designIdx$ may constitute 
noisy or indirect observations of $\target(\Par_\designIdx)$.

\begin{definition}[Simulator Observation Process]
\label{def:sim-obs}
The \inlinedef{simulator observation process} is a stochastic process that, given an
input $\Par \in \parSpace$, produces an observation $\simObsProcess(\Par) \in \simObsSpace$
according to
\begin{equation}
\simObsProcess(\Par) \sim \simObsDist(\cdot \given \target(\Par)),
\end{equation}
where $\simObsDist(\cdot \given \target(\Par))$ is a probability measure on $\simObsSpace$. 
The simulator noise is not necessarily independent 
across $\Par$, so in general $\simObsProcess(\ParBatch) \sim \simObsDist(\cdot \given \target(\ParBatch))$ 
gives the joint distribution over the responses at a set of inputs $\ParBatch \subset \parSpace$.
The training responses thus take the form
$\designResponse \Def 
\{\simObsProcess_{\designIdx}\}_{\designIdx=1}^{\Ndesign} \sim 
\simObsDist(\cdot \given \target(\design))$.
\end{definition}

In many cases considered throughout this survey, the target space $\targetRange$ and simulator
observation space $\simObsSpace$ coincide; i.e., $\simObsProcess_\designIdx$ is a direct 
noisy observation of $\target(\Par_\designIdx)$. However, there are notable exceptions
in which the simulator noise process is more complicated (e.g., \Cref{ex:cond-dens-target} below).
In the special case that these spaces coincide and 
the simulator observations are deterministic, we have 
$\eta(\cdot \given \target(\Par)) = \delta_{\target(\Par)}(\cdot)$ and 
$\simObsProcess_\designIdx = \target(\Par_\designIdx)$. The following examples identify target 
quantities that can be emulated to alleviate the computational burden in the problems described
in \Cref{sec:background}.

\begin{example}[Forward Model Target]
\label{ex:fwd-target}
Consider an inverse problem of the form $\predObs = \fwd(\Par) + \epsilon$, as in 
\Cref{eq:inv_prob_Gaussian}. In this case, we might set $\target \Def \fwd$, targeting 
the map implied by the forward model. If the forward model is deterministic, then 
$\simObsProcess_\designIdx = \fwd(\Par_\designIdx)$, in which case the goal of 
the emulator is to interpolate the points 
$\{\Par_\designIdx, \fwd(\Par_\designIdx)\}_{\designIdx=1}^{\Ndesign}$. Note that 
$\obsSpace = \targetRange = \simObsSpace$; the observation space, target space, 
and simulator observation space coincide.
\end{example}

\begin{example}[Log-Likelihood Target]
Consider the same setup as \Cref{ex:fwd-target}. Instead of targeting the underlying
forward model, another possibility is to target the (log) likelihood directly: 
$\target(\cdot) \Def \log \lik(\cdot)$. Since the simulator is deterministic, then
$\simObsProcess_\designIdx = \log \lik(\Par_\designIdx)$, so that the emulator
seeks to interpolate the training data 
$\{\Par_\designIdx, \log \lik(\Par_\designIdx)\}_{\designIdx=1}^{\Ndesign}$.
In this case, $\targetRange = \simObsSpace = \R$.
\end{example}

\begin{example}[Noisy Forward Model Target]
\label{ex:fwd-noisy-target}
Again consider an inverse problem $\predObs = \fwd(\Par) + \epsilon$, 
where now $\fwd(\Par) \sim \Gaussian(\SLMean(\Par), \SLCov(\Par))$. In this
stochastic setting, one possible target is 
$\target(\cdot) \Def \SLMean(\cdot) = \E[\fwd(\cdot)]$. The simulator 
generates noisy observations via 
$\simObsProcess(\Par) = \SLMean(\Par) + \noise(\Par)$, where
$\noise(\Par) \sim \Gaussian(0, \SLCov(\Par))$ 
(i.e., $\simObsDist(\cdot \given \target(\Par)) = \Gaussian(\target(\Par), \SLCov(\Par))$).
Therefore, the training observations are given by 
$\simObsProcess_\designIdx \sim  \Gaussian(\target(\Par_\designIdx), \SLCov(\Par_\designIdx))$
and the emulator seeks to smooth over the noisy training data
$\{\Par_\designIdx, \SLMean(\Par_\designIdx) + \noise(\Par_\designIdx)\}_{\designIdx=1}^{\Ndesign}$.
As in \Cref{ex:fwd-target}, we have $\obsSpace = \targetRange = \simObsSpace$.
The computer experiments literature has explored many generalizations of this setup,
including targeting other moments (e.g., $\target(\cdot) \Def [\SLMean(\cdot), \SLCov(\cdot)]$),
or quantiles of the distribution of $\fwd(\Par)$ \citep{Binois_2018,FadikarAgentBased}.
\end{example}

\begin{example}[ABC/SL Noisy Log-Likelihood Target]
\label{ex:llik-noisy-lfi-target}
Consider the SBI setting (\Cref{sec:classical-lfi}) with an approximate likelihood 
$\approxLik(\Par) \approx p(\obs \given \Par)$ constructed using either
ABC or SL. Both methods make use of noisy estimates 
$\overline{\lik}^{\numMCSamp}(\Par) \approx \approxLik(\Par)$ constructed from samples
$\{\predObs_\mcSampIdx\}_{\mcSampIdx=1}^{\numMCSamp} \sim p(\predObs \given \Par)$.
To eliminate the requirement for simulator runs at every evaluation point $\Par$, 
one can instead fit an emulator targeting 
$\target(\cdot) \Def \log \approxLik(\cdot)$ to sparse training data
$\{\Par_\designIdx, \simObsProcess_\designIdx\}_{\designIdx=1}^{\Ndesign}$.
The training response $\simObsProcess_\designIdx$ is generated by 
sampling $\{\predObs_n^{(\mcSampIdx)}\}_{\mcSampIdx=1}^{\numMCSamp} \simiid p(\predObs \given \Par_\designIdx)$
and then computing the log-likelihood estimate 
$\simObsProcess_\designIdx \Def \log \overline{\lik}^{\numMCSamp}(\Par_\designIdx)$
using $\{\predObs_n^{(\mcSampIdx)}\}_{\mcSampIdx=1}^{\numMCSamp}$ \citep{WilkinsonABCGP,VehtariParallelGP}.
Here we have $\targetRange = \simObsSpace = \R$.
Note that we maintain the convention that $\Ndesign$ is the number of design
inputs $\Par_\designIdx$. The number of simulator runs is $\Ndesign\numMCSamp$
due to replicate runs at the same inputs. One can also consider varying the 
number of replicate simulations based on the particular input \citep{Binois_2018}.
\end{example}

\begin{example}[Pseudo-Marginal Noisy Log-Likelihood Target]
Consider the pseudo-marginal setting from \Cref{sec:pseudo-marg}, in which a 
noisy, unbiased likelihood estimator is given by
$\overline{\lik}^{\numMCSamp}(\Par) = 
\numMCSamp^{-1} \sum_{\mcSampIdx=1}^{\numMCSamp} p(\obs \given \latent_\mcSampIdx, \Par)$,
where $\latent_\mcSampIdx \simiid \latentDens(\cdot \given \Par)$.
As in \Cref{ex:llik-noisy-lfi-target}, the need to run the simulator at each new $\Par$
can be addressed by fitting an emulator targeting
$\target(\cdot) \Def \log \lik(\Par)$ \citep{DrovandiPMGP}. However, in this case
the target is the \textit{exact}
log-likelihood, since unbiased likelihood estimates are available
in the pseudo-marginal setting.
A training observation at input $\Par_\designIdx$ is constructed by sampling 
$\{\latent_\designIdx^{(\mcSampIdx)}\}_{\mcSampIdx=1}^{\numMCSamp} \simiid \latentDens(\cdot \given \Par_\designIdx)$ and 
computing the log-likelihood estimate
$\simObsProcess_\designIdx \Def \log \overline{\lik}^{\numMCSamp}(\Par_\designIdx)$
using $\{\latent_\designIdx^{(\mcSampIdx)}\}_{\mcSampIdx=1}^{\numMCSamp}$.
Again we have $\targetRange = \simObsSpace = \R$.
\end{example}

\begin{example}[Conditional Density Target]
Consider the approach to conditional density estimation described in 
\Cref{sec:cond-dens-est}, in which training pairs 
$\{\Par_\designIdx, \predObs_\designIdx\}_{\designIdx=1}^{\Ndesign}$
generated via $\predObs_\designIdx \simiid p(\predObs \given \Par_\designIdx)$
are used to optimize the parameters $\finiteDimPar$ of a flexible conditional
density estimator $\condDensEst(\predObs \given \Par) \approx p(\predObs \given \Par)$.
In this case, the emulator directly targets the conditional density 
$\target(\Par) \Def p(\cdot \given \Par)$, implying that $\targetRange$ is a space
of probability densities over $\obsSpace$. In practice, $\predObs$ is often 
replaced by hand-crafted or learned summary statistics $\summaryStat(\predObs)$,
but we suppress this in the notation for brevity \citep{bayesflow_2020_original}.
The simulator observation process is given by 
$\simObsProcess(\Par) = \predObs \sim \simObsDist(\cdot \given \target(\Par))$, where
$\simObsDist(\cdot \given p(\cdot \given \Par)) = p(\cdot \given \Par)$. That is, the
target $\target(\Par)$ is a probability distribution, and the simulator observation
$\simObsProcess(\Par)$ is a sample from that distribution. Here, $\simObsSpace = \obsSpace$,
while $\targetRange$ is a space of densities over $\obsSpace$. 
This problem is thus structurally distinct from the previous examples, 
in that the training observations $\simObsProcess_\designIdx = \predObs_\designIdx$ 
provide only indirect information about the target $\target$.
\label{ex:cond-dens-target}
\end{example}

The preceding examples describe several modeling setups that yield training data 
$\trainData_{\Ndesign} \Def \{\Par_\designIdx, \simObsProcess_\designIdx\}_{\designIdx=1}^{\Ndesign}$.
We define a \inlinedef{regression-based emulator} as a predictive model $\targetEm$ (e.g., a 
regressor or interpolator) fit to training data of this form.
If $\targetEm$ is deterministic (i.e., $\targetEm(\Par)$ is a point prediction of $\target(\Par)$),
then the induced density approximation $\jointDens(\Par; \targetEm)$ can be fed to standard
inference algorithms such as MCMC, provided that emulator predictions can be computed
relatively cheaply (here let $\jointDens$ be the exact density or a target 
approximation as in the ABC/SL case). 
However, as noted in the introduction, this raises the concern that 
the posterior estimate can be simultaneously biased and overconfident.
Two complementary strategies to mitigate this concern
are to (1) reduce the error by generating more training data and (2) propagate the surrogate uncertainty.
As both are facilitated by emulators equipped with a notion of predictive 
uncertainty, we henceforth focus on \inlinedef{probabilistic emulators}---models that 
provide predictions in the form of probability distributions.

\begin{remark}
The terms ``emulator'' and ``surrogate'' are widely used in the literature in varying contexts,
often synonymously. In this review, we reserve the term \inlinedef{emulator} for the 
predictive model $\targetEm$ that directly approximates the target $\target$. We use
\inlinedef{surrogate} more broadly for any downstream quantity that depends on $\targetEm$.
For example, $\postDens(\Par; \targetEm)$ is the surrogate density for $\postDens(\Par; \target)$
that is induced by the underlying emulator $\targetEm$.
\end{remark}

\subsection{Probabilistic Emulators} \label{sec:prob-emulators}
In general, we allow $\targetEm$ to be any model that provides a predictive distribution, such
that $\targetEm(\Par)$ is a random quantity summarizing uncertainty in the true target
value $\target(\Par)$. To encompass both parametric and nonparametric models (e.g., Gaussian processes),
we adopt the general perspective of $\targetEm$ as a random function.
Conceptually, we think of the randomness in $\targetEm$ as quantifying \inlinedef{epistemic uncertainty}, 
which in principle could be reduced, were it computationally feasible to evaluate 
$\target$ at any given input \citep{epistemicAleatoric}. 

\begin{definition}[Probabilistic Emulator]
Given training data $\trainData_{\Ndesign}$ generated as described in 
\Cref{def:sim-obs}, a \inlinedef{probabilistic emulator} is a random function 
$\targetEm: \parSpace \to \targetRange$ whose distribution $\emDist$ depends
on $\trainData_{\Ndesign}$. For any finite collection 
$\ParBatch = \{\Par_b\}_{b=1}^{B} \subset \parSpace$, 
we write $\emDist(\ParBatch)$ for the
marginal distribution of 
$\targetEm(\ParBatch) = [\targetEm(\Par_1), \dots, \targetEm(\Par_B)]$
and $\emDist(\Par)$ for the marginal at a single point.
\end{definition}

The predictive distribution $\emDist$ represents the uncertainty about the deterministic target
$\target$, not the variability in the simulator observation process $\simObsProcess$.
Any emulator satisfying our definition must therefore separate epistemic uncertainty 
from observation noise. This poses no limitation for standard surrogate
constructions (e.g., a regression model with additive noise). 
In addition to 
$\emDist$, we will find it useful to consider the predictive distribution 
$p(\simObsProcess(\Par) \given \trainData_{\Ndesign})$ over the simulator 
observation process.

\begin{definition}[Predictive Simulator Observation Process]
\label{def:pred-sim-obs-process}
Given the prior simulator observation process 
$\simObsProcess(\Par) \sim \simObsDist(\cdot \given \target(\Par))$
and emulator $\targetEm \sim \emDist$, let 
$\simObsPredProcess(\Par)$ denote the predictive process over
simulator observations, with distribution
\begin{equation}
\simObsPredDist(\cdot \given \Par) \Def 
\int \simObsDist(\cdot \given \targetTraj(\Par)) \emDist(\d\targetTraj).
\end{equation}
The quantity $\simObsPredDist$ represents a convolution of the epistemic
uncertainty $\emDist$ with the simulator noise $\simObsDist$.
\end{definition}

While in general $\targetEm$ and $\simObsPredProcess$ may not even be 
defined over the same space, the following example demonstrates a 
special case where $\targetRange = \simObsSpace$ and the two 
processes are closely related.

\begin{example}
\label{ex:pred-process-comparison}
Suppose the simulator generates observations via the process 
\begin{align}
&\simObsProcess(\Par) = \target(\Par) + \noise(\Par),
&&\noise(\Par) \sim \Gaussian(0, \SLCov(\Par)).
\end{align}
If the emulator produces Gaussian predictions 
$\targetEm(\Par) \sim \Gaussian(\emMean(\Par), \emVar(\Par))$
(independent of $\noise(\Par)$), then 
the predictive simulator observation process is given by
\begin{equation}
\simObsPredProcess(\Par) \sim \Gaussian(\emMean(\Par), \emVar(\Par) + \SLCov(\Par)).
\end{equation}
In this case, $\targetEm$ and $\simObsPredProcess$ have the same expectation, and 
differ only in that the covariance of the latter accounts for the sampling variability
$\noise(\Par)$ of the observation process. This special case arises in 
standard Gaussian process modeling setups (see \Cref{sec:gp}).
\end{example}

The following notation is used throughout the remainder of the paper.
Let $\emE$ denote expectation with respect to $\emDist$.
Supposing for the moment that $\targetRange = \R$, define the pointwise 
predictive mean $\emMean(\Par) \Def \emE[\targetEm(\Par)]$ and variance 
$\emVar(\Par) \Def \Var_{\emDist}[\targetEm(\Par)]$. 
While some emulators provide only pointwise predictions,
many also encode correlational structure across input values. We thus denote 
the covariance 
$\emKer(\Par, \Par^\prime) \Def \Cov_{\emDist}[\targetEm(\Par), \targetEm(\Par^\prime)]$, 
noting that $\emKer(\Par, \Par) = \emVar(\Par)$. For a batch $\ParBatch$ of $\Nbatch$ inputs,
$\emMean(\ParBatch)$ and $\emVar(\ParBatch)$ are the mean vector and covariance matrix
of $\targetEm(\ParBatch)$, respectively. Similarly, given another set
$\ParBatch^\prime$ of $\Nbatch^\prime$ inputs, we write $\emKer(\ParBatch, \ParBatch^\prime)$ to denote
the $\Nbatch \times \Nbatch^\prime$ cross covariance. In the case that $\targetRange = \R^{S}$
for some finite integer $S$, we interpret $\emMean(\Par)$ and $\emVar(\Par)$ as the predictive mean vector and 
covariance matrix over the different output dimensions. To avoid complicating 
notation, we do not extend the vectorized input notation $\emMean(\ParBatch)$ to 
the multi-output setting. We make clarifying remarks when necessary to avoid ambiguity.
We use analogous notation for expectations with respect to
$\simObsPredDist$.

The randomness in the emulator induces randomness in 
downstream quantities such as $\normCst(\targetEm)$, 
$\jointDens(\cdot; \targetEm)$,  and $\postDens(\cdot; \targetEm)$. 
When discussing such induced random variables we make use of the 
shorthands $\normCstEm$, $\jointEm$, and $\postEm$ when explicit 
reference to the underlying emulator is not necessary. For a 
deterministic surrogate, the induced approximation to the posterior
distribution is immediate. In the probabilistic setting, it is not 
immediately clear how to perform approximate inference using 
the random posterior $\postEm$, a topic discussed in depth 
in \Cref{sec:post-approx}.

\subsection{Common Model Classes} \label{sec:model-classes}
In this section, we highlight several popular models that fall under our definition of 
\textit{probabilistic surrogate}.

\subsubsection{Gaussian Processes} \label{sec:gp}
Gaussian processes (GPs) are widely used as surrogate models, with extensive applications in 
response surface modeling for computer experiments 
\citep{design_analysis_computer_experiments,SanterCompExp}, 
black-box optimization \citep{reviewBayesOpt}, reliability 
analysis \citep{contourEstimation,cole2021entropybased}, and parameter 
calibration \citep{KOH,computerModelCalibrationReview}. For in-depth treatments, we refer to
\citet{gramacy2020surrogates,gpML,StuartTeck2}.

A typical GP-based workflow consists of specifying a prior $\targetEm[0] \sim \GP(\emMean[0], \emKer[0])$ 
over the target function $\target$, defined by a prior mean function $\emMean[0](\cdot)$ and 
covariance function (i.e., kernel) $\emKer[0](\cdot, \cdot)$. The 
defining property of a GP is its Gaussian finite-dimensional marginal distributions
$\targetEm[0](\ParBatch) \sim \Gaussian(\emMean[0](\ParBatch), \emKer[0](\ParBatch, \ParBatch))$.
Commonly, the hyperparameters defining the mean and kernel are optimized, though Bayesian 
treatments are also possible \citep{fullyBayesianGPs}. With fixed hyperparameters
and Gaussian observation noise 
$\simObsProcess(\ParBatch) \given \target \sim \Gaussian(\target(\ParBatch), \sigma^2 I)$,
the emulator is constructed by closed-form conditioning
$\targetEm \Def \targetEm[0] \given [\simObsProcess(\design) = \designResponse] \sim \GP(\emMean, \emKer)$,
with the conditional mean and kernel given by
\begin{align}
\emMean(\ParBatch) 
&= \emMean[0](\ParBatch) + \emKer[0](\ParBatch, \design) \kerMat^{-1}[\designResponse - \emMean[0](\design)] \label{eq:kriging-equations} \\
\emKer(\ParBatch, \ParBatch) 
&= \emKer[0](\ParBatch, \ParBatch) - \emKer[0](\ParBatch, \design) \kerMat^{-1} \emKer[0](\design, \ParBatch),
\nonumber
\end{align}
where $\kerMat \Def \emKer[0](\design, \design) + \sigma^2 I$. The predictive observation process is likewise
given by the conditional 
$\simObsPredProcess \Def \simObsProcess \given [\simObsProcess(\design) = \designResponse]$. The result
is again a GP of the form in \Cref{eq:kriging-equations}, with the one modification that the 
predictive covariance $\emKer(\ParBatch, \ParBatch) + \sigma^2 I$ accounts for the 
observation noise. The relationship between $\targetEm$ and $\simObsPredProcess$ follows that given
in \Cref{ex:pred-process-comparison}.

Thus, in their most basic form GP surrogates are Gaussian predictors 
$\targetEm(\ParBatch) \sim \Gaussian(\emMean(\ParBatch), \emVar(\ParBatch))$, 
where $\emVar(\ParBatch) = \emKer(\ParBatch, \ParBatch)$. Extensions of the GP methodology
can produce more flexible emulators with non-Gaussian predictive distributions. 
A non-Gaussian noise model breaks the closed-form conditional 
and requires numerical methods for inference (e.g., MCMC). Other extensions include 
fully Bayesian GPs \citep{fullyBayesianGPs} and deep GPs \citep{deepGPVecchia,deepGPAL},
which typically yield predictions in the form of (infinite) mixtures of Gaussians.

\subsubsection{Polynomials}
Surrogates that consist of linear combinations of polynomial basis functions are 
commonly used in the engineering and applied mathematics literatures. 
Polynomial chaos expansions (PCEs; \citet[Chapter 9]{UQpredCompSci}) are a particular example whereby 
a polynomial basis is used to approximate a random variable $\target(\Par)$ as a function of
a random input $\Par \sim \priorDens$. In particular, the expansion is of the form
\begin{equation}
\targetEm[\dimBasis](\Par) = \sum_{\idxBasis=1}^{\dimBasis} c_{\idxBasis} \basisVec_{\idxBasis}(\Par),
\end{equation}
where $\basisVec_{1}, \dots, \basisVec_{\dimBasis}$ are orthogonal polynomials with respect 
to $\priorDens$. Once constructed, a PCE is often used to 
approximate moments of $\target(\Par)$. More relevant to our context,
the map $\targetEm[\dimBasis](\Par)$ can also be used as a surrogate for $\target(\Par)$.
With fixed coefficients $c_{\idxBasis}$, this map is deterministic and thus PCEs do not fall
within our definition of probabilistic surrogates. We nonetheless highlight them here
due to their popularity, the fact that they can be converted into random surrogates by 
considering Bayesian treatments of the coefficients 
\citep{BayesianPCE1,BayesianPCE2,BurknerSurrogate},
and their use in conjunction with probabilistic surrogates (e.g., as the mean function
of a GP; \citet{PCEGPWind,PCEGP2,SinsbeckNowak}). PCE surrogates have been
employed to accelerate Bayesian inversion in various applications 
\citep{dimRedPolyChaos,BurknerSurrogate,PCEBIP}.

\subsubsection{Neural Networks} \label{sec:neural-networks}
In regimes where the input dimension $\dimPar$ and computational 
budget $\Ndesign$ are large, neural networks are well-suited to 
serve as surrogate models. Bayesian treatment of neural 
network parameters provides probabilistic predictions, 
but in practice significant approximations are required for 
inference \citep{BNNSurvey,BayesOptNN}. Practical strategies
include Laplace approximations \citep{MacKayLaplaceNN}, variational 
inference \citep{VIforNNs}, Monte Carlo dropout
\citep{BayesianDropout}, and partially Bayesian networks
\citep{partialBNN,BayesLastLayer}. 
Another popular method is
\inlinedef{deep ensembles}, which summarize uncertainty 
via an ensemble of neural networks with randomized 
initializations \citep{deepEnsembles,Lueckmann2019}.
Alternative approaches, such as \inlinedef{epinets}, produce
predictive distributions without ensembling by modifying
the architecture of a single neural network
\citep{epistemicNN,BayesOptEpistemicNN}. 
We refer to \citet{UQDLReview,DLUQ,BNNSurvey} for thorough surveys of 
uncertainty quantification in deep learning.

\subsection{Finite vs. Infinite Dimensional Emulators} \label{sec:finite-vs-inf}
Since $\emDist$ is a probability distribution over 
\inlinedef{functions}, $\targetTraj \sim \emDist$ is a trajectory (i.e., sample path)
of $\targetEm$. While we adopt this functional perspective for generality, the 
randomness in $\targetEm$ may stem from a finite set of parameters.
To clarify this distinction, we call 
$\targetEm$ \inlinedef{finite-dimensional} if it can be written as 
$\targetEm(\cdot) = g(\cdot; \finiteDimPar_{\Ndesign})$ 
for some finite-dimensional random vector $\finiteDimPar_{\Ndesign}$ and non-random function $g$.
When this is not possible, then $\targetEm$ is \inlinedef{infinite-dimensional}.
Standard parametric models (linear regression, neural networks) are finite-dimensional,
while GPs (with standard choices of the kernel) are infinite-dimensional. A relevant consequence 
of this difference is that trajectories of finite-dimensional
emulators can be sampled via
\begin{align}
&\targetTraj(\cdot) \Def g(\cdot; \finiteDimPar),
&&\finiteDimPar \sim \law(\finiteDimPar_{\Ndesign}),
\end{align}
where $\law(\cdot)$ denotes the probability distribution of its argument.
The sampled trajectory $\targetTraj(\cdot)$ can now be evaluated at any input value. 
For infinite-dimensional models, 
it is computationally infeasible to sample trajectories, as this would require
storing an infinite-dimensional function in computer memory.
Instead, the surrogate can typically be characterized by the set of distributions 
$\emDist(\ParBatch)$ over finite-dimensional subsets $\ParBatch$.
Infinite-dimensional GPs can be approximated by finite-dimensional models in various 
ways, often relying on a finite set of basis functions 
\citep{pathwiseCond,EfficientSampGPPost}.

\subsection{Practical Considerations for Different Emulator Targets} \label{sec:practical-em-target}
While the choice of emulator target is inherently problem-dependent, one is 
often faced with multiple competing options for a particular problem.  
\Cref{sec:reg-emulators} presented several examples with varying targets; 
we now provide a more detailed discussion of their practical tradeoffs.
To this end, we group different approaches into three broad classifications:
\inlinedef{forward model emulation} (targeting the 
underlying forward map), \inlinedef{log-density emulation} (directly
emulating the log of the likelihood or posterior surface), and 
\inlinedef{conditional density estimation} (learning a flexible parametric
density estimator). The former two categories are also compared in 
\citet{StuartTeck1,GP_PDE_priors,random_fwd_models,RobertsUncProp}.

\begin{figure}[p] 
    \centering
    \includegraphics[width=\textwidth, height=0.80\textheight, keepaspectratio]{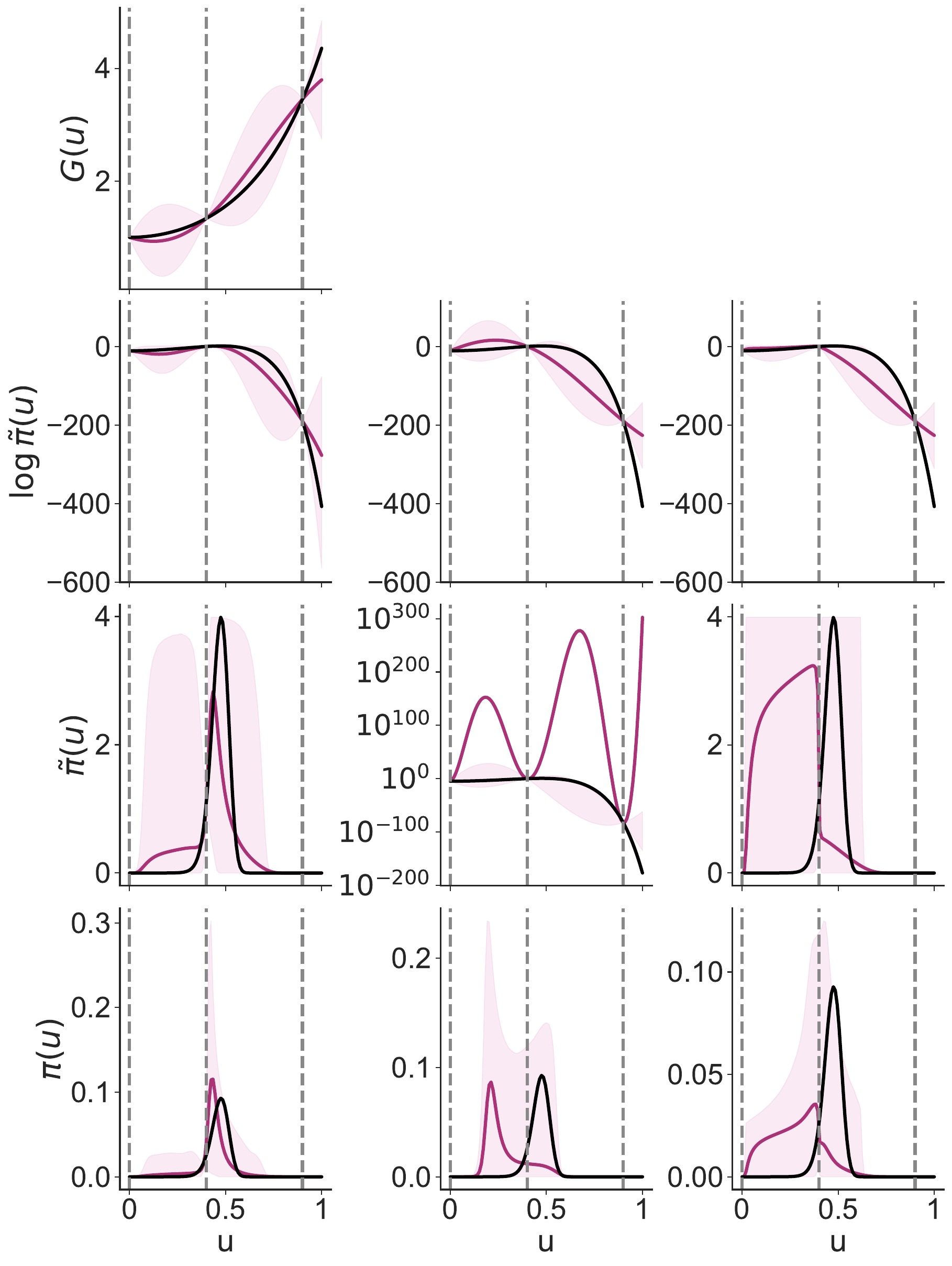}
    \vspace{-15pt}
    \caption{Pushforward distributions induced by GP forward model emulator (left column), 
    GP log-posterior emulator (middle column), and a clipped (upper bounded) GP log-posterior emulator (right column). 
    The plots summarize the pointwise marginal distributions of each quantity; in particular, the means (magenta lines)
    and 95\% credible intervals (shaded regions). The black lines are ground truth (no emulation) and the gray dashed lines 
    indicate the locations of the design inputs used to train the surrogates. The respective rows represent the surrogate-induced 
    distributions over the \emph{(1)} forward model, \emph{(2)} unnormalized log-posterior density, \emph{(3)} unnormalized 
    posterior density, and \emph{(4)} normalized posterior density. The top middle and top right entries are blank because the 
    log-density surrogates do not produce an approximation of the forward model. The magenta lines in the third and fourth 
    rows represent \texorpdfstring{$\jointApproxEUP$}{the unnormalized EUP} and \texorpdfstring{$\postApproxEP$}{the EP}, 
    respectively (see \Cref{sec:post-approx}).}
    \label{fig:em_dist_1d}
\end{figure}

\subsubsection{Forward Model Emulation} \label{sec:fwd}
The high computational cost in Bayesian inverse problems stems primarily from 
the underlying forward model. A natural strategy to alleviate this bottleneck 
is to fit an emulator directly targeting $\fwd$, as in 
\Cref{ex:fwd-target,ex:fwd-noisy-target}. The following example analyzes the
density surrogate $\jointDens(\Par; \targetEm)$ induced by an underlying
forward model emulator $\targetEm$, under Gaussian assumptions.

\begin{example}[Forward Model Emulation, Gaussian Setting]
\label{ex:fwd-em}
Consider the setting of a Bayesian inverse problem with an additive Gaussian noise model
(\Cref{eq:inv_prob_Gaussian}). Assume an emulator $\targetEm$ has been 
fit to approximate the true forward model $\fwd$. This induces an approximation of the 
unnormalized posterior density 
\begin{align}
\jointDens(\Par; \targetEm) 
&= \priorDens(\Par)\Gaussian(\obs \given \targetEm(\Par), \covNoise),
\label{eq:fwd-em-Gaussian}
\end{align}
for each $\Par \in \parSpace$. In general, the distribution of $\jointDens(\Par; \targetEm)$ depends on 
the predictive distribution of $\targetEm$. The additional assumption 
$\targetEm(\Par) \sim \Gaussian(\emMean(\Par), \emVar(\Par))$ facilitates closed-form computation
of the first two moments,
\begin{align*}
\emE\left[\jointDens(\Par; \targetEm) \right] 
&= \priorDens(\Par) \Gaussian(\obs \given \emMean(\Par), \covNoise + \emVar(\Par)) \\
\Var_{\emDist}\left[\jointDens(\Par; \targetEm) \right]
&= \priorDens^2(\Par) \bigg[\frac{\Gaussian\left(\obs \given \emMean(\Par), \frac{1}{2}\covNoise + 
\emVar(\Par)  \right)}{\det(2\pi \covNoise)^{1/2}} - \\
&\qquad \qquad\qquad \frac{\Gaussian\left(\obs \given \emMean(\Par), 
\frac{1}{2}\left[\covNoise + \emVar(\Par)\right]  \right)}{\det(2\pi [\covNoise + \emVar(\Par)])^{1/2}}\bigg],
\end{align*}
a fact that has been leveraged frequently in the literature
\citep{StuartTeck1,GP_PDE_priors,hydrologicalModel,hydrologicalModel2,Surer2023sequential,
VillaniAdaptiveGP,weightedIVAR,idealizedGCM,CES}. The left column of \Cref{fig:em_dist_1d}
visualizes how the uncertainty in a GP forward model emulator propagates to downstream 
quantities within this setting.
\end{example}

In general, the architecture of a forward model emulator is dictated by the 
properties of the forward model itself.
A common challenge in this context is the high dimensionality of the observation space,
particularly for models with complex spatial or temporal structure; examples include 
epidemic modeling \citep{FadikarAgentBased},
engineering design \citep{PODemulation}, ecological forecasting 
\citep{emPostDens,DagonCLM}, and climate modeling \citep{ESM_modeling_2pt0,idealizedGCM}.
Emulators for such complex model outputs typically rely on some form of dimensionality 
reduction. Common approaches include:
\begin{enumerate}
\item \emph{Basis Expansion}: Representing the model output with respect to a small
number of basis vectors, and emulating the scalar basis coefficients
(e.g., via principal components analysis or proper orthogonal decomposition;
\citet{HigdonBasis,FadikarAgentBased,PODemulation}).
\item \emph{Feature Extraction}: Targeting low-dimensional summary statistics
of the model outputs \citep{ESM_modeling_2pt0,idealizedGCM,CLMBayesianCalibration,CLMSurrogates}.
In ABC/SL, this implies emulating the same summary statistics used to form the 
approximate likelihoods \citep{GPSABC}.
\item \emph{Specialized Architectures}: Exploiting the recursive structure of model outputs 
for forward models with dynamical structure
\citep{GP_dynamic_emulation, Bayesian_emulation_dynamic,Liu_West_dynamic_emulation,
dynamic_nonlinear_simulators_GP,FerRNNHMC}.
\end{enumerate}

Other options include emulating each output independently 
\citep{ScalableConstrainedBO,onSiteCalibration,Surer2023sequential,GPSABC}
or adopting a tractable model for cross-output correlations
\citep{MultiTaskGP,BOHighDimOutputs}.

Despite these challenges, there are many practical benefits to the forward model emulation strategy.
Practitioners often wish to perform various tasks using the forward model,
aside from parameter estimation (e.g., sensitivity analysis). A forward model surrogate may be 
re-used across different tasks.
Moreover, in principle $\targetEm$ does not depend on either the observed data or the particular
statistical model (the prior and likelihood). Therefore, the cost of emulator training
is \inlinedef{amortized} over multiple data realizations $\obs$ or changes to the statistical
model. However, in practice the design of the forward model may be inherently intertwined with the 
data, limiting the potential for amortization.

\begin{example}[Limitations of Amortization]
\label{ex:amortization-caveats}
Consider the setting from \Cref{ex:ode}, which considers dynamical models of the form
$\frac{\d}{\d\Time} \state(\Time, \Par) = \odeRHS(\state(\Time, \Par), \Par)$. Suppose
the observed data $\obs$ consists of noisy monthly averages of one of the state variables.
A natural emulation strategy is to target the corresponding monthly averages of the
model outputs. However, this construction ties the emulator to the structure of the data.
If statistical considerations lead one to change the temporal resolution of $\obs$ to 
weekly averages, the emulator will need to be re-structured and trained again.
Even if the data structure is constant, in order to be re-used for new data realizations
$\obs$ the emulator must also account for the fact that underlying forcing data and 
initial conditions will also change. Specialty emulators for dynamic models seek 
to address this challenge \citep{GP_dynamic_emulation}.
\end{example}

\subsubsection{Log-Density Emulation} \label{sec:ldens-em}
Given that most posterior inference algorithms require only an unnormalized density as input, 
an alternative to forward model emulation is to directly target the log-likelihood or 
the unnormalized log-posterior, approaches we collectively refer to as \inlinedef{log-density emulation}.
While one could also consider modeling the likelihood or posterior surface directly,
emulating on the log scale is generally preferred to improve numerical stability, enforce
non-negativity of the resulting density approximation, and yield a smoother 
target surface \citep{wang2018adaptive,WilkinsonABCGP}. The following example analyzes
the density surrogate $\jointDens(\Par; \targetEm)$ induced by a log-likelihood emulator
$\targetEm$, under Gaussian assumptions.

\begin{example}[Log-Likelihood Emulation, Gaussian Setting]
\label{ex:ldens-em}
In the case that the emulator is fit to the log-likelihood, the induced unnormalized posterior approximation
is of the form 
\begin{align}
\jointDens(\Par; \targetEm) &= \priorDens(\Par) \Exp{\targetEm(\Par)}.
\end{align}
Under the assumption $\targetEm(\Par) \sim \Gaussian(\emMean(\Par), \emVar(\Par))$, this
quantity is log-normally distributed, with moments
\begin{align}
\emE\left[\jointDens(\Par; \targetEm) \right] 
&= \priorDens(\Par)\Exp{\emMean(\Par) + \frac{1}{2} \emVar(\Par)} \label{eq:ldens-gp-moments} \\
\Var_{\emDist}\left[\jointDens(\Par; \targetEm) \right]
&= \priorDens^2(\Par) \left[\Exp{\emVar(\Par)} - 1 \right] \Exp{2\emMean(\Par) + \emVar(\Par)}, \nonumber
\end{align}
a fact noted in \citet{DrovandiPMGP,VehtariParallelGP,StuartTeck2,GP_PDE_priors}. 
Gaussian process log-density emulators have been used both for deterministic 
inverse problems 
\citep{StuartTeck2,GP_PDE_priors,quantileApprox,ActiveLearningMCMC,JosephMinEnergy,
gp_surrogates_random_exploration,ActiveLearningMCMC,FerEmulation,FATES_CES}, as well 
as in the SBI setting 
\citep{WilkinsonABCGP,DrovandiPMGP,trainDynamics,OakleyllikEm,VehtariParallelGP,gpEmMCMC}.

The predictive distribution of a GP log-density emulator is plotted in \Cref{fig:em_dist_1d}
(second row, second column). Inspecting the induced distribution $\jointEm$ (third row, second column),
we see that the predictive mean $\emE\left[\jointEm(\Par) \right]$ is highly sensitive
to the GP variance, providing a poor summary of $\jointEm$. The third column shows how this
sensitivity is reduced when an upper bound is enforced on the emulator predictive distribution.
\end{example}

The primary appeal of log-density emulation is that it collapses a potentially high-dimensional 
observation space into a single scalar output \citep{scalarization,trainDynamics,CompositeBO}. 
However, this simplification introduces several challenges. Log-densities often exhibit a large dynamic 
range and nonstationary behavior, which can be difficult for standard surrogates to capture 
\citep{wang2018adaptive,Surer2023sequential,WilkinsonABCGP}. Furthermore, because additive errors in the log-domain 
become multiplicative errors when exponentiated, the resulting posterior approximations are highly sensitive 
to emulator misspecification \citep{RobertsUncProp}. A further complication arises when likelihood 
parameters (e.g., the noise covariance $\covNoise$ in \Cref{eq:inv_prob_Gaussian}) are unknown. 
While one can expand the emulator's input space to include these parameters \citep{llikRBF, emPostDens}, this increases 
the input dimensionality and can further complicate the target surface. Alternatively, some likelihoods 
admit a sufficient statistic that can be emulated independently of the likelihood parameters \citep{FerEmulation}.
This is similar in spirit to emulating a measure of discrepancy between the observed 
and simulated data in ABC applications \citep{BOforLFI,EfficientAcqABC}.
Finally, we note that log-density emulators are fit with respect to a particular data realization $\obs$
and statistical modeling setup. While forward model emulators offer some potential for amortizing 
the training cost over changes in $\obs$ or the model
(subject to caveats; see \Cref{ex:amortization-caveats}), log-density emulators are 
by design dependent on fixed values of these quantities.

\paragraph{Log-Likelihood vs. Log-Posterior Emulation.}
Since the prior density is typically cheap to evaluate, it appears natural to approximate
only the log-likelihood. However, several studies opt to emulate the unnormalized log-posterior 
directly \citep{emPostDens, Kandasamy2017, gp_surrogates_random_exploration}. 
While the choice may seem numerically inconsequential (one can be converted to the other 
via a deterministic shift by the log-prior), there are tradeoffs to consider. 
Emulating the log-posterior allows one to incorporate within the surrogate the known constraint 
that the posterior must decay in the tails. In contrast, the log-likelihood's tail behavior may be 
difficult to characterize a priori, presenting an additional challenge for surrogate model specification.
Conversely, constraining tail behavior in certain surrogate models (e.g., GPs) can be difficult, 
and the failure to do so in log-posterior emulation can lead to pathological results \citep{gpEmMCMC}. 
Another downside of log-posterior emulation is that the surrogate depends on the prior; 
any update to the prior model may necessitate re-training (depending on the probabilistic structure
of the emulator). A final consideration is whether 
the log-posterior or log-likelihood surface is easier to emulate. In general, this will depend on 
the relative strength of the prior and the likelihood, as well as their particular functional forms.
In settings where the likelihood dominates the prior, the practical difference between these two 
target surfaces diminishes. 

\subsubsection{Conditional Density Emulation}
We return to the setting of \Cref{sec:sbi} and \Cref{ex:cond-dens-target}, 
in which a neural network
$\condDensEst(\predObs \given \Par)$ is trained to learn the 
mapping from $\Par$ to the corresponding conditional density 
$p(\cdot \given \Par)$. The emulator targets the 
probability density itself, not its value at a particular point.
The training responses are sampled
from the conditionals 
$\simObsProcess_\designIdx = \predObs_\designIdx \sim p(\predObs \given \Par_\designIdx)$
at a set of design points $\{\Par_\designIdx\}$. As opposed to the preceding emulator
classes, the responses are not simply noisy observations of the target.
Nonetheless, this approach fits naturally within our framework for 
uncertainty propagation and active learning. Once an estimate 
$\finiteDimPar_\Ndesign$ for the neural network parameters is produced,
the challenge shifts to performing inference using the surrogate
$\jointEm(\Par) = \priorDens(\Par) \condDensEst[\finiteDimPar_\Ndesign](\obs \given \Par)$.
In practice, $\finiteDimPar_\Ndesign$ is usually a point estimate and thus 
emulator uncertainty is not quantified. However, various techniques 
can be utilized to obtain a predictive 
distribution (see \Cref{sec:neural-networks}), and thus propagate emulator
uncertainty during inference \citep{Lueckmann2019}. 

As with the other emulator
classes, the selection of the input training points $\Par_\designIdx$ is a 
design decision that will control the approximation accuracy in different
regions of $\parSpace$ \citep{SNLE,MurrayNLEvsNPE,Lueckmann2019}. Design
considerations are also influenced by whether the emulator is going
to be applied to many different data realizations $\obs$ or a single 
fixed observation. The former setting typically requires much
larger upfront simulation cost, but this cost is amortized over subsequent
datasets \citep{bayesflow_2020_original}.

We note that the flexibility in the input design is not shared by 
the closely related approach of \textit{neural posterior estimation}, 
which directly learns the posterior from samples \citep{murrayNPE,bayesflow_2020_original}. 
This approach does not fall within our framework as it emulates a target
as a function of $\predObs$ (rather than $\Par$) and does not offer 
as much flexibility in the input design (samples $\Par_\designIdx$ must
come from the prior, or otherwise be appropriately re-weighted).

Neural estimation methods often require generating a large amount of training data,
making them inapplicable in the case of expensive simulators. Addressing
this challenge is an active area of research \citep{costAwareSBI,multifidelitySBI}.
One approach to enable neural estimation with limited simulator budgets
is to incorporate a second layer of emulation, for example by using an underlying
forward model emulator to generate training data for a conditional
density estimator \citep{BurknerAmortizedSurrogate}.

\subsection{Model Checking and Diagnostics} \label{sec:model-checking}

Proper model checking and validation is an integral step in any statistical analysis.
In addition to standard Bayesian model checking for the underlying statistical 
model \citep{gelmanBDA,VisBayesianWorkflow}, the emulator adds another component
which should be validated. The modular framework provides the opportunity to
evaluate the emulator predictive performance (with respect to $\target$)
in isolation, as is standard for any predictive model.
The particular evaluation metrics will depend on the 
model class, but standard techniques like cross-validation can be applied in 
conjunction with scoring rules for evaluating probabilistic predictions \citep{scoringRules}.
Diagnostics for GP emulators are discussed in \citet{GPEmulatorDiagnostics}.

In addition to evaluating $\targetEm$ in isolation, the emulator fit should be
validated with respect to its use in downstream posterior estimation; e.g., 
the quality of the induced approximations $\likEm$, $\jointEm$, and $\postEm$
can be assessed. \citet{BurknerSurrogate} adapt Bayesian simulation-based
calibration checks to jointly validate the underlying statistical model,
the emulator, and the posterior inference algorithm. One consideration that
should be kept in mind is that $\postEm(\cdot) = \postDens(\cdot; \targetEm)$ 
is only well-defined if each trajectory $\jointDens(\cdot; \targetTraj)$
is integrable, so that $\normCst(\targetTraj)$ exists. This may not hold, 
for example, for a stationary GP over an unbounded domain $\parSpace$.
Ideally, the emulator would encode the inductive bias that the posterior 
tails must decay. However, this is challenging in general and common 
practical solutions simply truncate the prior to avoid difficulties
associated with tail behavior 
\citep{gp_surrogates_random_exploration,FerEmulation,RobertsUncProp}.
More careful theoretical treatment related to the existence of 
$\normCst(\targetTraj)$ for GP emulators is detailed in 
\citet{StuartTeck1,random_fwd_models}.

\section{Surrogate-Based Posterior Approximation} \label{sec:post-approx}
The ultimate objective of Bayesian inference is to characterize the posterior distribution 
$\postDens$, which is then used for downstream decision-making. Consequently, the 
central challenge in surrogate-based Bayesian inference is not merely fitting the emulator, but 
effectively utilizing it to approximate $\postDens$ \citep{StuartTeck1,SinsbeckNowak}. 
Any discrepancy between the emulator and the target map introduces bias into the posterior estimate. 
While iterative refinement can asymptotically eliminate this bias (\Cref{sec:active-learning}), 
computational budgets often preclude this luxury. To mitigate bias and properly calibrate uncertainty, 
the surrogate's predictive uncertainty must be acknowledged within the posterior approximation
\citep{BilionisBayesSurrogates,RobertsUncProp}. However, as noted by \citet{BurknerSurrogate}, 
there is no single ``correct'' mechanism for propagating 
surrogate uncertainty. In this section, we survey the conceptual frameworks for deriving these 
approximations and detail specific estimators proposed in the literature.

\begin{remark}
Throughout the remainder of the paper, we let $\postDens$ denote the target distribution for 
Bayesian inference. This is typically the exact posterior, but could also be the approximate
targets $\approxPost$ utilized in the ABC and SL settings (\Cref{sec:classical-lfi}).
Our aim is to focus on the uncertainty in the surrogate $\postDens(\Par; \targetEm)$ induced
by the emulator, not to analyze additional layers of approximation considered in methods
like ABC/SL.
\end{remark}

\subsection{The Plug-In Mean}
We begin by establishing a baseline: ignoring surrogate uncertainty entirely. The \inlinedef{plug-in mean} 
approximation replaces the random surrogate with its deterministic predictive mean $\emMean$, yielding:
\begin{equation}
\postApproxMean(\Par) \Def \frac{\jointDens(\Par; \emMean)}{\normCst(\emMean)}.
\label{eq:mean-approx}
\end{equation} 
While computationally straightforward and widely applied
\citep{VehtariParallelGP,trainDynamics,emPostDens,BurknerSurrogate,CLMBayesianCalibration,Lueckmann2019,BilionisBayesSurrogates}, 
this approach is only justified if the emulator is highly accurate, an assumption rarely satisfied in complex problems. 
In general, disregarding the surrogate uncertainty leads to biased posterior 
approximations with miscalibrated uncertainties 
\citep{BurknerSurrogate,BilionisBayesSurrogates,StuartTeck1,RobertsUncProp}.

\subsection{Frameworks for Constructing Posterior Estimators}
To move beyond the plug-in baseline, we outline three conceptual frameworks that
provide a rigorous foundation for the propagation of surrogate uncertainty.

\subsubsection{Decision Theory for Normalized Density} \label{sec:random-measure}
Given that the emulator $\targetEm$ is a random function, the induced posterior $\postEm$ is a 
random probability measure, provided that $\targetEm$ is constructed such that this measure
is well-defined. We generally treat $\postEm$ as a random density to avoid measure-theoretic technicalities.
One principled approach is to construct a deterministic approximation that best summarizes 
the uncertainty encoded in $\postEm$. 
\citet{RobertsUncProp} formalize this via the variational objective
\begin{equation}
\qDensOpt \Def \argmin_{\qDens \in \qSpace} \emE[\loss(\qDens, \postEm)],
\label{eq:random-measure-variational}
\end{equation}
for a loss function $\loss$ and space of densities $\qSpace$.
Estimators derived from this perspective naturally account for the coupling between the unnormalized density and the 
normalizing constant. However, the global dependence on $\normCstEm$ (a function of the 
entire random function $\targetEm$) often renders inference challenging for approximations of this form.
The random measure perspective is introduced in \citet{StuartTeck1} and adopted in 
\citet{RobertsUncProp,StuartTeck2,random_fwd_models,TeckHyperpar}.

\subsubsection{Decision Theory for Unnormalized Density} \label{sec:decision-theory}
To circumvent the difficulties stemming from the random normalizing constant, an alternative approach
is to estimate the \textit{unnormalized} density $\jointEm$ directly. When only the likelihood depends
on the surrogate, this is often equivalent to producing an approximate likelihood. 
Working with the unnormalized density is computationally 
attractive as $\jointEm(\Par)$ depends only on the marginal distribution of $\targetEm(\Par)$. Adopting a 
decision-theoretic viewpoint \citep{SinsbeckNowak}, we seek the estimator $\qFuncOpt$ that minimizes the 
expected loss
\begin{equation}
\qFuncOpt \Def \argmin_{\qFunc \in \qFuncSpace} \emE[\loss(\qFunc, \jointEm)],
\label{eq:decision-theory}
\end{equation}
over a set of candidate functions $\qFuncSpace$. The final approximate posterior is obtained by normalizing 
$\qFuncOpt$ post-hoc. While computationally convenient, these \inlinedef{pointwise estimators} ignore the 
correlation structure of the surrogate \citep{RobertsUncProp}, as well as the probabilistic coupling between
$\normCstEm$ and $\jointEm(\Par)$. This decision theoretic approach was originally proposed in 
\citet{SinsbeckNowak} and subsequently utilized in 
\citet{EfficientAcqABC,VehtariParallelGP,gpEmMCMC,StuartTeck2}.

\subsubsection{Uncertainty Propagation via Computation} \label{sec:unc-prop-computation}
Finally, rather than defining an explicit target density, the posterior approximation can be defined 
\textit{implicitly} as the output of a randomized algorithm. \citet{gpEmMCMC} adopt this perspective,
seeking to characterize how emulator uncertainty propagates through an MCMC algorithm.
In this case, posterior inference is viewed as a direct forward propagation of surrogate
uncertainty through the sampling machinery.

This viewpoint naturally leads to the modification of standard inference algorithms via the 
injection of additional noise; e.g., sampling emulator realizations at each iteration of an 
MCMC or importance sampling scheme. A comprehensive survey of such algorithms is presented 
in \citet{noisyMCSurvey}. Noisy MCMC schemes are studied in 
\citet{noisyMCMC,stabilityNoisyMH,pseudoMarginalMCMC}.

\subsection{Concrete Posterior Estimators}
We now examine specific estimators derived from these frameworks.

\subsubsection{The Expected Posterior (EP)} \label{sec:ep}
Under the random measure framework (\Cref{sec:random-measure}), minimizing the expected 
Kullback-Leibler (KL) divergence yields the mean of the random density, 
or the \inlinedef{expected posterior} (EP; \citet{BurknerSurrogate}):
\begin{align}
\postApproxEP(\Par) \Def \emE[\postEm(\Par)] = \int \postDens(\Par; \targetTraj) \emDist(\d\targetTraj).
\label{eq:ep}
\end{align}
\citet{RobertsUncProp} derive the EP from this viewpoint and argue for its use as a principled default 
posterior approximation in many applications. From the perspective of modular Bayesian inference, 
the EP can also be viewed as a \textit{cut distribution}, provided that $\targetEm$ is fit
using Bayesian methods \citep{PlummerCut,BurknerSurrogate}. In particular, supposing the 
emulator predictive distribution takes the form of a Bayesian posterior
$\emDist(\targetTraj) \propto \emDist[0](\targetTraj)p(\trainData_\Ndesign \given \targetTraj)$,
then the EP arises as a KL-optimal approximation to the joint model
\begin{equation}
p(\Par, \targetTraj, \predObs, \trainData_\Ndesign)
= \priorDens(\Par)\emDist[0](\targetTraj)p(\predObs \given \Par, \targetTraj)p(\trainData_\Ndesign \given \targetTraj),
\label{eq:joint-model}
\end{equation}
subject to the constraint that the $\targetTraj$-marginal equals $\emDist$
(and is thus not informed by $\predObs$; \citet{RobertsUncProp}).

Functionally, the EP is a mixture distribution averaging the 
posteriors induced by all possible surrogate trajectories. This interpretation suggests a nested sampling 
scheme (\Cref{alg:ep}) to marginalize over the trajectories $\targetTraj \sim \emDist$.

\begin{algorithm}[H]
    \caption{Direct sampling from $\postApproxEP$}
    \label{alg:ep}
    \begin{algorithmic}[1]
    \Function{sampleEP}{$\postEm, \NSample, M$}     
        \For{$\sampleIndex \gets 1, \dots, \NSample$} \Comment{Parallelizable}
        		\State $\targetTraj^{(\sampleIndex)} \sim \emDist$ \Comment{Sample emulator trajectory}
		\State $\Par^{(\sampleIndex, 1)}, \dots, \Par^{(\sampleIndex, M)} \sim \postDens(\cdot; \targetTraj^{(\sampleIndex)})$ \Comment{Sample posterior given trajectory}
	\EndFor
	\State \Return $\{\Par^{(\sampleIndex, m)}\}_{1 \leq \sampleIndex \leq \NSample, \ 1 \leq m \leq M}$
	\EndFunction
    \end{algorithmic}
\end{algorithm}

A practical implementation of this algorithm typically employs Metropolis-within-Monte-Carlo (MwMC), in which 
a separate MCMC scheme is run for each iteration of the loop (line 2) in order to produce approximate samples $\Par^{(\sampleIndex, m)}$
from each posterior trajectory $\postDens(\cdot; \targetTraj^{(\sampleIndex)})$ \citep{garegnani2021NoisyMCMC, BurknerSurrogate}.
However, sampling full trajectories (line 3) is non-trivial for infinite-dimensional models like GPs, a limitation that 
likely hindered early adoption of the EP \citep{StuartTeck1,SinsbeckNowak,VehtariParallelGP}. Recent advances utilize 
approximate inference schemes tailored to GP emulators to overcome this challenge \citep{RobertsUncProp,trainDynamics}.

\subsubsection{The Expected Unnormalized Posterior (EUP)} \label{sec:eup}
To avoid integrating over the normalizing constant, as required by the EP, one can average the numerator 
and denominator independently. This yields the \inlinedef{expected unnormalized posterior (EUP)}
\footnote{
\citet{StuartTeck1,StuartTeck2,GP_PDE_priors} refer to this as the \textit{marginal} approximation,
while \citet{BurknerSurrogate} instead use the term \textit{expected likelihood}. We prefer 
\textit{expected unnormalized posterior}, as used in \citet{RobertsUncProp}, since it also encompasses
log-posterior emulators.
}
\begin{align}
\postApproxEUP(\Par) \Def 
\frac{\emE[\jointEm(\Par)]}{\emE[\normCstEm]}
= \int \postDens(\Par; \targetTraj) \emDist(\d\targetTraj \given \obs),
\label{eq:eup}
\end{align}
where we have defined $\emDist(\d\targetTraj \given \obs) \propto \normCst(\targetTraj) \emDist(\d\targetTraj)$.
Analogous to the EP, the EUP is also a mixture of posterior trajectories, but the mixing distribution
re-weights $\emDist$ by $\normCst(\targetTraj)$ (which depends on $\obs$). 
Thus, the EUP up-weights trajectories $\postDens(\Par; \targetTraj)$ with high ``evidence'' 
$\normCst(\targetTraj)$ \citep{RobertsUncProp}. The EUP also corresponds to the $L^2$-optimal estimator under the unnormalized 
decision-theoretic framework from \Cref{sec:decision-theory} \citep{SinsbeckNowak}. 
Because the expectation $\emE[\jointEm(\Par)]$ 
can often be computed in closed form (see examples below) or unbiasedly estimated, the EUP enables 
the use of standard MCMC or pseudo-marginal algorithms 
\citep{pseudoMarginalMCMC,garegnani2021NoisyMCMC,DrovandiPMGP}.

\begin{table}[t] 
{\centering
\begin{tabular}{lccc}
\toprule
& \textbf{Plug-In Mean} & \textbf{EUP} & \textbf{EP} \\
\midrule \addlinespace[1.5ex]
\textbf{Ratio Estimator} & 
$\displaystyle \frac{\jointDens(\Par; \emE[\targetEm])}{\normCst(\emE[\targetEm])}$ &
$\displaystyle \frac{\emE[\jointDens(\Par; \targetEm)]}{\emE[\normCst(\targetEm)]}$ &
$\displaystyle \emE\left[\frac{\jointDens(\Par; \targetEm)}{\normCst(\targetEm)}\right]$ \\ [4ex]

\textbf{Mixing Dist.} & 
$\delta_{\emMean}(\d\targetTraj)$ & 
$\emDist(\d\targetTraj \given \obs)$ & 
$\emDist(\d\targetTraj)$ \\
\bottomrule
\end{tabular}\par}
\vspace{0.2cm}
\caption{\parbox[t]{\textwidth}{\raggedright
    Different perspectives on normalized density approximations defined by various forms of averaging.
    \emph{(Top)} Interpretations as ratio estimators. The EUP is a ratio of expectations, while the 
    EP is an expectation of a ratio.  Only the EP considers the coupling between
    the numerator and denominator.
    \emph{(Bottom)} Interpretations as mixture distributions of the form 
    $\int \postDens(\Par; \targetTraj) \omega(\d\targetTraj)$; i.e., a weighted mixture over 
    surrogate-induced posterior trajectories. The bottom row gives the form of the 
    mixing distribution $\omega(\d\targetTraj)$ for each approximation.
    We write $\emDist(\d\targetTraj \given \obs)$ to denote the distribution
    proportional to $\normCst(\targetTraj) \emDist(\d\targetTraj)$, and $\delta_{\emMean}$
    for the Dirac measure centered at $\emMean$.
}}
\label{tab:post-approx-comparison}
\end{table}

The EUP was originally proposed in \citet{BilionisBayesSurrogates}, in which it is derived as 
the marginal $p(\Par \given \obs, \trainData_{\Ndesign})$ under the joint model in
\Cref{eq:joint-model}. This implies that the EUP is not ``modular'' in the cut Bayesian
inference sense---the feedback from $\obs$ to $\targetTraj$ is introduced via the 
uncertainty propagation mechanism.
The hierarchical perspective extends to the case where $\emDist$
is not given by a Bayesian posterior; in particular, the EUP is the marginal 
$p(\Par \given \obs)$ under the model 
\begin{equation}
\targetTraj  \sim \emDist, \qquad 
\Par \sim \priorDens, \qquad
\predObs \given \targetTraj, \Par \sim p(\predObs \given \Par, \targetTraj).
\label{eq:eup-prob-model}
\end{equation}
Similar hierarchical modeling perspectives are considered in 
\citet{SinsbeckNowak,StuartTeck1,StuartTeck2,RobertsUncProp}.

Bounds on the Hellinger distance between the EUP and ground truth posterior under 
GP emulators are proved in \citet{StuartTeck1,StuartTeck2,random_fwd_models}.
\citet{RobertsUncProp} study the EUP as an approximation
to the EP, demonstrating that the two distributions can deviate when $\postEm(\Par)$ and
$\normCst(\Par)$ are highly correlated at particular ``influential'' values of $\Par$,
which is more likely to occur when the pointwise marginals $\jointEm(\Par)$ are heavy-tailed.
\Cref{tab:post-approx-comparison} compares the plug-in mean, EUP,
and EP estimators as (i) ratio estimators, and (ii) mixture distributions aggregating surrogate-induced
posterior trajectories.

\begin{example}[EUP, Gaussian Forward Model Emulator]
\label{ex:eup-fwd-em}
Consider the setting from \Cref{ex:fwd-em}, where
$\jointEm(\Par) = \priorDens(\Par) \Gaussian(\obs \given \targetEm(\Par), \covNoise)$
and $\targetEm(\Par) \sim \Gaussian(\emMean(\Par), \emVar(\Par))$.
Under these assumptions, the EUP is given by
\begin{align}
\postApproxEUP(\Par)
&\propto \emE\left[\priorDens(\Par)\Gaussian(\obs \given \targetEm(\Par), \covNoise) \right] \nonumber \\
&= \priorDens(\Par) \Gaussian(\obs \given \emMean(\Par), \covNoise + \emVar(\Par)). \label{eq:eup-gaussian}
\end{align}
Thus, in this setting the EUP admits a natural data space interpretation. 
The original inverse problem $\predObs = \fwd(\Par) + \noise$ is replaced by
$\predObs = \emMean(\Par) + \zeta_\Ndesign(\Par) + \noise$, with the parameter-dependent
noise term $\zeta_\Ndesign(\Par) \sim \Gaussian(0, \emVar(\Par))$ accounting for the 
uncertainty in the forward model \citep{CES,StuartTeck1}. Observe that
$\jointApproxEUP(\Par) \to \priorDens(\Par)$ as $\emVar(\Par) \to \infty$,
implying that ignorance about the true forward model results in 
prior reversion \citep{RobertsUncProp}. The particular EUP estimator in \Cref{eq:eup-gaussian}
is utilized in \citet{Surer2023sequential,weightedIVAR,StuartTeck2,GP_PDE_priors,CES,
idealizedGCM,VillaniAdaptiveGP,hydrologicalModel,hydrologicalModel2}.
\end{example}

\begin{example}[EUP, Gaussian Log-Likelihood Emulator]
\label{ex:eup-ldens-em}
Recall the setup from \Cref{ex:ldens-em}, where $\jointEm(\Par) = \priorDens(\Par)\Exp{\targetEm(\Par)}$
and $\targetEm(\Par) \sim \Gaussian(\emMean(\Par), \emVar(\Par))$.
Under these assumptions, the EUP is given by
\begin{align}
\postApproxEUP(\Par)
&\propto \emE\left[\priorDens(\Par)\Exp{\targetEm(\Par)} \right] \nonumber \\
&= \priorDens(\Par) \Exp{\emMean(\Par) + \frac{1}{2} \emVar(\Par)} \label{eq:eup-ldens-gauss} \\
&= \jointApproxMean(\Par) \Exp{\frac{1}{2} \emVar(\Par)}. \nonumber
\end{align}
The final expression demonstrates that, in this setting, the EUP inflates the plug-in mean estimator
where uncertainty is large \citep{StuartTeck1,StuartTeck2,VehtariParallelGP}. 
The estimator scales very quickly as surrogate uncertainty increases, with 
$\jointApproxEUP(\Par) \to \infty$ as $\emVar(\Par) \to \infty$.
Thus, unlike the forward model analog in \Cref{ex:eup-fwd-em}, the EUP does not exhibit prior
reversion under surrogate ignorance \citep{RobertsUncProp}.
Due to the exponential dependence in \Cref{eq:eup-ldens-gauss}, this distribution can be highly 
sensitive to small changes in $\emMean(\Par)$ and $\emVar(\Par)$, making it particularly susceptible to 
extreme concentration in small regions with large uncertainty 
\citep{RobertsUncProp,VehtariParallelGP,DrovandiPMGP}.
This is illustrated in the second column of \Cref{fig:post_norm_approx_1d}; while the GP 
fit (first row) appears reasonable, the induced EUP is effectively a point mass centered at the 
most uncertain location, a pathology that the EP avoids. The third column illustrates that the
EUP is stabilized by enforcing an upper bound on the GP predictive distribution.
\end{example}

\begin{figure}
    \centering
    \includegraphics[width=\textwidth, height=0.98\textheight, keepaspectratio]{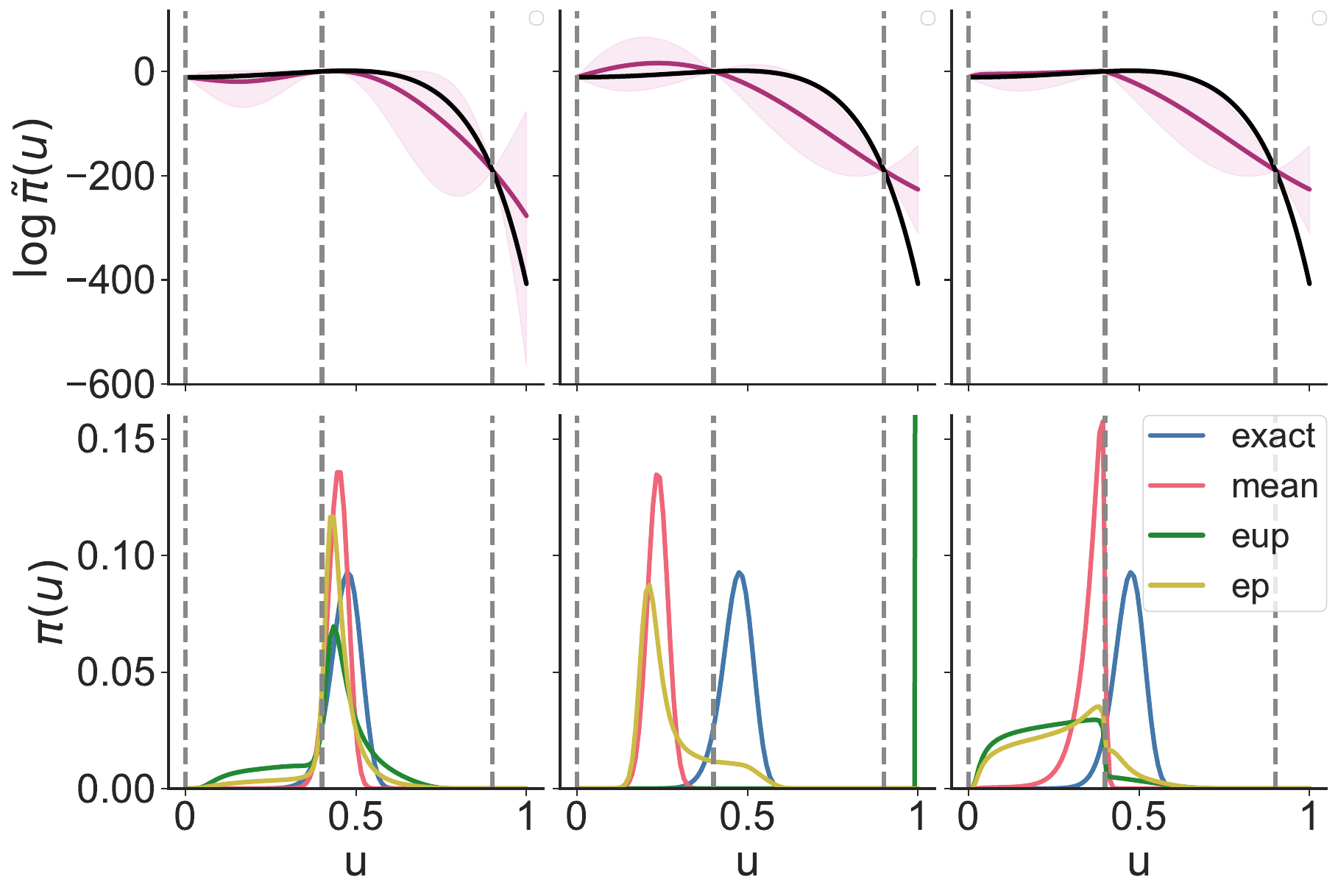}

    \caption{A continuation of the example in \Cref{fig:em_dist_1d}. From left to right, the columns correspond
    	to the GP forward model emulator, GP log-posterior emulator, and clipped (upper bounded) GP log-posterior emulator.
	The top row summarizes the pointwise marginal distributions of 
	\texorpdfstring{$\log \jointEm(\Par)$}{the log unnormalized posterior approximation}, showing the 
	mean (magenta line) and 90\% intervals against the true log-posterior (black). The bottom
	row presents different normalized posterior approximations relative to the true posterior (blue line). The
	vertical dashed lines indicate the locations of the design points.}
\label{fig:post_norm_approx_1d}
\end{figure}

\subsubsection{Other Pointwise Estimators}
The instability of the EUP in log-density emulation (\Cref{ex:eup-ldens-em}) has motivated alternative pointwise estimators.
\paragraph{Marginal Quantiles.} To improve robustness, \citet{VehtariParallelGP} propose using pointwise quantiles of 
    $\jointEm(\Par)$ rather than the mean. As detailed in \Cref{ex:quantile-pw-ldens-em}, setting the quantile probability 
    $\quantileProb > 0.5$ inflates the posterior in regions with higher uncertainty, but scales more moderately relative to the EUP.
\paragraph{Marginal Modes.} \citet{gpEmMCMC} utilize the pointwise mode of $\jointEm(\Par)$. This effectively penalizes 
    regions with high surrogate variance, encouraging the posterior to decay in these regions. While this can have the desirable
    effect of promoting tail decay in the posterior estimate, it neglects uncertain regions that may in reality contain posterior mass.
\paragraph{Expected Log-Likelihood.} In the forward model emulation setting, \citet{BurknerSurrogate} consider a
    posterior estimate proportional to \\ $\priorDens(\Par) \Exp{\emE[\log \lik(\Par; \targetEm)]}$ and draw 
    a connection with power-scaled likelihoods. However, the authors ultimately recommend the EP and EUP as preferred 
    alternatives to this method.

\begin{example}[Quantile Estimator, Gaussian Log-Likelihood Emulator]
\label{ex:quantile-pw-ldens-em}
We return to the setting from \Cref{ex:ldens-em} with a Gaussian log-likelihood emulator. 
Let $\emQ(\cdot)$ denote the $\quantileProb \in (0, 1)$ quantile of its argument with respect to $\emDist$. 
A quantile-based estimate in this setting takes the form
\begin{align*}
\postApproxQuantile(\Par) \propto
\priorDens(\Par) \emQ(\Exp{\targetEm(\Par)})
&= \priorDens(\Par) \Exp{\emMean(\Par) + \GaussianQuantile(\quantileProb) \emSD(\Par)} \\
&= \jointApproxMean(\Par) \Exp{\GaussianQuantile(\quantileProb) \emSD(\Par)},
\end{align*}
where $\GaussianQuantile(\quantileProb)$ is the $\quantileProb$-quantile of $\Gaussian(0,1)$. 
This expression is of the same form 
as the EUP in \Cref{ex:eup-ldens-em}, but the uncertainty inflation term scales more slowly as a function
of $\emSD(\Par)$. The special case $\quantileProb = 1/2$ (i.e., the median) 
reduces to the plug-mean approximation, while values $\quantileProb > 1/2$ imply 
the density will be inflated in regions of higher surrogate uncertainty. The quantile approximation
arises from the unnormalized decision theoretic viewpoint by considering losses 
of the $L^1$ variety \citep{VehtariParallelGP,gpEmMCMC}. 
Quantile-based estimators are also utilized in \citet{quantileApprox,FATES_CES}.
\end{example}

\subsection{Comparison and Practical Recommendations}
In general, there is no one-size-fits-all uncertainty propagation method; the best choice in 
a particular instance depends on the goals of the practitioner and the particular problem
at hand. Following \citet{RobertsUncProp}, we recommend the EP as a reasonable default for
problems where uncertainty calibration is of central importance or where the likelihood 
surrogate $\lik(\cdot; \targetEm)$ is misspecified with respect to $\lik(\cdot)$.
The EP enforces a clear separation between learning $\target$ and learning $\Par$, which
aligns with the view that the emulator is a computational artifact, rather than a true 
latent parameter in $\obs$'s data generating process. This separation is especially 
desirable when $\lik(\cdot; \targetEm)$ is misspecified, following from the fact that 
the EP can be interpreted as a cut posterior, which is specifically designed to mitigate
problems of misspecification \citep{PlummerCut,modularization}. The EP simply mixes 
together each component distribution $\postDens(\cdot; \targetTraj)$ with mixing weights 
$\emDist(\targetTraj)$. These weights encode the quality of the surrogate fit to the target, 
without considering the downstream fit to the data $\obs$. Generally speaking, the 
EP tends to be more diffuse (i.e., conservative) relative to the EUP, and more likely to 
preserve qualitative features such as multimodality in $\postEm$ \citep{RobertsUncProp}. 

By contrast, the EUP weights each component distribution by $\normCst(\targetTraj) \emDist(\targetTraj)$,
where $\normCst(\targetTraj)$ depends on $\obs$. Thus, the weight given to a posterior trajectory
$\postDens(\cdot; \targetTraj)$ takes into account both the quality of the surrogate fit to the
target, in addition to the marginal likelihood of $\obs$ given $\targetTraj$. This follows
from a hierarchical Bayesian perspective where the emulator is treated as a latent 
variable governing the data generating process for $\obs$. The EUP yields a good approximation
to the true posterior when $\lik(\cdot; \targetEm)$ is well-specified (the approximation is 
exact when $\lik(\cdot; \targetEm)$ is unbiased). However, under misspecification the 
EUP can produce undesirable and counterintuitive results. For example, a 
trajectory $\targetTraj$ may be a poor approximation to $\target$, but the EUP may up-weight
$\postDens(\cdot; \targetTraj)$ if the bad emulator prediction happens to fit the 
data $\obs$ well. In general, weighting trajectories by $\normCst(\targetTraj)$ is reasonable when 
a high value of $\normCst(\targetTraj)$ provides a reliable signal of approximation quality.
This happens when the true simulator is well-specified, and emulator approximation errors 
with respect to $\target$ degrade the overall model fit to $\obs$ in a way that is reflected
in $\normCst(\targetTraj)$. In general, these conditions are difficult to satisfy. One clear
example is the log-density emulation setting in \Cref{ex:eup-ldens-em}. In this case, 
$\normCst(\targetTraj)$ is artificially inflated by the emulator uncertainty, in which 
case the weights are a poor reflection of ``evidence'' for $\obs$. For a more thorough 
comparison of the EUP and EP, we refer to \citet{RobertsUncProp,BurknerSurrogate}.

Loosely speaking, when the shape of $\postDens(\cdot; \targetTraj)$ does not vary significantly
with $\targetTraj$ then the difference between the EUP and EP diminishes \citep{RobertsUncProp}. 
Certain emulators such as GPs naturally lead to greater shape variability relative to more 
constrained parametric approximations. If computational resources
allow, both the EP and EUP can be computed to test the sensitivity to a particular method. In some
cases, neither distribution may provide an adequate summary of $\postEm$; the final row of 
\Cref{fig:em_dist_1d} illustrates that the full distribution of this random posterior 
contains more information than any of the deterministic approximations alone provide.
Further work is needed to develop computationally efficient means of summarizing this uncertainty.
Moreover, in certain settings deterministic approximations other than the EP or EUP may 
be more adequate, such as those designed for robustness in specific cases \citep{generalizedCut}.

\section{Active Learning for Bayesian Inversion} \label{sec:active-learning}
Up until this point, we have considered a fixed surrogate model fit to training data 
obtained from simulator runs at a set of design points $\design \Def \{\Par_1, \dots, \Par_{\Ndesign}\}$.
The simplest approach to select these points is to sample from some distribution over $\parSpace$
(e.g., the prior $\priorDens$). Popular variants employ space-filling adjustments to encourage
the design points to be spread out \citep{initDesignReview}. 
Such a ``one-shot'' design is appealing in that it allows for maximal 
parallelism in simulator runs. However, the posterior distribution commonly concentrates on a small 
subset of the prior support, implying that a prior-based 
design will likely allocate the majority of points in regions with negligible posterior mass.
In many such cases, one cannot hope to achieve an acceptable posterior approximation with 
a static design based solely on prior information. 

To address this issue, the design can instead be constructed 
sequentially, with information from the current simulations used to inform the selection of new design points.
We refer to this broad class of algorithms as \inlinedef{active learning} for Bayesian inverse problems.
At the extreme end of this spectrum is a \inlinedef{pure sequential} strategy, where simulations are performed serially, 
one at a time. This allows for the most informed design point selection, but fails to exploit parallel computing 
resources and is thus impractical in many large-scale applications.
To balance these extremes, practical approaches typically adopt a batch-sequential workflow: simulator runs are 
spread across multiple rounds, with a batch of model runs executed in parallel each round.
For simplicity, we focus on this ``synchronous parallel'' setting, though asynchronous parallelism
is also possible \citep{parallelBOThompson}. 

The central challenge faced in active learning is 
navigating the ``exploration vs. exploitation'' 
trade-off \citep{reviewBayesOpt,BadiaRL}. Algorithms must exploit current information 
to refine the surrogate in known high-probability regions, while simultaneously exploring high-uncertainty 
areas to avoid missing other important regions.

Active learning algorithms have been explored extensively in the context of Bayesian inference. Our focus
here is on the sub-class of methods where design points are selected sequentially to refine a surrogate
model with the goal of improving a posterior approximation. We identify three (non-mutually exclusive)
strategies that are commonly used for this purpose:
\begin{enumerate}
    \item \inlinedef{Design Optimization}: Explicitly optimizing an objective function to select the next batch of points.
    \item \inlinedef{Tempering}: Constructing a sequence of intermediate target distributions to guide the surrogate toward the posterior.
    \item \inlinedef{MCMC-based Exploration}: Embedding the design construction directly within a posterior sampling algorithm.
\end{enumerate}

\subsection{Design Optimization} \label{sec:design-optimization}
Many active learning algorithms explicitly define a design criterion, or \inlinedef{acquisition function}, that quantifies the utility 
of running the simulator at a set of candidate locations. This objective function is then optimized at each round
of the active learning procedure to select the next batch of design points. Similar approaches are also widely applied in other
applications involving surrogate models. For example, the field of Bayesian optimization focuses on defining 
acquisition functions tailored to the optimization of  black-box functions \citep{reviewBayesOpt}. Acquisition functions for 
Bayesian inversion are analogous, but should instead encode the explicit goal of estimating the posterior distribution.

To better appreciate the considerations involved in defining an acquisition function, it is useful to consider how the 
criterion should behave in the absence of uncertainty. In Bayesian optimization, clearly the ideal behavior is to 
return the optimum of the function being optimized. In the present context, it is less clear which arrangement of points 
is optimal in representing a probability distribution with a surrogate model. Results suggest that the points should be sampled from an 
over-dispersed variant of the distribution, thereby ensuring adequate coverage of the tails \citep{StuartTeck2,briol2017sampling}. 
These results are consistent with the intuitive notion that design points should be placed in high-probability regions, while also 
ensuring adequate coverage to constrain the global surrogate behavior. With imperfect information, acquisition functions must 
balance this ideal with the need to reduce uncertainty in poorly explored regions.

The following sections explore design criteria that seek to navigate this particular flavor of an exploration-exploitation tradeoff.
We begin by formalizing the problem, discuss theoretical frameworks used for deriving design criteria, and then 
offer a survey of concrete acquisitions used in the literature. We close the section with a brief discussion of practical methods for 
carrying out the optimization of acquisition functions.

\subsubsection{Problem Setup}
Throughout this section, we consider an emulator $\targetEm$ that has been trained 
on an existing design $\trainData_{\Ndesign} = \{\design, \designResponse\}$, 
with the goal being to select a new batch $\ParBatch \subset \parSpace$ of $\Nbatch$ design points. 
For an acquisition function $\acq: \parSpace^\Nbatch \to \R$, the batch is 
selected by solving the optimization problem
\begin{equation}
\ParBatchOpt \Def \argmin_{\ParBatch \in \parSpace^\Nbatch} \acq(\ParBatch).
\label{eq:acq-opt}
\end{equation}
By focusing on only one step of a potentially multi-stage sequential design algorithm, we 
implicitly restrict our focus to \textit{myopic} (i.e., one-step look-ahead) strategies, 
which do not factor in the effect of future design acquisitions when selecting the current batch.
Non-myopic strategies have been considered (e.g., via dynamic programming), though they 
typically come at the cost of significant 
computational expense \citep{BONonMyopic,SURThesis}.

For a candidate batch $\ParBatch$, let $\ParBatchAug \Def \{\design, \ParBatch\}$ denote the extended 
set of design points, $\targetEm[\Naugment]$
the updated emulator, and $\jointEm[\Naugment]$ the resulting unnormalized posterior surrogate.
When necessary, we make explicit the dependence on the new training data by writing
$\targetEm[\Naugment]^{\ParBatch,\batchResponse}$ and 
$\jointEm[\Naugment]^{\ParBatch,\batchResponse}$,
where $\batchResponse \sim \simObsDist(\cdot \given \target(\ParBatch))$.
We similarly write $\acq(\ParBatch) = \acq(\ParBatch; \design, \designResponse)$ when it 
is helpful to emphasize the dependence of the acquisition function on the current training data.

\begin{remark}
The process of updating the surrogate from $\targetEm$ to $\targetEm[\Naugment]$ will depend on the
particular model in use. For GP surrogates, this entails conditioning 
$\targetEm[\Naugment] \Def \targetEm \given [\simObsPredProcess(\ParBatch) = \batchResponse]$, typically requiring 
$\BigO(\Ndesign^2 \Nbatch + \Nbatch^3)$ operations. The GP hyperparameters may also be updated, at the 
cost of $\BigO([\Ndesign + \Nbatch]^3)$ scaling. For parametric models, the surrogate update will entail 
re-estimation of the model parameters using the augmented dataset.
\end{remark}

\subsubsection{Frameworks for Constructing Design Criteria}
Similar to the posterior estimators in \Cref{sec:post-approx}, design criteria are commonly 
derived from first-principles frameworks.

\paragraph{Decision Theoretic Framework.}
Recall the unnormalized density decision theoretic framework formalized in \Cref{eq:decision-theory}, 
whereby an unnormalized posterior estimate is chosen as the optimizer of an expected loss:
\begin{equation}
\qFunc_\Ndesign \Def \argmin_{\qFunc \in \qFuncSpace} \emE[\loss(\qFunc, \jointEm)].
\end{equation}
The minimal value of the objective function is known as the \textit{Bayes' risk}:
\begin{equation}
\bayesrisk \Def \emE[\loss(\qFunc_\Ndesign, \jointEm[\Ndesign])].
\end{equation}
If the surrogate is updated using $\{\ParBatch, \batchResponse\}$, it will yield a corresponding Bayes' risk
$\bayesrisk[\Naugment]^{\ParBatch,\batchResponse}$. A natural design strategy is to choose $\ParBatch$
so that this risk is as small as possible. The problem, of course, 
is that we cannot observe $\batchResponse$ without evaluating $\target(\ParBatch)$.
Thus, we marginalize these values under the assumption $\batchResponse \sim \simObsPredDist(\ParBatch)$, 
which models the unobserved responses according to the current predictive distribution over simulator 
observations (recall \Cref{def:pred-sim-obs-process}). This yields an acquisition function of the form 
\begin{equation}
\acq(\ParBatch) \Def 
\E_{\batchResponse \sim \simObsPredDist(\ParBatch)}\left[\bayesrisk[\Naugment]^{\ParBatch,\batchResponse}\right].
\label{eq:acq-bayes-risk}
\end{equation} 
In principle, this framework could be modified to consider a loss
$\loss(\qDens, \postEm)$ between \textit{normalized} densities (as in \Cref{eq:random-measure-variational}).
To our knowledge this has not been considered in practice due to the practical difficulties caused by 
the normalizing constant.

\paragraph{Stepwise Uncertainty Reduction (SUR).}
\textit{Stepwise Uncertainty Reduction (SUR)} is another theoretical framework that has been used to develop 
and analyze sequential design algorithms for applications such as optimization and 
reliability analysis \citep{supermartingaleSUR,BectSUR}. The
framework has also recently been used in solving Bayesian inverse problems \citep{weightedIVAR}.
Let $\surmetric(\targetEm)$ denote a scalar-valued uncertainty measure we would
like to minimize. SUR acquisition functions are defined as the expected one-step-ahead uncertainty
\begin{equation}
\acq(\ParBatch) \Def 
\E_{\batchResponse \sim \simObsPredDist(\ParBatch)}\left[\surmetric(\targetEm[\Naugment]^{\ParBatch,\batchResponse})\right].
\end{equation}
This framework is quite general, and particular choices of $\surmetric$ coincide with criteria derived using the 
Bayesian decision theoretic approach. The primary benefit of SUR is theoretical, as it provides general conditions 
to ensure that the uncertainty will converge to zero if the algorithm is repeated indefinitely. 

\subsubsection{Concrete Acquisition Functions}
We now provide a survey of design criteria used in the literature.

\begin{figure}
    \centering
    \includegraphics[width=\textwidth, height=0.98\textheight, keepaspectratio]{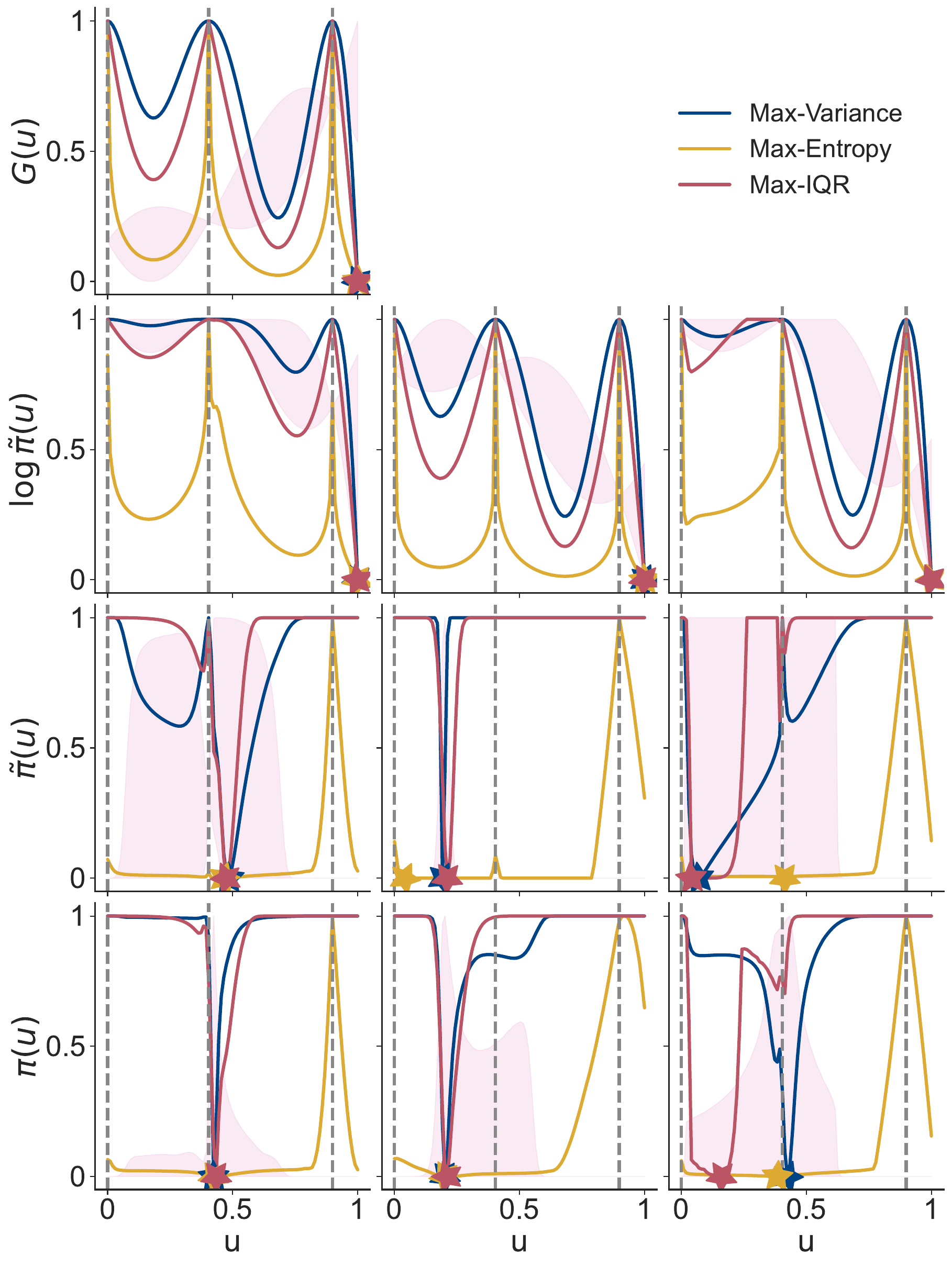}

    \caption{A continuation of \Cref{fig:em_dist_1d,fig:post_norm_approx_1d}. The shaded regions
    summarize the surrogate pushforward distributions as before. The lines correspond to different
    ``maximum uncertainty'' criteria \texorpdfstring{$\acq(\Par)$}{} targeting uncertainty in 
    the respective distributions. For example, the blue lines in the left column (from top to 
    bottom) show $-\Var(\targetEm(\Par))$, $-\Var(\log \jointEm(\Par))$, 
    $-\Var(\jointEm(\Par))$, and $-\Var(\postEm(\Par))$, where $\targetEm$ is a forward model emulator.
    The two other lines similarly show the (negated) pointwise entropy and interquartile range, and
    the columns vary the underlying emulator. The star markers indicate the optimal value for each 
    acquisition function.}
    \label{fig:acq_pw}
\end{figure}

\paragraph{Local Single-Point Criteria.} \label{sec:single-point}
Before describing acquisition functions derived from the above frameworks, we introduce some 
common heuristic strategies. When optimizing for a single design point ($\Nbatch = 1$), an intuitive 
approach is to simply select the input where the uncertainty in $\postEm(\Par)$ is largest.
Owing to the challenge of dealing with the random normalizing constant $\normCstEm$,
previous work has opted to instead target the uncertainty in the unnormalized density surrogate
$\jointEm(\Par)$. If predictive variance
is used as the measure of uncertainty, this yields the \inlinedef{maximum variance} criterion
\begin{equation}
\acq(\Par) \Def -\Var_{\emDist}[\jointEm(\Par)], \label{eq:acq-maxvar}
\end{equation}
which is negated to align with our convention of 
minimizing acquisition functions.
Alternative ``maximum uncertainty'' criteria can be defined by changing the measure
of uncertainty (e.g., variance, entropy, interquartile range) and the target quantity
for which uncertainty is assessed (e.g., $\jointEm$, $\targetEm$).
\Cref{fig:acq_pw} plots $\acq(\Par)$ under varying combinations of these 
choices for different underlying emulators. The particular criterion in \Cref{eq:acq-maxvar}
is used in \citet{Lueckmann2019} for sequential neural likelihood estimation, with the 
variance approximated using a deep ensemble. The below example notes other cases in 
which it is used in conjunction with GP emulators.

\begin{example}[Maximum Variance, Log-Density Emulator]
\label{ex:maxvar-ldens}
Consider the Gaussian log-likelihood emulation setting from \Cref{ex:eup-ldens-em}
where 
\begin{align}
&\jointEm(\Par) = \priorDens(\Par) \Exp{\targetEm(\Par)},
&&\targetEm(\Par) \sim \Gaussian(\emMean(\Par), \emVar(\Par)).
\end{align}
In this case, the maximum variance acquisition is
\begin{equation}
\acq(\Par) = -\priorDens^2(\Par) \left[\Exp{\emVar(\Par)} - 1 \right] \Exp{2\emMean(\Par) + \emVar(\Par)},
\label{eq:maxvar-llik}
\end{equation}
which favors the selection of points both where the current log-likelihood estimate $\emMean(\Par)$ and the emulator 
uncertainty $\emVar(\Par)$ are large. This criterion is used in 
\citet{Kandasamy2017,AlawiehIterativeGP}, though
\citet{VehtariParallelGP,wang2018adaptive} argue that it provides a misleading measure of uncertainty due to the 
heavy tails of $\jointEm(\Par)$. They propose using the interquartile range or entropy as more robust 
uncertainty metrics in this setting. The acquisition in \Cref{eq:maxvar-llik}
is plotted in \Cref{fig:acq_pw} (second row, second column), compared to the 
interquartile range and entropy alternatives.
\end{example}

A variety of other single-point acquisition strategies have been explored in the 
literature. Several authors have borrowed criteria from Bayesian optimization 
to target the placement of new design points in high-posterior 
regions \citep{Takhtaganov2018,BOforLFI,gp_surrogates_random_exploration}.
The strategy proposed by \citet{gp_surrogates_random_exploration} explicitly decouples
exploration and exploitation via a two-point acquisition; the first point is chosen as
the maximizer of the current log-posterior estimate (exploitation) and the second 
is sampled from $\priorDens$ (exploration).

While nominally designed for the pure sequential $\Nbatch = 1$ setting, single-point 
criteria can be adapted for batch acquisition. The simple approach of selecting the $\Nbatch$
points with the smallest acquisition values can result in a batch of very similar points.
Greedy heuristics, discussed in \Cref{sec:solve-opt}, can prevent this 
undesireable ``clumping'' behavior. Alternatively, ``stochastic'' variants of maximum 
uncertainty criteria can be defined by sampling from a density proportional to 
$-\acq(\Par)$, rather than optimizing the criterion \citep{BOforLFI,EfficientAcqABC}. 
This encourages exploration, and provides a mechanism for batch selection.  

\paragraph{Global Multipoint Criteria.}
While computationally convenient, single-point criteria suffer from two primary drawbacks: (1) they are not 
naturally defined in the batch sequential setting; and (2) they are purely local, neglecting the global impact a 
new design point has across $\parSpace$. The latter property commonly results in maximum uncertainty 
criteria targeting points in the edges or tails of the parameter space 
\citep[e.g.,][Section 6.2.1]{gramacy2020surrogates}. We now introduce a class 
of acquisition functions, referred to as \inlinedef{Expected Conditional Uncertainty (ECU)} criteria, that address 
both of these limitations.

Intuitively, ECU acquisitions arise from the following logic: we would like 
to select an input batch $\ParBatch$ that results in an updated (unnormalized) posterior 
estimate $\jointEm[\Naugment]$ with the lowest
possible uncertainty on average over $\parSpace$.
\footnote{As before, this could be extended to operate directly on the normalized densities. While arguably more principled,
this also leads to more challenging calculations.}
Given that we cannot typically compute this uncertainty without
observing $\batchResponse$, we utilize the same approach described for the Bayes' risk
(\Cref{eq:acq-bayes-risk})
and marginalize over this quantity with respect to $\simObsPredDist$. As a concrete example, 
if we again choose variance as our measure of uncertainty, we obtain the expected conditional variance
\begin{equation}
\acq(\ParBatch) \Def 
\int_{\parSpace} 
\E_{\batchResponse \sim \simObsPredDist(\ParBatch)} \left\{\Var_{\emDist[\Naugment]}[\jointDens(\Par; \targetEm[\Naugment]^{\ParBatch,\batchResponse}) \given \ParBatch, \batchResponse]\right\} \ \weightdens(\d\Par),
 \label{eq:acq-intvar}
\end{equation}
where $\weightdens$ is a measure over $\parSpace$. Variations of ECU criteria can be created by changing the measure 
of uncertainty (e.g., variance, entropy, interquartile range), the target quantity (e.g., $\jointEm$, $\targetEm$), and the 
weighting measure $\weightdens$. For example, targeting the variance of $\targetEm$ with uniform $\weightdens$ 
yields the classical integrated mean squared prediction error \citep{Mercer_kernels_IVAR}.
\citet{gpEmMCMC} target uncertainty in the acceptance probability of a Metropolis-Hastings
algorithm, deriving ECU-style acquisitions under a Bayesian decision theoretic framework.
Similarly, \citet{EfficientAcqABC} derive an integrated variance criterion targeting uncertainty 
in the ABC approximation to the (unnormalized) posterior. 

The expected variance term in \Cref{eq:acq-intvar} admits a closed form in the settings of
\Cref{ex:eup-ldens-em,ex:eup-fwd-em}, as shown below. These examples also demonstrate how ECU
criteria can be interpreted within the decision theoretic and SUR frameworks.

\begin{example}[Expected Conditional Variance, Log-Density Emulator]
\label{ex:evar-ldens}
We continue \Cref{ex:maxvar-ldens}, extending the maximum variance criterion into the expected conditional
variance. With this setup, the ECU-variance criterion in \Cref{eq:acq-intvar} reduces to
\begin{equation}
\acq(\ParBatch) = \int_{\parSpace} 
\Var_{\emDist}\left[\jointEm(\Par) \given \simObsPredProcess(\ParBatch) = \emMean[\Ndesign](\ParBatch) \right] \varInflation_{\Ndesign}(\Par; \ParBatch)
\weightdens(\d\Par)
\label{eq:acq-intvar-saa-ldens}
\end{equation}
where
\begin{align*}
\Var_{\emDist}\left[\jointEm(\Par) \given \simObsPredProcess(\ParBatch) = \emMean[\Ndesign](\ParBatch) \right]
&= \priorDens^2(\Par) \Exp{2\emMean[\Ndesign]\left(\Par\right) + \emVar[\Naugment](\Par; \ParBatch)} \cdot \\
&\qquad\qquad\qquad \left[\Exp{\emVar[\Naugment](\Par; \ParBatch)} - 1 \right] \\
\shortintertext{and}
\varInflation_{\Ndesign}(\Par; \ParBatch)
&= \Exp{2\left(\emVar[\Ndesign](\Par) - \emVar[\Naugment](\Par; \ParBatch)\right)}.
\end{align*}

We write $\emVar[\Naugment](\Par; \ParBatch)$ for the predictive variance of the updated GP, which
does not depend on the unobserved response $\batchResponse$.
The first term in \Cref{eq:acq-intvar-saa-ldens} is the variance of the unnormalized posterior surrogate
that results from updating $\targetEm$ with $\{\ParBatch, \emMean(\ParBatch)\}$; 
i.e., the current GP prediction is treated as if it were the true response.
The second term accounts for the uncertainty in the true response: the penalty
$\emVar[\Ndesign](\Par) - \emVar[\Naugment](\Par; \ParBatch)$ is large when the input batch 
$\ParBatch$ is highly ``influential.''

The expected conditional variance acquisition corresponds to the Bayes' risk under an $L^2$ 
loss \citep{VehtariParallelGP}. It can also be viewed as an SUR acquisition with uncertainty measure
$\surmetric(\targetEm) \Def \int \Var_{\emDist}[\jointEm(\Par)] \weightdens(\d\Par)$.
This acquisition is considered in \citet{VehtariParallelGP}, but the authors ultimately recommend
alternative robust criteria defined using the interquartile range in place of the variance.
\end{example}

\begin{example}[Expected Conditional Variance, Forward Model Emulator]
\label{ex:evar-fwd}
Consider the Gaussian forward model setting with 
$\jointEm(\Par) = \priorDens(\Par) \Gaussian(\obs \given \targetEm(\Par), \covNoise)$, 
$\targetEm(\Par) \sim \Gaussian(\emMean(\Par), \emVar(\Par))$. Letting 
$\CovComb(\Par) \Def \covNoise + \emVar(\Par)$ and 
$\CovComb[\Naugment](\Par) \Def \covNoise + \emVar[\Naugment](\Par; \ParBatch)$, 
the expected conditional 
variance criterion simplifies to
\begin{align}
\acq(\ParBatch) &= 
\int_{\parSpace} \priorDens^2(\Par) \bigg[\frac{\Gaussian\left(\obs \given \emMean(\Par), \CovComb(\Par) - \frac{1}{2}\covNoise \right)}{2^{\dimObs/2} \det(2\pi \covNoise)^{1/2}} - \label{eq:acq-int-var-fwd} \\
&\qquad\qquad \frac{\Gaussian\left(\obs \given \emMean(\Par), \CovComb(\Par) - \frac{1}{2}\CovComb[\Naugment](\Par; \ParBatch) \right)}{2^{\dimObs/2} \det(2\pi \CovComb[\Naugment](\Par; \ParBatch))^{1/2}} \bigg] \weightdens(\d\Par). \nonumber
\end{align}
\citet{SinsbeckNowak} are the first to propose an ECU-variance criterion in the forward model emulation setting,
and \citet{Surer2023sequential,SurerBatchStochastic} utilize similar acquisitions. 
\citet{weightedIVAR} alternatively consider a criterion of the form
\begin{equation}
\acq(\ParBatch) = \int_{\parSpace} \emVar[\Naugment](\Par; \ParBatch) \postApproxEUP(\d\Par),
\label{eq:weighted-ivar}
\end{equation}
corresponding to the SUR uncertainty metric 
$\surmetric(\targetEm) \Def \int \Var_{\emDist}[\targetEm(\Par)] \postApproxEUP(\d\Par)$. 
The integrand considers uncertainty
in the forward model emulator (exploration), while the measure $\weightdens = \postApproxEUP$ 
weights by the current posterior approximation (exploitation).  
\end{example}

Outside of special cases \citep{Binois_2018,MakTargetedVar,Koermer2024} the outer 
integral (over $\parSpace$) in ECU criteria is not tractable. Monte Carlo approximations
are possible but would yield a stochastic optimization problem. A simpler and more common solution is 
to utilize a sample average approximation
\begin{align}
&\acq(\ParBatch) \Def \frac{1}{J}
\sum_{j=1}^{J} 
\E_{\batchResponse \sim \simObsPredDist(\ParBatch)} 
\left\{\Var_{\emDist[\Naugment]}[\postDens(\Par_j; \targetEm[\Naugment]^{\ParBatch,\batchResponse}) \given \ParBatch, \batchResponse]\right\},
&&\Par_j \overset{\mathrm{iid}}{\sim} \weightdens,
 \label{eq:acq-intvar-saa}
\end{align}
in which the $\Par_j$ are sampled once and then fixed, maintaining a deterministic 
objective function \citep{Mercer_kernels_IVAR,botorch}.

\subsubsection{Solving the Optimization Problem} \label{sec:solve-opt}
We comment briefly on algorithms to actually carry out the optimization of the acquisition 
function in \Cref{eq:acq-opt}. Detailed treatments of this topic are available across the Bayesian 
optimization and computer experiments literatures \citep{maxAcq,botorch,gramacy2020surrogates,WynnDiscreteExchange}.
In general, the task of minimizing $\acq(\ParBatch)$ involves a challenging non-convex optimization problem over
$\dimPar\Nbatch$ variables. Gradient-based optimization of acquisitions defined by sample-average approximations
(\Cref{eq:acq-intvar-saa}) has been shown to perform favorably in moderate dimensions \citep{Mercer_kernels_IVAR}.
This approach is implemented in the \verb+BoTorch+ Python package \citep{botorch}. Alternatively, a large portion of the literature 
eschews continuous optimization in favor of discrete methods. These algorithms take the form of a subset selection 
problem, with the optimal batch selected from a finite set of candidate points. This presents a 
challenging combinatorial optimization problem, commonly solved by employing stochastic exchange 
algorithms (e.g., Federov exchange) to find a local minimum \citep{FederovExchange,WynnDiscreteExchange,LOEPPKY20101452}.
\footnote{Exchange algorithms can also be employed to optimize over continuous space and leverage gradient information.}

To avoid the joint optimization over $\parSpace^{\Nbatch}$ altogether, greedy approximations are also common in 
practice \citep{VehtariParallelGP,Surer2023sequential}. These methods seek to approximate the 
joint optimum using a sequence of $\Nbatch$ simpler optimization problems over $\parSpace$. 
Since this procedure proceeds one point at a time, it provides the opportunity to use single-point criteria 
(\Cref{sec:single-point}) in batch optimization. 
Under this approach, the first point $\Par_1$ is selected
by optimizing $\acq(\Par)$ as in the pure sequential setting. The remaining $b = 2, \dots, \Nbatch - 1$ points are added
sequentially by minimizing $\acq(\Par; \ParBatch_{\Ndesign+b-1}, \batchResponse_{\Ndesign+b-1})$.
In pure sequential design, 
the response $\simObsProcess(\Par_b)$ is observed every iteration. Greedy batch design differs in that 
$\simObsProcess(\Par_b)$ is not observed until the entire batch has been selected. 
This implies that the objective function must be 
approximated in solving this sequence of optimization problems, typically by imputing the unknown
responses $\simObsProcess(\Par_b)$ with pseudo-observations. 
\textit{Kriging Believer} and \textit{Constant Liar} are two such heuristics
\citep{Ginsbourger2010}. After selecting $\Par_b$, the Kriging Believer strategy updates the surrogate with 
the pseudo-observation $\{\Par_b, \emMean(\Par_b)\}$, while Constant Liar instead uses 
$\{\Par_b, c\}$ for some pre-determined constant $c$. An extension of this approach consists 
of generating multiple candidate batches using varying heuristics, and then leveraging a batch 
acquisition function to select among the candidates \citep{Chevalier2013}. 

\subsection{Surrogate-Guided Sampling} \label{sec:surrogate-guided-sampling}
Instead of optimizing an explicit acquisition function, the particular goal of posterior approximation
suggests a simple alternative: sample the batch $\ParBatch$ from the current posterior estimate (e.g., $\postApproxEP$).
This yields an active learning procedure with the form of a fixed-point iteration; the current posterior estimate is 
used to select the next batch, which is then used to update the posterior estimate, and so on 
\citep{murrayNPE,SNLE,hydrologicalModel,quantileApprox,hydrologicalModel2}. \citet{FerEmulation} adopt a 
variant of this strategy, sampling from a mixture of the current approximate posterior and the prior, with the mixture 
weights representing a user-defined tuning parameter. \citet{turbulenceModelAdaptiveLHS} consider a similar 
prior-posterior mixture, with a Latin hypercube adjustment to 
encourage the new design points to be spread out. \citet{adaptiveMultimodal} draw samples from the current approximate 
posterior, and then update these samples with an iteration of an ensemble Kalman algorithm, a method with close
connections to the \textit{calibrate, emulate, sample} workflow (\citet{CES}; see \Cref{sec:tempering}).
Surrogate-based sampling is also commonly used in sequential algorithms for constructing 
neural density estimators \citep{murrayNPE,SNLE,MurrayNLEvsNPE}. This approach is especially convenient
in this setting given that (unlike acquisition-based approaches) it does not require a probabilistic 
emulator, which can be difficult to train for neural networks.

These sampling-based approaches enjoy the practical benefits of being inherently parallel and avoiding a difficult
optimization problem. In this sense, they are similar in spirit to Thompson sampling for
batch Bayesian optimization \citep{parallelBOThompson}. Sampling from a mixture of the prior
and current posterior estimate can be viewed as an explicit decoupling of exploration and exploitation, similar to the
methods proposed in \citet{gp_surrogates_random_exploration}. One could try to combine the strengths of
sampling-based and optimization-based approaches by first sampling candidates from the current posterior estimate 
and then running an optimization algorithm to select the final subset; the ECU acquisition in \Cref{eq:weighted-ivar} 
might be viewed as an instance of this approach. 

Overall, sampling-based approaches offer a promising avenue in large-scale applications where parallelization is crucial, 
and the simulator cost only allows for a few rounds of sequential design with large batch sizes.

\subsection{Tempered Target Distributions} \label{sec:tempering}
The sequential design strategies of the previous section are motivated by the observation
that a static, prior-based design will often fail to allocate points in high-density regions when the 
posterior is concentrated. They address this problem by building up the design sequentially,
targeting the placement of points where the surrogate-estimated posterior is large, balanced
with exploration of regions with high surrogate uncertainty. A concern with such algorithms is that 
they may require many iterations to identify the high-density regions, 
exhausting the budget exploring unimportant areas. This problem arises when the surrogate 
approximation of the posterior is quite poor in the early stages of the algorithm, causing 
acquisition functions based on the surrogate to inefficiently 
explore the parameter space \citep{WilkinsonABCGP}.

One approach to mitigate this issue builds on the recognition that it may be unrealistic and 
counterproductive to directly target the posterior in the initial stages of the algorithm. 
Instead, the surrogate target can be allowed to vary over the rounds in the active 
learning procedure, tracking a sequence of intermediate distributions that bridge
between the prior and posterior. Since the earlier distributions are less concentrated, the surrogate 
approximation will tend to be more reasonable and thus more useful in guiding the placement of the 
next round of design points. \inlinedef{Tempering} strategies of this form are widely utilized in 
computational statistics; we focus on their use in conjunction with surrogate models. 

Tempering algorithms traverse from prior to posterior through a sequence of interpolating
distributions $\postDens_0, \postDens_1, \dots, \postDens_T$, where $\postDens_0$ is the prior and $\postDens_T = \postDens$
the posterior. Common choices are based on geometric bridges of the form 
$\postDens_t(\Par) \propto \priorDens(\Par) \lik(\Par)^{\beta_t}$ with $0 = \beta_0 < \dots < \beta_T = 1$.
This approach forms the basis of sequential Monte Carlo (SMC) algorithms, used both for statistical 
estimation in time-varying systems (i.e., particle filters) as well as posterior estimation for 
static inverse problems \citep{DoucetSMC,ElementsSMC}. 
Unlike MCMC, these methods are naturally parallelizable,
but can still require many likelihood evaluations. Surrogate models offer one avenue to 
lower this cost, and may be integrated into a tempering framework in a variety of different 
ways. 

The role of the intermediate distributions $\postDens_t$ may differ depending on the emulator target.
In forward model emulation, the target $\fwd$ is constant. However, the design
strategy may vary over $t$ to produce an approximation of $\fwd$ optimized for 
$\postDens_t$. The sequence of surrogates may thus be \inlinedef{local}, in the 
sense that certain design points are discarded or downweighted to target 
predictions in the high-density region with respect to $\postDens_t$ \citep{WilkinsonABCGP,Li_2014}.
In other cases, the surrogate target itself may evolve.
For example, at iteration $t$ a log-posterior surrogate is fit to training 
data of the form $\{\Par_\designIdx, \log \postDens_t(\Par_\designIdx)\}$, which 
again may constitute a sub-sampled or weighted modification of the total 
training set. 
Note that design points from previous rounds can still be re-used 
without additional model simulations. For example, in the case of geometric 
bridging distributions,
$\log \postDens_t(\Par) = \log \priorDens(\Par) + \beta_t \log \lik(\Par)$,
implying that previous likelihood evaluations simply need to be rescaled for use in future steps.

While tempering methods are diverse, we restrict our focus to three broad categories:
(1) algorithms relying on sequential importance sampling; (2) algorithms that 
utilize approximate Kalman-style updates to traverse between the bridging distributions;
and (3) other related methods.

\paragraph{Importance Sampling.}
In principle, importance sampling provides the means to transform samples from a prior proposal 
distribution into posterior samples using evaluations of the 
likelihood function \cite[Chapter 9]{OwenMCBook}.
When utilizing tempering, these updates are applied sequentially such that $\pi_t$ acts as the 
proposal distribution for $\postDens_{t+1}$. To reduce the computational cost of the likelihood evaluations, 
surrogates can be leveraged in various ways to approximate these updates. One natural strategy is to replace 
$\pi_t$ with a surrogate-based estimate $\approxPost_t$, then apply the importance sampling 
machinery to sample $\approxPost_{t+1}$. Samples from $\approxPost_{t+1}$ can then 
be chosen as the new batch of design points. \citet{quantileApprox} adopt this approach
with a GP log-posterior surrogate, never discarding design points and thus maintaining 
a global surrogate approximation throughout the algorithm. By contrast, \citet{Li_2014}
focus on the construction of local polynomial chaos surrogates. The $\approxPost_t$
are estimated using a variational approximation, with the variational objective approximated
using importance sampling and the current surrogate model.
\citet{adaptiveReducedTempering} propose an adaptive SMC scheme that actively samples new 
design points to reduce surrogate error. Both \citet{Li_2014,adaptiveReducedTempering}
automatically adapt the temperature schedule $\{\beta_t\}$, negating the need for 
a priori specification.

\paragraph{Ensemble Kalman Updates.}
While importance sampling techniques are widely applicable, they can suffer from known pathologies
in moderate to high dimensions. Ensemble Kalman methods, rooted in linear Gaussian methodology, sacrifice 
exact updates $\postDens_t \mapsto \postDens_{t+1}$ for robustness and are widely used in 
high-dimensional settings \citep{IglesiasEKI}. 
These approximate methods can be used as a drop-in replacement for the importance sampling updates
discussed in the previous section. \citet{adaptiveMultimodal} proceed in this spirit, first sampling 
from an intermediate posterior approximation using MCMC and then feeding the resulting samples 
through an ensemble Kalman update to produce the next batch of design points. 
\citet{PicchiniABCEKI} utilize ensemble Kalman updates to build a sequence of proposal distributions
for an MCMC sampler.

Instead of interleaving Kalman steps with surrogate updates,
another popular strategy consists of first running several iterations of a Kalman algorithm
and then utilizing the resulting approximate posterior samples as an initial design for 
an emulator. This methodology, known as the \textit{Calibrate, Emulate, Sample (CES)} approach \citep{CES},
is an instance of the general idea of using an approximate sampling or optimization algorithm 
to cheaply generate design points in the high-density regions. Surrogate-based posterior approximations
may then be constructed to refine the initial coarse approximation. A variety of ensemble Kalman methods, 
falling under the umbrella label \textit{Ensemble Kalman Inversion (EKI)}, may be utilized in the first 
stage of CES algorithms \citep{EKIRace,CESSoftware,idealizedGCM,FATES_CES}.

\paragraph{Other Approaches.}
A variety of other methods utilize tempering-like strategies outside of the traditional SMC framework.
For example, \citet{JosephMinEnergy,JosephMEDSampling} solve a sequence of optimization problems 
to deterministically sample from a probability distribution using a minimum energy criterion.
These algorithms utilize both tempering and local GP surrogates to improve efficiency.

The framework of \citet{wang2018adaptive} does not explicitly use tempering, but similarly leverages
a sequence of intermediate distributions. They factor the unnormalized posterior density as 
$\jointDens(\Par) = \priorDens(\Par) \lik(\Par) = \Exp{g_t(\Par)}\postDens_t(\Par)$, where a GP emulator targets 
$g_t(\Par) \Def \log\{\jointDens(\Par) / \postDens_t(\Par)\}$ and the density $\postDens_t$ is a free parameter.
The emulator target varies as $\postDens_t$ changes, but the quantity 
$\Exp{g_t(\Par)}\postDens_t(\Par)$ always provides an approximation of $\jointDens$.
For $\pi_t = \priorDens$,
$g_t(\Par)$ is the log-likelihood, but if $\postDens_t = \postDens$ then $g_t(\Par)$ reduces to a 
constant. Thus, the idea is that the emulation target becomes simpler to approximate over the course
of the algorithm. \citet{adaptiveMultimodal} adapt this idea in developing an algorithm tailored to 
multi-modal posteriors.

\citet{WilkinsonABCGP} suggests an alternative method that iteratively eliminates ``implausible''
regions of the parameter space. A sequence of local GP surrogates are trained only on points
in the current plausible region, and are used for determining the plausible region for the
next iteration.

\subsection{MCMC-Based Sequential Design} \label{sec:mcmc-design}
Finally, we highlight a class of active learning strategies that interleave surrogate refinement 
and posterior sampling within a unified algorithm. The goals of such approaches
vary, spanning purely approximate schemes as well as so-called ``exact-approximate'' algorithms that 
strive for asymptotic convergence to the exact posterior $\postDens$. We provide a high-level 
overview of these methods here, referring to \citet{noisyMCSurvey} for a more detailed survey. 

The idea of accelerating MCMC using a surrogate dates back to the delayed acceptance mechanism of 
\citet{DelayedAcceptance}. Delayed acceptance schemes are exact-approximate algorithms that
improve efficiency by filtering our poor proposals using the surrogate, but ultimately require exact 
simulations when a proposal is accepted \citep{MCMC_GP_proposal,RasmussenGPHMC}. 
To further reduce computational costs,
many methods also use the surrogate in making acceptance decisions, which generally
results in inexact samplers (i.e., the target distribution is no longer $\postDens$). 
Others seek a middle ground by using the surrogate in acceptance decisions only when 
the model is confident in its predictions; when uncertainty exceeds some threshold,
an exact simulation is triggered.
These additional simulator runs are used to update the surrogate, thus refining the emulator 
design over the course of the algorithm. This updating is often 
halted after the simulation budget is exceeded, at which point the sampler proceeds 
with a fixed surrogate targeting an approximate posterior \citep{gpEmMCMC,bliznyuk2012}. 
At this stage, any of the posterior estimates from 
\Cref{sec:post-approx} may be employed. 

In defining a criterion for triggering additional simulations, a principled choice is to consider 
uncertainty in the Metropolis-Hastings acceptance probability $\accProbMH(\Par, \widetilde{\Par})$
\citep{GPSABC,ConradLocalExactMCMC,gpEmMCMC}. When 
the surrogate-induced uncertainty in this quantity is large, the Markov chain is unsure whether 
or not to accept the proposed value $\widetilde{\Par}$. \citet{GPSABC} characterize the probability
of making an incorrect decision under the distribution induced by a GP surrogate in a 
simulation-based inference setting. \citet{gpEmMCMC} refine these results under a GP 
log-likelihood surrogate, providing an analytical characterization of the uncertainty 
in $\accProbMH(\Par, \widetilde{\Par})$. \citet{ActiveLearningMCMC,DrovandiPMGP} adopt a more 
heuristic approach, triggering design augmentation when the variance in 
$\targetEm(\widetilde{\Par})$ exceeds a predefined threshold.
\citet{ConradLocalExactMCMC} also consider uncertainty in the acceptance probability, but
in contrast to the above approaches, they (1) allow surrogate refinement to continue indefinitely,
theoretically establishing asymptotic exactness under certain conditions; and (2) utilize local 
surrogate models, such that only a local subset of the design is used when predicting at a 
particular location.

Due to the serial nature of MCMC 
algorithms, these methods are limited in their ability to leverage 
parallel computing resources. They are thus typically most appropriate for simulators of modest 
computational expense, when it is still feasible to perform many runs sequentially.
This is especially true for exact-approximate schemes; the price to be paid for 
asymptotic exactness is a larger number of simulator runs \citep{ConradLocalExactMCMC}.

\section{Conclusion} \label{sec:conclusion}
In this paper, we presented a comprehensive synthesis of surrogate-based methodology
for Bayesian inference. Our review emphasized uncertainty propagation and active
learning, two key elements of a robust surrogate-based Bayesian workflow.
While research in these areas has enabled more efficient and reliable use of 
surrogates, there remains significant opportunity for further progress.

Training emulators to approximate very expensive simulators continues to present
a significant hurdle. Gaussian processes are often the tool of choice in this
setting. However, constructing GP models for high-dimensional input spaces or 
nonstationary target surfaces is highly non-trivial 
\citep{highDimBO,nonstationaryGP}. Neural networks offer scalability, but 
typically require large training sets that are unattainable for expensive
computer models. Moreover, while many approaches have been proposed to produce 
predictive distributions for neural networks, uncertainty quantification
in deep learning generally remains an open problem \citep{UQDLReview}. 
The development of emulators combining the strengths of GPs and neural 
networks is a promising research direction. Emulators geared towards posterior 
approximation may achieve this balance by exploiting structure in the Bayesian 
inverse problem (e.g., \citet{CuiLIS}).

These challenges are further exacerbated when the simulator is noisy, which typically
necessitates larger simulator budgets. Neural density estimation is a powerful 
SBI paradigm for constructing flexible and amortized emulators \citep{frontierSBI}. 
However, it often requires large training sets that are not feasible for very 
expensive computer models. By contrast, methods from the computer experiments 
community target the small training set regime, but are typically 
not explicitly aligned with the goal of posterior approximation \citep{stochasticComputerModels}. 
Combining ideas from these two communities may prove a fruitful avenue for research;
for example, using emulators to generate training data for a neural
estimator \citep{BurknerAmortizedSurrogate}, or designing more query-efficient conditional
density estimators \citep{FadikarAgentBased}.

In discussing emulator uncertainty propagation in \Cref{sec:post-approx}, we emphasized
the EP and EUP, while highlighting a number of alternative posterior approximations.
Given that there is no one-size-fits-all approach, further work is needed to provide
guidance in specific applications. In addition, considerations such as robustness to 
misspecification may lead to alternative uncertainty propagation methods with 
superior performance in particular settings (e.g., improving the robustness
of posteriors derived using log-density emulators). Connections to areas such as 
semi-modular and generalized Bayesian inference may prove helpful 
in this regard \citep{semiModular,generalizedCut}.

As illustrated in \Cref{sec:active-learning}, many strategies have been proposed
for sequentially refining a surrogate for posterior approximation. In general, 
more work is required to better understand their tradeoffs and inform the 
choice of a particular algorithm in practice. Both experimental and theoretical
results are likely needed to close this gap, building upon early work in this area
\citep{StuartTeck2,gp_surrogates_random_exploration,VehtariParallelGP,weightedIVAR}.
While many previous studies have focused on active learning with single-point or 
small batch acquisitions, there is relatively little research in the challenging setting
of sequential design algorithms with few iterations and 
large batch sizes \citep{VehtariParallelGP,Surer2023sequential,MurrayNLEvsNPE}. 
This situation arises when the simulator is expensive but parallel computing 
resources are available. Surrogate-guided sampling (\Cref{sec:surrogate-guided-sampling})
appears to be a natural and practical approach in this context, but to our knowledge
has not been rigorously investigated.

We hope this review serves as a first step towards categorizing the many 
choices that users face when incorporating an emulator within a Bayesian analysis.
Further work is needed to better understand the tradeoffs between these choices,
provide general guidance for practitioners, and develop standardized diagnostics 
for surrogate-based Bayesian workflow.

\appendix

\section{Gaussian Process Derivations}
We now derive the expressions presented in the examples considering GP emulators throughout the paper.
Let $\targetEm \sim \GP(\emMean, \emKer)$, with $\emVar(\Par) = \emKer(\Par, \Par)$.
We focus on two settings amenable to analytical calculations: 
(1) forward model emulation with a Gaussian likelihood; and (2) log-likelihood emulation.
All expectations are assumed to be with respect to $\targetEm$ unless otherwise noted.

Before proceeding with the specific cases, we prove a general result that characterizes
the uncertainty in the GP predictive mean due to an unobserved target response. This 
result will be used in deriving the expected conditional uncertainty acquisition functions.
A similar result is presented in \citet{VehtariParallelGP}.

\begin{lemma}
\label[lemma]{lemma:pred-mean-dist}
Let $\targetEm[\Naugment]^{\ParBatch,\batchResponse} \sim \GP(\emMean[\Naugment]^{\ParBatch,\batchResponse}, \emKer[\Naugment]^{\ParBatch})$ 
denote the GP $\targetEm$ conditioned on the
additional training points $\{\ParBatch,\batchResponse\}$. Assume 
$\batchResponse \sim \Gaussian(\emMean(\ParBatch), \emVar(\ParBatch))$. Then, for any $\Par \in \parSpace$, 
the randomness in 
$\emMean[\Naugment]^{\ParBatch,\batchResponse}(\Par)$ as a function of $\batchResponse$ is characterized by
\begin{equation}
\emMean[\Naugment]^{\ParBatch,\batchResponse}(\Par) \sim \Gaussian(\emMean(\Par), \emKer[\Ndesign](\Par) - \emKer[\Naugment]^{\ParBatch}(\Par)).
\end{equation}
\end{lemma}

\begin{proof}
Note that the GP predictive kernel depends only on the input locations, not the responses. First recall
the form of the GP predictive mean

\begin{equation}
\emMean[\Naugment]^{\ParBatch,\batchResponse}(\Par) =
\emMean(\Par) + \emKer(\Par, \ParBatch) \kerMat_{\Ndesign}^{-1}[\batchResponse - \emMean(\ParBatch)] \label{eq:pred-mean} \\
\emKer(\ParBatch, \Par),
\end{equation}
where $\kerMat_{\Ndesign} \Def \emVar(\ParBatch) + \sigma^2 I$ as in \Cref{eq:pred-mean} (in this case $\targetEm$
acts as the GP prior). This quantity is an affine function of the Gaussian random vector $\batchResponse$ and is thus
Gaussian distributed. Computing its mean and variance yields the desired result.
\end{proof}

\subsection{Forward Model Emulation}

The forward model emulation setting yields the unnormalized posterior surrogate
$\jointEm(\Par) = \priorDens(\Par)\Gaussian(\obs \given \targetEm(\Par), \covNoise)$.
We begin by computing the mean and variance of this quantity. The variance derivation
uses the following lemma.

\begin{lemma} 
\label[lemma]{lemma:squared_Gaussian_density}
Let $\Gaussian(m, C)$ be a Gaussian distribution on $\R^{\dimObs}$ with $C$ a symmetric, positive-definite 
matrix. Then, for $\predObs \in \R^{\dimObs}$, 
\begin{align*}
\Gaussian(\predObs \given m, C)^2 
&= 2^{-\dimObs/2} \det(2\pi C)^{-1/2} \Gaussian(\predObs \given m, C/2) .
\end{align*}
\end{lemma}

\begin{proof}
\begin{align*}
\Gaussian(\predObs \given m, C)^2 
&= \det(2\pi C)^{-1} \Exp{-\frac{1}{2} (\predObs - m)^\top \left[\frac{1}{2}C \right]^{-1}(\predObs - m)} \\
&= \det(2\pi C)^{-1} \det(2\pi (1/2)C)^{1/2} \Gaussian\left(\predObs \given m, C/2\right) \\
&= 2^{-\dimObs/2} \det(2\pi C)^{-1/2} \Gaussian\left(\predObs \given m, C/2\right)
\end{align*}
\end{proof}

The expressions for the first two moments of $\jointEm(\Par)$ in \Cref{ex:fwd-em} follow directly from the 
following result with $A$ set to the identity.

\begin{proposition} 
\label[proposition]{prop:Gaussian_marginal_moments}
Let $\mu \sim \Gaussian(m, C)$ be a $\dimPar$-dimensional Gaussian vector and 
$A \in \R^{\dimObs \times \dimPar}$ a fixed matrix.
Assuming $\covNoise$ and $C$ are symmetric, positive definite matrices,
then
\begin{align*}
\E\left[\Gaussian(\predObs \given A \mu, \covNoise) \right] &= \Gaussian(\predObs \given Am, \covNoise + ACA^\top) \\
\Var\left[\Gaussian(\predObs \given A \mu, \covNoise) \right] 
&= \frac{\Gaussian\left(\predObs \given Am, \frac{1}{2} \covNoise + ACA^\top \right)}{2^{\dimObs/2} \det(2\pi \covNoise)^{1/2}} - 
\frac{\Gaussian\left(\predObs \given Am, \frac{1}{2}[\covNoise + ACA^\top] \right)}{2^{\dimObs/2} \det(2\pi[\covNoise + ACA^\top])^{1/2}}
\end{align*}
\end{proposition}

\begin{proof} 
The expectation follows immediately from the expression for the convolution of two Gaussians.
For the variance, we have 
\begin{align}
\Var\left[\Gaussian(\predObs \given A \mu, \covNoise) \right] 
&= \E\left[\Gaussian(\predObs \given A \mu, \covNoise)^2 \right] - \E\left[\Gaussian(\predObs \given A \mu, \covNoise) \right]^2 \nonumber \\
&= \E\left[\Gaussian(\predObs \given A \mu, \covNoise)^2 \right] - \Gaussian(\predObs \given Am, \covNoise + ACA^\top)^2. \label{two_terms_variance}
\end{align}
Starting with the first term, we apply \Cref{lemma:squared_Gaussian_density} and 
the convolution formula, respectively, to obtain 
\begin{align*}
\E\left[\Gaussian(\predObs \given A \mu, \covNoise)^2 \right]
&= 2^{-\dimObs/2} \det(2\pi \covNoise)^{-1/2} \E[\Gaussian(\predObs \given A\mu, \covNoise/2)] \\
&= 2^{-\dimObs/2} \det(2\pi \covNoise)^{-1/2} \Gaussian(\predObs \given Am, \covNoise/2 + ACA^\top).
\end{align*}
For the second term in \Cref{two_terms_variance}, another application of \Cref{lemma:squared_Gaussian_density} gives
\begin{align*}
\Gaussian(\predObs \given Am, &\covNoise + ACA^\top)^2 \\
&= 2^{-\dimObs/2} \det(2\pi[\covNoise + ACA^\top])^{-1/2} \Gaussian(\predObs \given Am, [\covNoise + ACA^\top]/2).
\end{align*}
Plugging these expressions back into \Cref{two_terms_variance} completes the proof. 
\end{proof}

The EUP $\jointApproxEUP(\Par) =
\priorDens(\Par)\Gaussian(\obs \given \emMean(\Par), \covNoise + \emVar(\Par))$ follows
immediately from this result. We next derive the expected variance term that appears in the 
acquisition function from \Cref{ex:evar-fwd}.

\begin{proposition}
Let $\jointDens(\Par; \targetEm[\Naugment]^{\ParBatch,\batchResponse}) = 
\priorDens(\Par) \Gaussian(\obs \given \targetEm[\Naugment]^{\ParBatch,\batchResponse}(\Par), \covNoise)$ denote the 
unnormalized posterior surrogate induced by updating the emulator $\targetEm$ with new training
data $\{\ParBatch, \batchResponse\}$. Under the assumption 
$\batchResponse \sim \Gaussian(\emMean(\ParBatch), \emVar(\ParBatch))$, the expected conditional variance
is given by
\begin{align}
\E_{\batchResponse} \left\{\Var[\jointDens(\Par; \targetEm[\Naugment]^{\ParBatch,\batchResponse}) \given \ParBatch, \batchResponse]\right\}
&= \priorDens^2(\Par) \bigg[\frac{\Gaussian\left(\obs \given \emMean(\Par), \CovComb(\Par) - \frac{1}{2}\covNoise \right)}{2^{\dimObs/2} \det(2\pi \covNoise)^{1/2}} - \\
&\frac{\Gaussian\left(\obs \given \emMean(\Par), \CovComb(\Par) - \frac{1}{2}\CovComb[\Naugment](\Par; \ParBatch) \right)}{2^{\dimObs/2} \det(2\pi \CovComb[\Naugment](\Par; \ParBatch))^{1/2}} \bigg],
\end{align}
where $\CovComb(\Par; \ParBatch) \Def \covNoise + \emVar(\Par)$, 
$\CovComb[\Naugment](\Par; \ParBatch) \Def \covNoise + \emVar[\Naugment](\Par; \ParBatch)$, and 
$\emVar[\Naugment](\Par; \ParBatch) \Def \emKer[\Naugment]^{\ParBatch}(\Par, \Par)$.
\end{proposition}

\begin{proof}
The prior density is a constant, and thus is pulled out of the variance and squared. 
Applying \Cref{prop:Gaussian_marginal_moments}, the conditional variance term thus becomes
\begin{align}
&\Var[\Gaussian(\obs \given \targetEm[\Naugment]^{\ParBatch,\batchResponse}(\Par), \covNoise) \given \ParBatch, \batchResponse]
= \\ &\frac{\Gaussian(\obs | \emMean[\Naugment]^{\ParBatch,\batchResponse}(\Par), 
\frac{1}{2}\covNoise + \emVar[\Naugment](\Par; \ParBatch))}{\det(2\pi \covNoise)^{1/2}}
- \frac{\Gaussian(\obs \given \emMean[\Naugment]^{\ParBatch,\batchResponse}(\Par), \frac{1}{2}[\covNoise + \emVar[\Naugment](\Par; \ParBatch)])}{\det(2\pi [\covNoise + \emVar[\Naugment](\Par; \ParBatch)])^{1/2}}. \nonumber
\end{align}
The denominators in the above expression can be pulled out of the outer expectation and are seen to equal 
the denominators in the desired expression. We thus complete the proof by 
computing the expectation of the numerators with respect to $\batchResponse$. 
Applying \Cref{lemma:pred-mean-dist} and \Cref{prop:Gaussian_marginal_moments} yields
\begin{align*}
&\E_{\batchResponse}[\Gaussian(\obs \given \emMean[\Naugment]^{\ParBatch,\batchResponse}(\Par), 
\covNoise/2 + \emVar[\Naugment](\Par; \ParBatch))] \\
&= \Gaussian(\obs \given \emMean(\Par), 
[\covNoise/2 + \emVar[\Naugment](\Par;\ParBatch)] + [\emVar(\Par) - \emVar[\Naugment](\Par; \ParBatch)]) \\
&= \Gaussian(\obs \given \emMean(\Par), \covNoise/2 + \emVar(\Par))
\end{align*}
and
\begin{align*}
&\E_{\batchResponse}[\Gaussian(\obs \given \emMean[\Naugment]^{\ParBatch,\batchResponse}(\Par), 
[\covNoise + \emVar[\Naugment](\Par; \ParBatch)]/2)] \\
&= \Gaussian(\obs \given \emMean(\Par), [\covNoise + \emVar[\Naugment](\Par; \ParBatch)]/2 + 
[\emVar(\Par) - \emVar[\Naugment](\Par; \ParBatch)]) \\
&= \Gaussian(\obs \given \emMean(\Par), [\covNoise/2 + \emVar(\Par)] - \emVar[\Naugment](\Par; \ParBatch)/2).
\end{align*}
We now rearrange the covariances of the above expressions to obtain 
\begin{align*}
&\covNoise/2 + \emVar(\Par) = \CovComb(\Par) - \covNoise/2 \\
&[\covNoise/2 + \emVar(\Par)] - \emVar[\Naugment](\Par; \ParBatch)/2
= \CovComb(\Par) - \CovComb[\Naugment](\Par) /2. 
\end{align*}
\end{proof}

\subsection{Log-Density Emulation}
The log-likelihood emulation setting yields the unnormalized posterior surrogate
$\jointEm(\Par) = \priorDens(\Par)\Exp{\targetEm(\Par)}$. This is a
log-normal process and thus the mean and variance in \Cref{ex:ldens-em}
follow immediately from the known forms of log-normal moments. We now derive 
the form of the expected conditional variance criterion in \Cref{ex:evar-ldens}.

\begin{proposition}
Let $\jointDens(\Par; \targetEm[\Naugment]^{\ParBatch,\batchResponse}) = 
\priorDens(\Par)\Exp{\targetEm[\Naugment]^{\ParBatch,\batchResponse}(\Par)}$ denote the 
unnormalized posterior surrogate induced by updating the emulator $\targetEm$ with new training
data $\{\ParBatch, \batchResponse\}$. Under the assumption 
$\batchResponse \sim \Gaussian(\emMean(\ParBatch), \emVar(\ParBatch))$, the expected conditional variance
is given by
\begin{align}
\E_{\batchResponse} \left\{\Var[\jointDens(\Par; \targetEm[\Naugment]^{\ParBatch,\batchResponse}) \given \ParBatch, \batchResponse]\right\}
&= \Var_{\emDist}\left[\jointEm(\Par) \given \simObsPredProcess(\ParBatch) = \emMean[\Ndesign](\ParBatch) \right] \varInflation_{\Ndesign}(\Par; \ParBatch),
\label{eq:evar-llik}
\end{align}
where
\begin{equation}
\varInflation_{\Ndesign}(\Par; \ParBatch)
= \Exp{2\left(\emVar[\Ndesign](\Par) - \emVar[\Naugment](\Par; \ParBatch)\right)}.
\end{equation}
The first term in \Cref{eq:evar-llik} is computed by conditioning the GP on 
$\{\ParBatch, \emMean[\Ndesign](\ParBatch)\}$ and then evaluating the log-normal variance of the resulting
unnormalized posterior surrogate. See \Cref{ex:evar-ldens} for the explicit expression.
\end{proposition}

\begin{proof}
After pulling out the constant $\priorDens(\Par)$, note that 
\begin{equation}
\Exp{\targetEm(\Par)} \given [\simObsPredProcess(\ParBatch) = \batchResponse] \sim 
\lognormal(\emMean[\Naugment]^{\ParBatch,\batchResponse}(\Par), \emVar[\Naugment](\Par; \ParBatch)).
\end{equation}
Applying the formula for a log-normal variance yields
\begin{align}
\Var[\Exp{\targetEm(\Par)} \given \simObsPredProcess(\ParBatch) = \batchResponse]
&= \cst \cdot \Exp{2\emMean[\Naugment]^{\ParBatch,\batchResponse}(\Par)}, \label{eq:formula_plug_in}
\end{align}
where 
\begin{align*}
\cst \Def \left[\Exp{\emVar[\Naugment](\Par; \ParBatch)} -1 \right] \Exp{\emVar[\Naugment](\Par; \ParBatch)}
\end{align*}
is not a function of the random variable $\batchResponse$. Applying \Cref{lemma:pred-mean-dist} gives
\begin{align*}
\Exp{2\emMean[\Naugment]^{\ParBatch,\batchResponse}(\Par)}
&\sim \lognormal(2\emMean(\Par), 4[\emVar(\Par) - \emVar[\Naugment](\Par; \ParBatch)]),
\end{align*}
which has the log-normal mean 
\begin{align*}
\E_{\batchResponse}[\Exp{2\emMean[\Naugment]^{\ParBatch,\batchResponse}(\Par)}]
&= \Exp{2\emMean(\Par)} \Exp{2[\emVar(\Par) - \emVar[\Naugment](\Par; \ParBatch)]} \\
&=  \Exp{2\emMean(\Par)} \varInflation(\Par; \ParBatch).
\end{align*}
Plugging this expression back into \Cref{eq:formula_plug_in} gives 
\begin{align*}
\E_{\batchResponse} \Var[\Exp{\targetEm(\Par)} \given \simObsPredProcess(\ParBatch) = \batchResponse]
&= \cst \cdot \Exp{2\emMean(\Par)} \varInflation(\Par; \ParBatch) \\
&= \Var_{\emDist}\left[\jointEm(\Par) \given \simObsProcess(\ParBatch) = \emMean[\Ndesign](\ParBatch) \right] \varInflation_{\Ndesign}(\Par; \ParBatch).
\end{align*}
\end{proof}

\bibliographystyle{imsart-authoryear}
\bibliography{references}

\begin{thebibliography}{176}

\bibitem[\protect\citeauthoryear{Aakula, Fer and Vira}{2025}]{FerRNNHMC}
\begin{barticle}[author]
\bauthor{\bsnm{Aakula},~\bfnm{Viivi}\binits{V.}},
  \bauthor{\bsnm{Fer},~\bfnm{Istem}\binits{I.}} \AND
  \bauthor{\bsnm{Vira},~\bfnm{Julius}\binits{J.}}
(\byear{2025}).
\btitle{Emulator-based calibration of a dynamic grassland model using recurrent
  neural networks and Hamiltonian Monte Carlo}.
\bjournal{European Journal of Agronomy}
\bvolume{170}
\bpages{127769}.
\bdoi{https://doi.org/10.1016/j.eja.2025.127769}
\end{barticle}
\endbibitem

\bibitem[\protect\citeauthoryear{Abdar et~al.}{2021}]{UQDLReview}
\begin{barticle}[author]
\bauthor{\bsnm{Abdar},~\bfnm{Moloud}\binits{M.}},
  \bauthor{\bsnm{Pourpanah},~\bfnm{Farhad}\binits{F.}},
  \bauthor{\bsnm{Hussain},~\bfnm{Sadiq}\binits{S.}},
  \bauthor{\bsnm{Rezazadegan},~\bfnm{Dana}\binits{D.}},
  \bauthor{\bsnm{Liu},~\bfnm{Li}\binits{L.}},
  \bauthor{\bsnm{Ghavamzadeh},~\bfnm{Mohammad}\binits{M.}},
  \bauthor{\bsnm{Fieguth},~\bfnm{Paul}\binits{P.}},
  \bauthor{\bsnm{Cao},~\bfnm{Xiaochun}\binits{X.}},
  \bauthor{\bsnm{Khosravi},~\bfnm{Abbas}\binits{A.}},
  \bauthor{\bsnm{Acharya},~\bfnm{U.~Rajendra}\binits{U.~R.}},
  \bauthor{\bsnm{Makarenkov},~\bfnm{Vladimir}\binits{V.}} \AND
  \bauthor{\bsnm{Nahavandi},~\bfnm{Saeid}\binits{S.}}
(\byear{2021}).
\btitle{A review of uncertainty quantification in deep learning: Techniques,
  applications and challenges}.
\bjournal{Inf. Fusion}
\bvolume{76}
\bpages{243–297}.
\bdoi{10.1016/j.inffus.2021.05.008}
\end{barticle}
\endbibitem

\bibitem[\protect\citeauthoryear{Alawieh, Goodman and
  Bell}{2020}]{AlawiehIterativeGP}
\begin{barticle}[author]
\bauthor{\bsnm{Alawieh},~\bfnm{Leen}\binits{L.}},
  \bauthor{\bsnm{Goodman},~\bfnm{Jonathan}\binits{J.}} \AND
  \bauthor{\bsnm{Bell},~\bfnm{John~B.}\binits{J.~B.}}
(\byear{2020}).
\btitle{Iterative construction of Gaussian process surrogate models for
  Bayesian inference}.
\bjournal{Journal of Statistical Planning and Inference}
\bvolume{207}
\bpages{55-72}.
\bdoi{https://doi.org/10.1016/j.jspi.2019.11.002}
\end{barticle}
\endbibitem

\bibitem[\protect\citeauthoryear{Alquier et~al.}{2016}]{noisyMCMC}
\begin{barticle}[author]
\bauthor{\bsnm{Alquier},~\bfnm{Pierre}\binits{P.}},
  \bauthor{\bsnm{Friel},~\bfnm{Nial}\binits{N.}},
  \bauthor{\bsnm{Everitt},~\bfnm{Richard}\binits{R.}} \AND
  \bauthor{\bsnm{Boland},~\bfnm{Aidan}\binits{A.}}
(\byear{2016}).
\btitle{Noisy {M}onte {C}arlo: convergence of {M}arkov chains with approximate
  transition kernels}.
\bjournal{Statistics and Computing}
\bvolume{26}
\bpages{29--47}.
\bdoi{10.1007/s11222-014-9521-x}
\end{barticle}
\endbibitem

\bibitem[\protect\citeauthoryear{Andrieu, Doucet and
  Holenstein}{2010}]{ParticleMCMC}
\begin{barticle}[author]
\bauthor{\bsnm{Andrieu},~\bfnm{Christophe}\binits{C.}},
  \bauthor{\bsnm{Doucet},~\bfnm{Arnaud}\binits{A.}} \AND
  \bauthor{\bsnm{Holenstein},~\bfnm{Roman}\binits{R.}}
(\byear{2010}).
\btitle{Particle {M}arkov chain {M}onte {C}arlo methods}.
\bjournal{Journal of the Royal Statistical Society Series B: Statistical
  Methodology}
\bvolume{72}
\bpages{269--342}.
\bdoi{10.1111/j.1467-9868.2009.00736.x}
\end{barticle}
\endbibitem

\bibitem[\protect\citeauthoryear{Andrieu and
  Roberts}{2009}]{pseudoMarginalMCMC}
\begin{barticle}[author]
\bauthor{\bsnm{Andrieu},~\bfnm{Christophe}\binits{C.}} \AND
  \bauthor{\bsnm{Roberts},~\bfnm{Gareth~O.}\binits{G.~O.}}
(\byear{2009}).
\btitle{{The pseudo-marginal approach for efficient Monte Carlo computations}}.
\bjournal{The Annals of Statistics}
\bvolume{37}
\bpages{697 -- 725}.
\bdoi{10.1214/07-AOS574}
\end{barticle}
\endbibitem

\bibitem[\protect\citeauthoryear{Astudillo and Frazier}{2019}]{CompositeBO}
\begin{binproceedings}[author]
\bauthor{\bsnm{Astudillo},~\bfnm{Raul}\binits{R.}} \AND
  \bauthor{\bsnm{Frazier},~\bfnm{Peter}\binits{P.}}
(\byear{2019}).
\btitle{{B}ayesian Optimization of Composite Functions}.
In \bbooktitle{Proceedings of the 36th International Conference on Machine
  Learning}
(\beditor{\bfnm{Kamalika}\binits{K.}~\bsnm{Chaudhuri}} \AND
  \beditor{\bfnm{Ruslan}\binits{R.}~\bsnm{Salakhutdinov}}, eds.).
\bseries{Proceedings of Machine Learning Research}
\bvolume{97}
\bpages{354--363}.
\bpublisher{PMLR}.
\end{binproceedings}
\endbibitem

\bibitem[\protect\citeauthoryear{Badia et~al.}{2020}]{BadiaRL}
\begin{bmisc}[author]
\bauthor{\bsnm{Badia},~\bfnm{Adrià~Puigdomènech}\binits{A.~P.}},
  \bauthor{\bsnm{Sprechmann},~\bfnm{Pablo}\binits{P.}},
  \bauthor{\bsnm{Vitvitskyi},~\bfnm{Alex}\binits{A.}},
  \bauthor{\bsnm{Guo},~\bfnm{Daniel}\binits{D.}},
  \bauthor{\bsnm{Piot},~\bfnm{Bilal}\binits{B.}},
  \bauthor{\bsnm{Kapturowski},~\bfnm{Steven}\binits{S.}},
  \bauthor{\bsnm{Tieleman},~\bfnm{Olivier}\binits{O.}},
  \bauthor{\bsnm{Arjovsky},~\bfnm{Martín}\binits{M.}},
  \bauthor{\bsnm{Pritzel},~\bfnm{Alexander}\binits{A.}},
  \bauthor{\bsnm{Bolt},~\bfnm{Andew}\binits{A.}} \AND
  \bauthor{\bsnm{Blundell},~\bfnm{Charles}\binits{C.}}
(\byear{2020}).
\btitle{Never Give Up: Learning Directed Exploration Strategies}.
\end{bmisc}
\endbibitem

\bibitem[\protect\citeauthoryear{Bai, Teckentrup and
  Zygalakis}{2024}]{GP_PDE_priors}
\begin{barticle}[author]
\bauthor{\bsnm{Bai},~\bfnm{Ting}\binits{T.}},
  \bauthor{\bsnm{Teckentrup},~\bfnm{Aretha~L.}\binits{A.~L.}} \AND
  \bauthor{\bsnm{Zygalakis},~\bfnm{Konstantinos~C.}\binits{K.~C.}}
(\byear{2024}).
\btitle{Gaussian processes for Bayesian inverse problems associated with linear
  partial differential equations}.
\bjournal{Statistics and Computing}
\bvolume{34}.
\bdoi{10.1007/s11222-024-10452-2}
\end{barticle}
\endbibitem

\bibitem[\protect\citeauthoryear{Baker et~al.}{2022}]{stochasticComputerModels}
\begin{barticle}[author]
\bauthor{\bsnm{Baker},~\bfnm{Evan}\binits{E.}},
  \bauthor{\bsnm{Barbillon},~\bfnm{Pierre}\binits{P.}},
  \bauthor{\bsnm{Fadikar},~\bfnm{Arindam}\binits{A.}},
  \bauthor{\bsnm{Gramacy},~\bfnm{Robert~B.}\binits{R.~B.}},
  \bauthor{\bsnm{Herbei},~\bfnm{Radu}\binits{R.}},
  \bauthor{\bsnm{Higdon},~\bfnm{David}\binits{D.}},
  \bauthor{\bsnm{Huang},~\bfnm{Jiangeng}\binits{J.}},
  \bauthor{\bsnm{Johnson},~\bfnm{Leah~R.}\binits{L.~R.}},
  \bauthor{\bsnm{Ma},~\bfnm{Pulong}\binits{P.}},
  \bauthor{\bsnm{Mondal},~\bfnm{Anirban}\binits{A.}},
  \bauthor{\bsnm{Pires},~\bfnm{Bianica}\binits{B.}},
  \bauthor{\bsnm{Sacks},~\bfnm{Jerome}\binits{J.}} \AND
  \bauthor{\bsnm{Sokolov},~\bfnm{Vadim}\binits{V.}}
(\byear{2022}).
\btitle{Analyzing Stochastic Computer Models: A Review with Opportunities}.
\bjournal{Statistical Science}
\bvolume{37}
\bpages{64--89}.
\bdoi{10.1214/21-STS822}
\end{barticle}
\endbibitem

\bibitem[\protect\citeauthoryear{Balandat et~al.}{2020}]{botorch}
\begin{binproceedings}[author]
\bauthor{\bsnm{Balandat},~\bfnm{Maximilian}\binits{M.}},
  \bauthor{\bsnm{Karrer},~\bfnm{Brian}\binits{B.}},
  \bauthor{\bsnm{Jiang},~\bfnm{Daniel~R.}\binits{D.~R.}},
  \bauthor{\bsnm{Daulton},~\bfnm{Samuel}\binits{S.}},
  \bauthor{\bsnm{Letham},~\bfnm{Benjamin}\binits{B.}},
  \bauthor{\bsnm{Wilson},~\bfnm{Andrew~Gordon}\binits{A.~G.}} \AND
  \bauthor{\bsnm{Bakshy},~\bfnm{Eytan}\binits{E.}}
(\byear{2020}).
\btitle{{BoTorch: A Framework for Efficient Monte-Carlo Bayesian
  Optimization}}.
In \bbooktitle{Advances in Neural Information Processing Systems 33}.
\end{binproceedings}
\endbibitem

\bibitem[\protect\citeauthoryear{Bastos and
  O’Hagan}{2009}]{GPEmulatorDiagnostics}
\begin{barticle}[author]
\bauthor{\bsnm{Bastos},~\bfnm{Leonardo~S.}\binits{L.~S.}} \AND
  \bauthor{\bsnm{O’Hagan},~\bfnm{Anthony}\binits{A.}}
(\byear{2009}).
\btitle{Diagnostics for Gaussian Process Emulators}.
\bjournal{Technometrics}
\bvolume{51}
\bpages{425--438}.
\bdoi{10.1198/TECH.2009.08019}
\end{barticle}
\endbibitem

\bibitem[\protect\citeauthoryear{Bect, Bachoc and
  Ginsbourger}{2019}]{supermartingaleSUR}
\begin{barticle}[author]
\bauthor{\bsnm{Bect},~\bfnm{Julien}\binits{J.}},
  \bauthor{\bsnm{Bachoc},~\bfnm{Fran{\c{c}}ois}\binits{F.}} \AND
  \bauthor{\bsnm{Ginsbourger},~\bfnm{David}\binits{D.}}
(\byear{2019}).
\btitle{{A supermartingale approach to Gaussian process based sequential design
  of experiments}}.
\bjournal{Bernoulli}
\bvolume{25}
\bpages{2883 -- 2919}.
\bdoi{10.3150/18-BEJ1074}
\end{barticle}
\endbibitem

\bibitem[\protect\citeauthoryear{Bect et~al.}{2011}]{BectSUR}
\begin{barticle}[author]
\bauthor{\bsnm{Bect},~\bfnm{Julien}\binits{J.}},
  \bauthor{\bsnm{Ginsbourger},~\bfnm{David}\binits{D.}},
  \bauthor{\bsnm{Li},~\bfnm{Ling}\binits{L.}},
  \bauthor{\bsnm{Picheny},~\bfnm{Victor}\binits{V.}} \AND
  \bauthor{\bsnm{Vazquez},~\bfnm{Emmanuel}\binits{E.}}
(\byear{2011}).
\btitle{Sequential design of computer experiments for the estimation of a
  probability of failure}.
\bjournal{Statistics and Computing}
\bvolume{22}
\bpages{773-793}.
\bdoi{10.1007/s11222-011-9241-4}
\end{barticle}
\endbibitem

\bibitem[\protect\citeauthoryear{Bharti et~al.}{2025}]{costAwareSBI}
\begin{binproceedings}[author]
\bauthor{\bsnm{Bharti},~\bfnm{Ayush}\binits{A.}},
  \bauthor{\bsnm{Huang},~\bfnm{Daolang}\binits{D.}},
  \bauthor{\bsnm{Kaski},~\bfnm{Samuel}\binits{S.}} \AND
  \bauthor{\bsnm{Briol},~\bfnm{Francois-Xavier}\binits{F.-X.}}
(\byear{2025}).
\btitle{Cost-aware simulation-based inference}.
In \bbooktitle{Proceedings of The 28th International Conference on Artificial
  Intelligence and Statistics}
(\beditor{\bfnm{Yingzhen}\binits{Y.}~\bsnm{Li}},
  \beditor{\bfnm{Stephan}\binits{S.}~\bsnm{Mandt}},
  \beditor{\bfnm{Shipra}\binits{S.}~\bsnm{Agrawal}} \AND
  \beditor{\bfnm{Emtiyaz}\binits{E.}~\bsnm{Khan}}, eds.).
\bseries{Proceedings of Machine Learning Research}
\bvolume{258}
\bpages{28--36}.
\bpublisher{PMLR}.
\end{binproceedings}
\endbibitem

\bibitem[\protect\citeauthoryear{Bhattacharyya}{2020}]{BayesianPCE1}
\begin{barticle}[author]
\bauthor{\bsnm{Bhattacharyya},~\bfnm{Biswarup}\binits{B.}}
(\byear{2020}).
\btitle{Global sensitivity analysis: A Bayesian learning based polynomial chaos
  approach}.
\bjournal{Journal of Computational Physics}
\bvolume{415}
\bpages{109539}.
\bdoi{https://doi.org/10.1016/j.jcp.2020.109539}
\end{barticle}
\endbibitem

\bibitem[\protect\citeauthoryear{Bilionis and
  Zabaras}{2013}]{BilionisBayesSurrogates}
\begin{barticle}[author]
\bauthor{\bsnm{Bilionis},~\bfnm{Ilias}\binits{I.}} \AND
  \bauthor{\bsnm{Zabaras},~\bfnm{Nicholas}\binits{N.}}
(\byear{2013}).
\btitle{Solution of inverse problems with limited forward solver evaluations: a
  Bayesian perspective}.
\bjournal{Inverse Problems}
\bvolume{30 015004}.
\bdoi{10.1088/0266-5611/30/1/015004}
\end{barticle}
\endbibitem

\bibitem[\protect\citeauthoryear{Binois et~al.}{2018}]{Binois_2018}
\begin{barticle}[author]
\bauthor{\bsnm{Binois},~\bfnm{Mickaël}\binits{M.}},
  \bauthor{\bsnm{Huang},~\bfnm{Jiangeng}\binits{J.}},
  \bauthor{\bsnm{Gramacy},~\bfnm{Robert~B.}\binits{R.~B.}} \AND
  \bauthor{\bsnm{Ludkovski},~\bfnm{Mike}\binits{M.}}
(\byear{2018}).
\btitle{Replication or Exploration? Sequential Design for Stochastic Simulation
  Experiments}.
\bjournal{Technometrics}
\bvolume{61}
\bpages{7–23}.
\bdoi{10.1080/00401706.2018.1469433}
\end{barticle}
\endbibitem

\bibitem[\protect\citeauthoryear{Bliznyuk, Ruppert and
  Shoemaker}{2012}]{bliznyuk2012}
\begin{barticle}[author]
\bauthor{\bsnm{Bliznyuk},~\bfnm{Nikolay}\binits{N.}},
  \bauthor{\bsnm{Ruppert},~\bfnm{David}\binits{D.}} \AND
  \bauthor{\bsnm{Shoemaker},~\bfnm{Christine~A.}\binits{C.~A.}}
(\byear{2012}).
\btitle{Local Derivative-Free Approximation of Computationally Expensive
  Posterior Densities}.
\bjournal{Journal of Computational and Graphical Statistics}
\bvolume{21}
\bpages{676--695}.
\bnote{PMID: 29861616; PMCID: PMC5978778}.
\bdoi{10.1080/10618600.2012.681255}
\end{barticle}
\endbibitem

\bibitem[\protect\citeauthoryear{Bliznyuk et~al.}{2008}]{llikRBF}
\begin{barticle}[author]
\bauthor{\bsnm{Bliznyuk},~\bfnm{Nikolay}\binits{N.}},
  \bauthor{\bsnm{Ruppert},~\bfnm{David}\binits{D.}},
  \bauthor{\bsnm{Shoemaker},~\bfnm{Christine}\binits{C.}},
  \bauthor{\bsnm{Regis},~\bfnm{Rommel}\binits{R.}},
  \bauthor{\bsnm{Wild},~\bfnm{Stefan}\binits{S.}} \AND
  \bauthor{\bsnm{Mugunthan},~\bfnm{Pradeep}\binits{P.}}
(\byear{2008}).
\btitle{Bayesian Calibration and Uncertainty Analysis for Computationally
  Expensive Models Using Optimization and Radial Basis Function Approximation}.
\bjournal{Journal of Computational and Graphical Statistics}
\bvolume{17}
\bpages{270-294}.
\bdoi{10.1198/106186008X320681}
\end{barticle}
\endbibitem

\bibitem[\protect\citeauthoryear{Bonilla, Chai and
  Williams}{2007}]{MultiTaskGP}
\begin{binproceedings}[author]
\bauthor{\bsnm{Bonilla},~\bfnm{Edwin~V}\binits{E.~V.}},
  \bauthor{\bsnm{Chai},~\bfnm{Kian}\binits{K.}} \AND
  \bauthor{\bsnm{Williams},~\bfnm{Christopher}\binits{C.}}
(\byear{2007}).
\btitle{Multi-task Gaussian Process Prediction}.
In \bbooktitle{Advances in Neural Information Processing Systems}
(\beditor{\bfnm{J.}\binits{J.}~\bsnm{Platt}},
  \beditor{\bfnm{D.}\binits{D.}~\bsnm{Koller}},
  \beditor{\bfnm{Y.}\binits{Y.}~\bsnm{Singer}} \AND
  \beditor{\bfnm{S.}\binits{S.}~\bsnm{Roweis}}, eds.)
\bvolume{20}.
\bpublisher{Curran Associates, Inc.}
\end{binproceedings}
\endbibitem

\bibitem[\protect\citeauthoryear{Booth, Cooper and
  Gramacy}{2024}]{nonstationaryGP}
\begin{bmisc}[author]
\bauthor{\bsnm{Booth},~\bfnm{Annie~S.}\binits{A.~S.}},
  \bauthor{\bsnm{Cooper},~\bfnm{Andrew}\binits{A.}} \AND
  \bauthor{\bsnm{Gramacy},~\bfnm{Robert~B.}\binits{R.~B.}}
(\byear{2024}).
\btitle{Nonstationary Gaussian Process Surrogates}.
\bdoi{10.48550/arXiv.2305.19242}
\end{bmisc}
\endbibitem

\bibitem[\protect\citeauthoryear{Bradford and Imsland}{2021}]{PCEGP2}
\begin{bmisc}[author]
\bauthor{\bsnm{Bradford},~\bfnm{Eric}\binits{E.}} \AND
  \bauthor{\bsnm{Imsland},~\bfnm{Lars}\binits{L.}}
(\byear{2021}).
\btitle{Combining Gaussian processes and polynomial chaos expansions for
  stochastic nonlinear model predictive control}.
\bdoi{10.48550/arXiv.2103.05441}
\end{bmisc}
\endbibitem

\bibitem[\protect\citeauthoryear{Briol et~al.}{2017}]{briol2017sampling}
\begin{binproceedings}[author]
\bauthor{\bsnm{Briol},~\bfnm{Fran{\c{c}}ois-Xavier}\binits{F.-X.}},
  \bauthor{\bsnm{Oates},~\bfnm{Chris~J.}\binits{C.~J.}},
  \bauthor{\bsnm{Cockayne},~\bfnm{Jon}\binits{J.}},
  \bauthor{\bsnm{Chen},~\bfnm{Wilson~Ye}\binits{W.~Y.}} \AND
  \bauthor{\bsnm{Girolami},~\bfnm{Mark}\binits{M.}}
(\byear{2017}).
\btitle{On the Sampling Problem for Kernel Quadrature}.
In \bbooktitle{Proceedings of the 34th International Conference on Machine
  Learning}
(\beditor{\bfnm{Doina}\binits{D.}~\bsnm{Precup}} \AND
  \beditor{\bfnm{Yee~Whye}\binits{Y.~W.}~\bsnm{Teh}}, eds.).
\bseries{Proceedings of Machine Learning Research}
\bvolume{70}
\bpages{586--595}.
\bpublisher{PMLR}.
\end{binproceedings}
\endbibitem

\bibitem[\protect\citeauthoryear{Brooks et~al.}{2011}]{mcmcHandbook}
\begin{bbook}[author]
\beditor{\bsnm{Brooks},~\bfnm{Steve}\binits{S.}},
  \beditor{\bsnm{Gelman},~\bfnm{Andrew}\binits{A.}},
  \beditor{\bsnm{Jones},~\bfnm{Galin}\binits{G.}} \AND
  \beditor{\bsnm{Meng},~\bfnm{Xiao-Li}\binits{X.-L.}}, eds.
(\byear{2011}).
\btitle{Handbook of Markov Chain Monte Carlo},
\bedition{1st} ed.
\bpublisher{Chapman and Hall/CRC}.
\bdoi{10.1201/b10905}
\end{bbook}
\endbibitem

\bibitem[\protect\citeauthoryear{Bürkner et~al.}{2023}]{BayesianPCE2}
\begin{barticle}[author]
\bauthor{\bsnm{Bürkner},~\bfnm{Paul-Christian}\binits{P.-C.}},
  \bauthor{\bsnm{Kröker},~\bfnm{Ilja}\binits{I.}},
  \bauthor{\bsnm{Oladyshkin},~\bfnm{Sergey}\binits{S.}} \AND
  \bauthor{\bsnm{Nowak},~\bfnm{Wolfgang}\binits{W.}}
(\byear{2023}).
\btitle{A fully Bayesian sparse polynomial chaos expansion approach with joint
  priors on the coefficients and global selection of terms}.
\bjournal{Journal of Computational Physics}
\bvolume{488}
\bpages{112210}.
\bdoi{https://doi.org/10.1016/j.jcp.2023.112210}
\end{barticle}
\endbibitem

\bibitem[\protect\citeauthoryear{Carmona}{2023}]{semiModular}
\begin{bphdthesis}[author]
\bauthor{\bsnm{Carmona},~\bfnm{C.~U.}\binits{C.~U.}}
(\byear{2023}).
\btitle{Bayesian Semi-Modular Inference: theory and methods for inference in
  multi-modular settings under model misspecification},
\btype{PhD thesis},
\bpublisher{University of Oxford}.
\end{bphdthesis}
\endbibitem

\bibitem[\protect\citeauthoryear{C\'{e}rou, H\'{e}as and
  Rousset}{2025}]{adaptiveReducedTempering}
\begin{barticle}[author]
\bauthor{\bsnm{C\'{e}rou},~\bfnm{Fr\'{e}d\'{e}ric}\binits{F.}},
  \bauthor{\bsnm{H\'{e}as},~\bfnm{Patrick}\binits{P.}} \AND
  \bauthor{\bsnm{Rousset},~\bfnm{Mathias}\binits{M.}}
(\byear{2025}).
\btitle{Adaptive Reduced Tempering for Bayesian Inverse Problems and Rare Event
  Simulation}.
\bjournal{SIAM/ASA Journal on Uncertainty Quantification}
\bvolume{13}
\bpages{2022-2053}.
\bdoi{10.1137/24M170510X}
\end{barticle}
\endbibitem

\bibitem[\protect\citeauthoryear{Chevalier}{2013}]{SURThesis}
\begin{bphdthesis}[author]
\bauthor{\bsnm{Chevalier},~\bfnm{Clément}\binits{C.}}
(\byear{2013}).
\btitle{Fast uncertainty reduction strategies relying on Gaussian process
  models},
\btype{PhD thesis}.
\end{bphdthesis}
\endbibitem

\bibitem[\protect\citeauthoryear{Chevalier and
  Ginsbourger}{2013}]{Chevalier2013}
\begin{binproceedings}[author]
\bauthor{\bsnm{Chevalier},~\bfnm{Cl{\'e}ment}\binits{C.}} \AND
  \bauthor{\bsnm{Ginsbourger},~\bfnm{David}\binits{D.}}
(\byear{2013}).
\btitle{Fast Computation of the Multi-Points Expected Improvement with
  Applications in Batch Selection}.
In \bbooktitle{Learning and Intelligent Optimization}
(\beditor{\bfnm{Giuseppe}\binits{G.}~\bsnm{Nicosia}} \AND
  \beditor{\bfnm{Panos}\binits{P.}~\bsnm{Pardalos}}, eds.)
\bpages{59--69}.
\bpublisher{Springer Berlin Heidelberg}, \baddress{Berlin, Heidelberg}.
\end{binproceedings}
\endbibitem

\bibitem[\protect\citeauthoryear{Christen and Fox}{2005}]{DelayedAcceptance}
\begin{barticle}[author]
\bauthor{\bsnm{Christen},~\bfnm{J.~Andrés}\binits{J.~A.}} \AND
  \bauthor{\bsnm{Fox},~\bfnm{Colin}\binits{C.}}
(\byear{2005}).
\btitle{Markov chain Monte Carlo Using an Approximation}.
\bjournal{Journal of Computational and Graphical Statistics}
\bvolume{14}
\bpages{795--810}.
\bdoi{10.1198/106186005X76983}
\end{barticle}
\endbibitem

\bibitem[\protect\citeauthoryear{Cleary et~al.}{2021}]{CES}
\begin{barticle}[author]
\bauthor{\bsnm{Cleary},~\bfnm{Emmet}\binits{E.}},
  \bauthor{\bsnm{Garbuno-Inigo},~\bfnm{Alfredo}\binits{A.}},
  \bauthor{\bsnm{Lan},~\bfnm{Shiwei}\binits{S.}},
  \bauthor{\bsnm{Schneider},~\bfnm{Tapio}\binits{T.}} \AND
  \bauthor{\bsnm{Stuart},~\bfnm{Andrew~M.}\binits{A.~M.}}
(\byear{2021}).
\btitle{Calibrate, emulate, sample}.
\bjournal{Journal of Computational Physics}
\bvolume{424}
\bpages{109716}.
\bdoi{10.1016/j.jcp.2020.109716}
\end{barticle}
\endbibitem

\bibitem[\protect\citeauthoryear{Cole et~al.}{2023}]{cole2021entropybased}
\begin{barticle}[author]
\bauthor{\bsnm{Cole},~\bfnm{D.~Austin}\binits{D.~A.}},
  \bauthor{\bsnm{Gramacy},~\bfnm{Robert~B.}\binits{R.~B.}},
  \bauthor{\bsnm{Warner},~\bfnm{James~E.}\binits{J.~E.}},
  \bauthor{\bsnm{Bomarito},~\bfnm{Geoffrey~F.}\binits{G.~F.}},
  \bauthor{\bsnm{Leser},~\bfnm{Patrick~E.}\binits{P.~E.}} \AND
  \bauthor{\bsnm{Leser},~\bfnm{William~P.}\binits{W.~P.}}
(\byear{2023}).
\btitle{Entropy-based adaptive design for contour finding and estimating
  reliability}.
\bjournal{Journal of Quality Technology}
\bvolume{55}
\bpages{43--60}.
\bdoi{10.1080/00224065.2022.2053795}
\end{barticle}
\endbibitem

\bibitem[\protect\citeauthoryear{Conrad et~al.}{2016}]{ConradLocalExactMCMC}
\begin{barticle}[author]
\bauthor{\bsnm{Conrad},~\bfnm{Patrick~R.}\binits{P.~R.}},
  \bauthor{\bsnm{Marzouk},~\bfnm{Youssef~M.}\binits{Y.~M.}},
  \bauthor{\bsnm{Pillai},~\bfnm{Natesh~S.}\binits{N.~S.}} \AND
  \bauthor{\bsnm{and},~\bfnm{Aaron~Smith}\binits{A.~S.}}
(\byear{2016}).
\btitle{Accelerating Asymptotically Exact MCMC for Computationally Intensive
  Models via Local Approximations}.
\bjournal{Journal of the American Statistical Association}
\bvolume{111}
\bpages{1591--1607}.
\bdoi{10.1080/01621459.2015.1096787}
\end{barticle}
\endbibitem

\bibitem[\protect\citeauthoryear{Conti and
  O’Hagan}{2010}]{Bayesian_emulation_dynamic}
\begin{barticle}[author]
\bauthor{\bsnm{Conti},~\bfnm{Stefano}\binits{S.}} \AND
  \bauthor{\bsnm{O’Hagan},~\bfnm{Anthony}\binits{A.}}
(\byear{2010}).
\btitle{Bayesian emulation of complex multi-output and dynamic computer
  models}.
\bjournal{Journal of Statistical Planning and Inference}
\bvolume{140}
\bpages{640-651}.
\bdoi{https://doi.org/10.1016/j.jspi.2009.08.006}
\end{barticle}
\endbibitem

\bibitem[\protect\citeauthoryear{Conti et~al.}{2009}]{GP_dynamic_emulation}
\begin{barticle}[author]
\bauthor{\bsnm{Conti},~\bfnm{S.}\binits{S.}},
  \bauthor{\bsnm{Gosling},~\bfnm{J.~P.}\binits{J.~P.}},
  \bauthor{\bsnm{Oakley},~\bfnm{J.~E.}\binits{J.~E.}} \AND
  \bauthor{\bsnm{O'Hagan},~\bfnm{A.}\binits{A.}}
(\byear{2009}).
\btitle{{Gaussian process emulation of dynamic computer codes}}.
\bjournal{Biometrika}
\bvolume{96}
\bpages{663-676}.
\bdoi{10.1093/biomet/asp028}
\end{barticle}
\endbibitem

\bibitem[\protect\citeauthoryear{Cranmer, Brehmer and
  Louppe}{2020}]{frontierSBI}
\begin{barticle}[author]
\bauthor{\bsnm{Cranmer},~\bfnm{Kyle}\binits{K.}},
  \bauthor{\bsnm{Brehmer},~\bfnm{Johann}\binits{J.}} \AND
  \bauthor{\bsnm{Louppe},~\bfnm{Gilles}\binits{G.}}
(\byear{2020}).
\btitle{The frontier of simulation-based inference}.
\bjournal{Proceedings of the National Academy of Sciences}
\bvolume{117}
\bpages{30055--30062}.
\bdoi{10.1073/pnas.1912789117}
\end{barticle}
\endbibitem

\bibitem[\protect\citeauthoryear{Cranmer, Pavez and Louppe}{2015}]{CranmerNRE}
\begin{barticle}[author]
\bauthor{\bsnm{Cranmer},~\bfnm{Kyle}\binits{K.}},
  \bauthor{\bsnm{Pavez},~\bfnm{Juan}\binits{J.}} \AND
  \bauthor{\bsnm{Louppe},~\bfnm{Gilles}\binits{G.}}
(\byear{2015}).
\btitle{{Approximating Likelihood Ratios with Calibrated Discriminative
  Classifiers}}.
\end{barticle}
\endbibitem

\bibitem[\protect\citeauthoryear{Cui et~al.}{2014}]{CuiLIS}
\begin{barticle}[author]
\bauthor{\bsnm{Cui},~\bfnm{T}\binits{T.}},
  \bauthor{\bsnm{Martin},~\bfnm{J}\binits{J.}},
  \bauthor{\bsnm{Marzouk},~\bfnm{Y~M}\binits{Y.~M.}},
  \bauthor{\bsnm{Solonen},~\bfnm{A}\binits{A.}} \AND
  \bauthor{\bsnm{Spantini},~\bfnm{A}\binits{A.}}
(\byear{2014}).
\btitle{Likelihood-informed dimension reduction for nonlinear inverse
  problems}.
\bjournal{Inverse Problems}
\bvolume{30}
\bpages{114015}.
\bdoi{10.1088/0266-5611/30/11/114015}
\end{barticle}
\endbibitem

\bibitem[\protect\citeauthoryear{Dagon et~al.}{2020}]{DagonCLM}
\begin{barticle}[author]
\bauthor{\bsnm{Dagon},~\bfnm{K.}\binits{K.}},
  \bauthor{\bsnm{Sanderson},~\bfnm{B.~M.}\binits{B.~M.}},
  \bauthor{\bsnm{Fisher},~\bfnm{R.~A.}\binits{R.~A.}} \AND
  \bauthor{\bsnm{Lawrence},~\bfnm{D.~M.}\binits{D.~M.}}
(\byear{2020}).
\btitle{A machine learning approach to emulation and biophysical parameter
  estimation with the Community Land Model, version 5}.
\bjournal{Advances in Statistical Climatology, Meteorology and Oceanography}
\bvolume{6}
\bpages{223--244}.
\bdoi{10.5194/ascmo-6-223-2020}
\end{barticle}
\endbibitem

\bibitem[\protect\citeauthoryear{Dax et~al.}{2021}]{npeAstrophysics}
\begin{barticle}[author]
\bauthor{\bsnm{Dax},~\bfnm{Maximilian}\binits{M.}},
  \bauthor{\bsnm{Green},~\bfnm{Stephen~R.}\binits{S.~R.}},
  \bauthor{\bsnm{Gair},~\bfnm{Jonathan}\binits{J.}},
  \bauthor{\bsnm{Macke},~\bfnm{Jakob~H.}\binits{J.~H.}},
  \bauthor{\bsnm{Buonanno},~\bfnm{Alessandra}\binits{A.}} \AND
  \bauthor{\bsnm{Sch\"olkopf},~\bfnm{Bernhard}\binits{B.}}
(\byear{2021}).
\btitle{Real-Time Gravitational Wave Science with Neural Posterior Estimation}.
\bjournal{Phys. Rev. Lett.}
\bvolume{127}
\bpages{241103}.
\bdoi{10.1103/PhysRevLett.127.241103}
\end{barticle}
\endbibitem

\bibitem[\protect\citeauthoryear{Deligiannidis, Doucet and Pitt}{2018}]{corrPM}
\begin{barticle}[author]
\bauthor{\bsnm{Deligiannidis},~\bfnm{George}\binits{G.}},
  \bauthor{\bsnm{Doucet},~\bfnm{Arnaud}\binits{A.}} \AND
  \bauthor{\bsnm{Pitt},~\bfnm{Michael~K.}\binits{M.~K.}}
(\byear{2018}).
\btitle{The correlated pseudo-marginal method}.
\bjournal{Journal of the Royal Statistical Society Series B: Statistical
  Methodology}
\bvolume{80}
\bpages{839--870}.
\bdoi{10.1111/rssb.12280}
\end{barticle}
\endbibitem

\bibitem[\protect\citeauthoryear{Dietzel and Reichert}{2014}]{emPostDens}
\begin{barticle}[author]
\bauthor{\bsnm{Dietzel},~\bfnm{A.}\binits{A.}} \AND
  \bauthor{\bsnm{Reichert},~\bfnm{P.}\binits{P.}}
(\byear{2014}).
\btitle{Bayesian inference of a lake water quality model by emulating its
  posterior density}.
\bjournal{Water Resources Research}
\bvolume{50}
\bpages{7626-7647}.
\bdoi{https://doi.org/10.1002/2012WR013086}
\end{barticle}
\endbibitem

\bibitem[\protect\citeauthoryear{Dinkel et~al.}{2024}]{quantileApprox}
\begin{barticle}[author]
\bauthor{\bsnm{Dinkel},~\bfnm{Maximilian}\binits{M.}},
  \bauthor{\bsnm{Geitner},~\bfnm{Carolin~M}\binits{C.~M.}},
  \bauthor{\bsnm{Robalo~Rei},~\bfnm{Gil}\binits{G.}},
  \bauthor{\bsnm{Nitzler},~\bfnm{Jonas}\binits{J.}} \AND
  \bauthor{\bsnm{Wall},~\bfnm{Wolfgang~A}\binits{W.~A.}}
(\byear{2024}).
\btitle{Solving Bayesian inverse problems with expensive likelihoods using
  constrained Gaussian processes and active learning}.
\bjournal{Inverse Problems}
\bvolume{40}
\bpages{095008}.
\bdoi{10.1088/1361-6420/ad5eb4}
\end{barticle}
\endbibitem

\bibitem[\protect\citeauthoryear{Doucet, de~Freitas and
  Gordon}{2001}]{DoucetSMC}
\begin{binbook}[author]
\bauthor{\bsnm{Doucet},~\bfnm{Arnaud}\binits{A.}}, \bauthor{\bparticle{de}
  \bsnm{Freitas},~\bfnm{Nando}\binits{N.}} \AND
  \bauthor{\bsnm{Gordon},~\bfnm{Neil}\binits{N.}}
(\byear{2001}).
\btitle{An Introduction to Sequential Monte Carlo Methods}
In \bbooktitle{Sequential Monte Carlo Methods in Practice}
\bpages{3--14}.
\bpublisher{Springer New York}, \baddress{New York, NY}.
\bdoi{10.1007/978-1-4757-3437-9_1}
\end{binbook}
\endbibitem

\bibitem[\protect\citeauthoryear{Doumont et~al.}{2025}]{highDimBO}
\begin{bmisc}[author]
\bauthor{\bsnm{Doumont},~\bfnm{Colin}\binits{C.}},
  \bauthor{\bsnm{Fan},~\bfnm{Donney}\binits{D.}},
  \bauthor{\bsnm{Maus},~\bfnm{Natalie}\binits{N.}},
  \bauthor{\bsnm{Gardner},~\bfnm{Jacob~R.}\binits{J.~R.}},
  \bauthor{\bsnm{Moss},~\bfnm{Henry}\binits{H.}} \AND
  \bauthor{\bsnm{Pleiss},~\bfnm{Geoff}\binits{G.}}
(\byear{2025}).
\btitle{We Still Don't Understand High-Dimensional Bayesian Optimization}.
\end{bmisc}
\endbibitem

\bibitem[\protect\citeauthoryear{Drovandi, Moores and
  Boys}{2018}]{DrovandiPMGP}
\begin{barticle}[author]
\bauthor{\bsnm{Drovandi},~\bfnm{Christopher~C.}\binits{C.~C.}},
  \bauthor{\bsnm{Moores},~\bfnm{Matthew~T.}\binits{M.~T.}} \AND
  \bauthor{\bsnm{Boys},~\bfnm{Richard~J.}\binits{R.~J.}}
(\byear{2018}).
\btitle{Accelerating pseudo-marginal MCMC using Gaussian processes}.
\bjournal{Computational Statistics \& Data Analysis}
\bvolume{118}
\bpages{1-17}.
\bdoi{https://doi.org/10.1016/j.csda.2017.09.002}
\end{barticle}
\endbibitem

\bibitem[\protect\citeauthoryear{Dunbar et~al.}{2021}]{idealizedGCM}
\begin{barticle}[author]
\bauthor{\bsnm{Dunbar},~\bfnm{Oliver R.~A.}\binits{O.~R.~A.}},
  \bauthor{\bsnm{Garbuno-Inigo},~\bfnm{Alfredo}\binits{A.}},
  \bauthor{\bsnm{Schneider},~\bfnm{Tapio}\binits{T.}} \AND
  \bauthor{\bsnm{Stuart},~\bfnm{Andrew~M.}\binits{A.~M.}}
(\byear{2021}).
\btitle{Calibration and Uncertainty Quantification of Convective Parameters in
  an Idealized GCM}.
\bjournal{Journal of Advances in Modeling Earth Systems}
\bvolume{13}
\bpages{e2020MS002454}.
\bnote{e2020MS002454 2020MS002454}.
\bdoi{https://doi.org/10.1029/2020MS002454}
\end{barticle}
\endbibitem

\bibitem[\protect\citeauthoryear{Dunbar et~al.}{2024}]{CESSoftware}
\begin{barticle}[author]
\bauthor{\bsnm{Dunbar},~\bfnm{Oliver R.~A.}\binits{O.~R.~A.}},
  \bauthor{\bsnm{Bieli},~\bfnm{Melanie}\binits{M.}},
  \bauthor{\bsnm{Garbuno-Iñigo},~\bfnm{Alfredo}\binits{A.}},
  \bauthor{\bsnm{Howland},~\bfnm{Michael}\binits{M.}}, \bauthor{\bparticle{de}
  \bsnm{Souza},~\bfnm{Andre~Nogueira}\binits{A.~N.}},
  \bauthor{\bsnm{Mansfield},~\bfnm{Laura~Anne}\binits{L.~A.}},
  \bauthor{\bsnm{Wagner},~\bfnm{Gregory~L.}\binits{G.~L.}} \AND
  \bauthor{\bsnm{Efrat-Henrici},~\bfnm{N.}\binits{N.}}
(\byear{2024}).
\btitle{CalibrateEmulateSample.jl: Accelerated Parametric Uncertainty
  Quantification}.
\bjournal{Journal of Open Source Software}
\bvolume{9}
\bpages{6372}.
\bdoi{10.21105/joss.06372}
\end{barticle}
\endbibitem

\bibitem[\protect\citeauthoryear{Durkan, Papamakarios and
  Murray}{2018}]{MurrayNLEvsNPE}
\begin{bmisc}[author]
\bauthor{\bsnm{Durkan},~\bfnm{Conor}\binits{C.}},
  \bauthor{\bsnm{Papamakarios},~\bfnm{George}\binits{G.}} \AND
  \bauthor{\bsnm{Murray},~\bfnm{Iain}\binits{I.}}
(\byear{2018}).
\btitle{Sequential Neural Methods for Likelihood-free Inference}.
\end{bmisc}
\endbibitem

\bibitem[\protect\citeauthoryear{Eriksson and
  Poloczek}{2021}]{ScalableConstrainedBO}
\begin{binproceedings}[author]
\bauthor{\bsnm{Eriksson},~\bfnm{David}\binits{D.}} \AND
  \bauthor{\bsnm{Poloczek},~\bfnm{Matthias}\binits{M.}}
(\byear{2021}).
\btitle{Scalable Constrained Bayesian Optimization}.
In \bbooktitle{Proceedings of The 24th International Conference on Artificial
  Intelligence and Statistics}
(\beditor{\bfnm{Arindam}\binits{A.}~\bsnm{Banerjee}} \AND
  \beditor{\bfnm{Kenji}\binits{K.}~\bsnm{Fukumizu}}, eds.).
\bseries{Proceedings of Machine Learning Research}
\bvolume{130}
\bpages{730--738}.
\bpublisher{PMLR}.
\end{binproceedings}
\endbibitem

\bibitem[\protect\citeauthoryear{Fadikar et~al.}{2018}]{FadikarAgentBased}
\begin{barticle}[author]
\bauthor{\bsnm{Fadikar},~\bfnm{Arindam}\binits{A.}},
  \bauthor{\bsnm{Higdon},~\bfnm{Dave}\binits{D.}},
  \bauthor{\bsnm{Chen},~\bfnm{Jiangzhuo}\binits{J.}},
  \bauthor{\bsnm{Lewis},~\bfnm{Bryan}\binits{B.}},
  \bauthor{\bsnm{Venkatramanan},~\bfnm{Srinivasan}\binits{S.}} \AND
  \bauthor{\bsnm{Marathe},~\bfnm{Madhav}\binits{M.}}
(\byear{2018}).
\btitle{Calibrating a Stochastic, Agent-Based Model Using Quantile-Based
  Emulation}.
\bjournal{SIAM/ASA Journal on Uncertainty Quantification}
\bvolume{6}
\bpages{1685-1706}.
\bdoi{10.1137/17M1161233}
\end{barticle}
\endbibitem

\bibitem[\protect\citeauthoryear{Fedorov}{1972}]{FederovExchange}
\begin{bbook}[author]
\bauthor{\bsnm{Fedorov},~\bfnm{Valerii}\binits{V.}}
(\byear{1972}).
\btitle{Theory of Optimal Experiments Designs}.
\end{bbook}
\endbibitem

\bibitem[\protect\citeauthoryear{Fer et~al.}{2018}]{FerEmulation}
\begin{barticle}[author]
\bauthor{\bsnm{Fer},~\bfnm{I.}\binits{I.}},
  \bauthor{\bsnm{Kelly},~\bfnm{R.}\binits{R.}},
  \bauthor{\bsnm{Moorcroft},~\bfnm{P.~R.}\binits{P.~R.}},
  \bauthor{\bsnm{Richardson},~\bfnm{A.~D.}\binits{A.~D.}},
  \bauthor{\bsnm{Cowdery},~\bfnm{E.~M.}\binits{E.~M.}} \AND
  \bauthor{\bsnm{Dietze},~\bfnm{M.~C.}\binits{M.~C.}}
(\byear{2018}).
\btitle{Linking big models to big data: efficient ecosystem model calibration
  through Bayesian model emulation}.
\bjournal{Biogeosciences}
\bvolume{15}
\bpages{5801--5830}.
\bdoi{10.5194/bg-15-5801-2018}
\end{barticle}
\endbibitem

\bibitem[\protect\citeauthoryear{Frazier and Nott}{2025}]{generalizedCut}
\begin{barticle}[author]
\bauthor{\bsnm{Frazier},~\bfnm{David~T.}\binits{D.~T.}} \AND
  \bauthor{\bsnm{Nott},~\bfnm{David~J.}\binits{D.~J.}}
(\byear{2025}).
\btitle{Cutting Feedback and Modularized Analyses in Generalized Bayesian
  Inference}.
\bjournal{Bayesian Analysis}
\bvolume{20}
\bpages{1647--1675}.
\end{barticle}
\endbibitem

\bibitem[\protect\citeauthoryear{Frazier et~al.}{2023}]{FrazierSL}
\begin{barticle}[author]
\bauthor{\bsnm{Frazier},~\bfnm{David~T.}\binits{D.~T.}},
  \bauthor{\bsnm{Nott},~\bfnm{David~J.}\binits{D.~J.}},
  \bauthor{\bsnm{Drovandi},~\bfnm{Christopher}\binits{C.}} \AND
  \bauthor{\bsnm{Kohn},~\bfnm{Robert}\binits{R.}}
(\byear{2023}).
\btitle{Bayesian Inference Using Synthetic Likelihood: Asymptotics and
  Adjustments}.
\bjournal{Journal of the American Statistical Association}
\bvolume{118}
\bpages{2821--2832}.
\bdoi{10.1080/01621459.2022.2086132}
\end{barticle}
\endbibitem

\bibitem[\protect\citeauthoryear{Gabry et~al.}{2019}]{VisBayesianWorkflow}
\begin{barticle}[author]
\bauthor{\bsnm{Gabry},~\bfnm{Jonah}\binits{J.}},
  \bauthor{\bsnm{Simpson},~\bfnm{Daniel}\binits{D.}},
  \bauthor{\bsnm{Vehtari},~\bfnm{Aki}\binits{A.}},
  \bauthor{\bsnm{Betancourt},~\bfnm{Michael}\binits{M.}} \AND
  \bauthor{\bsnm{Gelman},~\bfnm{Andrew}\binits{A.}}
(\byear{2019}).
\btitle{Visualization in {B}ayesian workflow}.
\bjournal{Journal of the Royal Statistical Society Series A: Statistics in
  Society}
\bvolume{182}
\bpages{389--402}.
\bdoi{10.1111/rssa.12378}
\end{barticle}
\endbibitem

\bibitem[\protect\citeauthoryear{Gal and Ghahramani}{2016}]{BayesianDropout}
\begin{binproceedings}[author]
\bauthor{\bsnm{Gal},~\bfnm{Yarin}\binits{Y.}} \AND
  \bauthor{\bsnm{Ghahramani},~\bfnm{Zoubin}\binits{Z.}}
(\byear{2016}).
\btitle{Dropout as a Bayesian Approximation: Representing Model Uncertainty in
  Deep Learning}.
In \bbooktitle{Proceedings of The 33rd International Conference on Machine
  Learning}
(\beditor{\bfnm{Maria~Florina}\binits{M.~F.}~\bsnm{Balcan}} \AND
  \beditor{\bfnm{Kilian~Q.}\binits{K.~Q.}~\bsnm{Weinberger}}, eds.).
\bseries{Proceedings of Machine Learning Research}
\bvolume{48}
\bpages{1050--1059}.
\bpublisher{PMLR}, \baddress{New York, New York, USA}.
\end{binproceedings}
\endbibitem

\bibitem[\protect\citeauthoryear{Garegnani}{2021}]{garegnani2021NoisyMCMC}
\begin{bmisc}[author]
\bauthor{\bsnm{Garegnani},~\bfnm{Giacomo}\binits{G.}}
(\byear{2021}).
\btitle{Sampling Methods for Bayesian Inference Involving Convergent Noisy
  Approximations of Forward Maps}.
\bdoi{10.48550/arXiv.2111.03491}
\end{bmisc}
\endbibitem

\bibitem[\protect\citeauthoryear{Garud, Karimi and
  Kraft}{2017}]{initDesignReview}
\begin{barticle}[author]
\bauthor{\bsnm{Garud},~\bfnm{Sushant~S.}\binits{S.~S.}},
  \bauthor{\bsnm{Karimi},~\bfnm{Iftekhar~A.}\binits{I.~A.}} \AND
  \bauthor{\bsnm{Kraft},~\bfnm{Markus}\binits{M.}}
(\byear{2017}).
\btitle{Design of computer experiments: A review}.
\bjournal{Computers \& Chemical Engineering}
\bvolume{106}
\bpages{71-95}.
\bnote{ESCAPE-26}.
\bdoi{https://doi.org/10.1016/j.compchemeng.2017.05.010}
\end{barticle}
\endbibitem

\bibitem[\protect\citeauthoryear{Gelman et~al.}{2013}]{gelmanBDA}
\begin{bbook}[author]
\bauthor{\bsnm{Gelman},~\bfnm{Andrew}\binits{A.}},
  \bauthor{\bsnm{Carlin},~\bfnm{John~B}\binits{J.~B.}},
  \bauthor{\bsnm{Stern},~\bfnm{Hal~S}\binits{H.~S.}},
  \bauthor{\bsnm{Dunson},~\bfnm{David~B}\binits{D.~B.}},
  \bauthor{\bsnm{Vehtari},~\bfnm{Aki}\binits{A.}} \AND
  \bauthor{\bsnm{Rubin},~\bfnm{Donald~B}\binits{D.~B.}}
(\byear{2013}).
\btitle{Bayesian Data Analysis},
\bedition{3rd} ed.
\bpublisher{CRC Press}, \baddress{Boca Raton, FL}.
\end{bbook}
\endbibitem

\bibitem[\protect\citeauthoryear{Gelman et~al.}{2020}]{BayesianWorkflow}
\begin{bmisc}[author]
\bauthor{\bsnm{Gelman},~\bfnm{Andrew}\binits{A.}},
  \bauthor{\bsnm{Vehtari},~\bfnm{Aki}\binits{A.}},
  \bauthor{\bsnm{Simpson},~\bfnm{Daniel}\binits{D.}},
  \bauthor{\bsnm{Margossian},~\bfnm{Charles~C.}\binits{C.~C.}},
  \bauthor{\bsnm{Carpenter},~\bfnm{Bob}\binits{B.}},
  \bauthor{\bsnm{Yao},~\bfnm{Yuling}\binits{Y.}},
  \bauthor{\bsnm{Kennedy},~\bfnm{Lauren}\binits{L.}},
  \bauthor{\bsnm{Gabry},~\bfnm{Jonah}\binits{J.}},
  \bauthor{\bsnm{Bürkner},~\bfnm{Paul-Christian}\binits{P.-C.}} \AND
  \bauthor{\bsnm{Modrák},~\bfnm{Martin}\binits{M.}}
(\byear{2020}).
\btitle{Bayesian Workflow}.
\bdoi{10.48550/arXiv.2011.01808}
\end{bmisc}
\endbibitem

\bibitem[\protect\citeauthoryear{Ginsbourger, Le~Riche and
  Carraro}{2010}]{Ginsbourger2010}
\begin{binbook}[author]
\bauthor{\bsnm{Ginsbourger},~\bfnm{David}\binits{D.}},
  \bauthor{\bsnm{Le~Riche},~\bfnm{Rodolphe}\binits{R.}} \AND
  \bauthor{\bsnm{Carraro},~\bfnm{Laurent}\binits{L.}}
(\byear{2010}).
\btitle{Kriging Is Well-Suited to Parallelize Optimization}
In \bbooktitle{Computational Intelligence in Expensive Optimization Problems}
\bpages{131--162}.
\bpublisher{Springer Berlin Heidelberg}, \baddress{Berlin, Heidelberg}.
\bdoi{10.1007/978-3-642-10701-6_6}
\end{binbook}
\endbibitem

\bibitem[\protect\citeauthoryear{Gjini et~al.}{2025}]{EKIRace}
\begin{bmisc}[author]
\bauthor{\bsnm{Gjini},~\bfnm{Rebecca}\binits{R.}},
  \bauthor{\bsnm{Morzfeld},~\bfnm{Matthias}\binits{M.}},
  \bauthor{\bsnm{Dunbar},~\bfnm{Oliver R.~A.}\binits{O.~R.~A.}} \AND
  \bauthor{\bsnm{Schneider},~\bfnm{Tapio}\binits{T.}}
(\byear{2025}).
\btitle{The Ensemble Kalman Inversion Race}.
\end{bmisc}
\endbibitem

\bibitem[\protect\citeauthoryear{Gneiting and Raftery}{2007}]{scoringRules}
\begin{barticle}[author]
\bauthor{\bsnm{Gneiting},~\bfnm{Tilmann}\binits{T.}} \AND
  \bauthor{\bsnm{Raftery},~\bfnm{Adrian~E}\binits{A.~E.}}
(\byear{2007}).
\btitle{Strictly Proper Scoring Rules, Prediction, and Estimation}.
\bjournal{Journal of the American Statistical Association}
\bvolume{102}
\bpages{359--378}.
\bdoi{10.1198/016214506000001437}
\end{barticle}
\endbibitem

\bibitem[\protect\citeauthoryear{Gonzalez, Osborne and
  Lawrence}{2016}]{BONonMyopic}
\begin{binproceedings}[author]
\bauthor{\bsnm{Gonzalez},~\bfnm{Javier}\binits{J.}},
  \bauthor{\bsnm{Osborne},~\bfnm{Michael}\binits{M.}} \AND
  \bauthor{\bsnm{Lawrence},~\bfnm{Neil}\binits{N.}}
(\byear{2016}).
\btitle{GLASSES: Relieving The Myopia Of Bayesian Optimisation}.
In \bbooktitle{Proceedings of the 19th International Conference on Artificial
  Intelligence and Statistics}
(\beditor{\bfnm{Arthur}\binits{A.}~\bsnm{Gretton}} \AND
  \beditor{\bfnm{Christian~C.}\binits{C.~C.}~\bsnm{Robert}}, eds.).
\bseries{Proceedings of Machine Learning Research}
\bvolume{51}
\bpages{790--799}.
\bpublisher{PMLR}, \baddress{Cadiz, Spain}.
\end{binproceedings}
\endbibitem

\bibitem[\protect\citeauthoryear{Gorodetsky and
  Marzouk}{2016}]{Mercer_kernels_IVAR}
\begin{barticle}[author]
\bauthor{\bsnm{Gorodetsky},~\bfnm{Alex}\binits{A.}} \AND
  \bauthor{\bsnm{Marzouk},~\bfnm{Youssef}\binits{Y.}}
(\byear{2016}).
\btitle{Mercer Kernels and Integrated Variance Experimental Design: Connections
  Between Gaussian Process Regression and Polynomial Approximation}.
\bjournal{SIAM/ASA Journal on Uncertainty Quantification}
\bvolume{4}
\bpages{796-828}.
\bdoi{10.1137/15M1017119}
\end{barticle}
\endbibitem

\bibitem[\protect\citeauthoryear{Gramacy}{2020}]{gramacy2020surrogates}
\begin{bbook}[author]
\bauthor{\bsnm{Gramacy},~\bfnm{Robert~B.}\binits{R.~B.}}
(\byear{2020}).
\btitle{Surrogates: {G}aussian Process Modeling, Design and \ Optimization for
  the Applied Sciences}.
\bpublisher{Chapman Hall/CRC}, \baddress{Boca Raton, Florida}.
\bnote{\url{http://bobby.gramacy.com/surrogates/}}.
\end{bbook}
\endbibitem

\bibitem[\protect\citeauthoryear{Graves}{2011}]{VIforNNs}
\begin{binproceedings}[author]
\bauthor{\bsnm{Graves},~\bfnm{Alex}\binits{A.}}
(\byear{2011}).
\btitle{Practical Variational Inference for Neural Networks}.
In \bbooktitle{Advances in Neural Information Processing Systems}
(\beditor{\bfnm{J.}\binits{J.}~\bsnm{Shawe-Taylor}},
  \beditor{\bfnm{R.}\binits{R.}~\bsnm{Zemel}},
  \beditor{\bfnm{P.}\binits{P.}~\bsnm{Bartlett}},
  \beditor{\bfnm{F.}\binits{F.}~\bsnm{Pereira}} \AND
  \beditor{\bfnm{K.~Q.}\binits{K.~Q.}~\bsnm{Weinberger}}, eds.)
\bvolume{24}.
\bpublisher{Curran Associates, Inc.}
\end{binproceedings}
\endbibitem

\bibitem[\protect\citeauthoryear{Gutmann and Corander}{2016}]{BOforLFI}
\begin{barticle}[author]
\bauthor{\bsnm{Gutmann},~\bfnm{Michael~U.}\binits{M.~U.}} \AND
  \bauthor{\bsnm{Corander},~\bfnm{Jukka}\binits{J.}}
(\byear{2016}).
\btitle{Bayesian Optimization for Likelihood-Free Inference of Simulator-Based
  Statistical Models}.
\bjournal{Journal of Machine Learning Research}
\bvolume{17}
\bpages{1--47}.
\end{barticle}
\endbibitem

\bibitem[\protect\citeauthoryear{Hartig et~al.}{2011}]{HartigStochasticReview}
\begin{barticle}[author]
\bauthor{\bsnm{Hartig},~\bfnm{Florian}\binits{F.}},
  \bauthor{\bsnm{Calabrese},~\bfnm{Justin~M.}\binits{J.~M.}},
  \bauthor{\bsnm{Reineking},~\bfnm{Björn}\binits{B.}},
  \bauthor{\bsnm{Wiegand},~\bfnm{Thorsten}\binits{T.}} \AND
  \bauthor{\bsnm{Huth},~\bfnm{Andreas}\binits{A.}}
(\byear{2011}).
\btitle{Statistical inference for stochastic simulation models – theory and
  application}.
\bjournal{Ecology Letters}
\bvolume{14}
\bpages{816-827}.
\bdoi{https://doi.org/10.1111/j.1461-0248.2011.01640.x}
\end{barticle}
\endbibitem

\bibitem[\protect\citeauthoryear{He et~al.}{2025}]{DLUQ}
\begin{bmisc}[author]
\bauthor{\bsnm{He},~\bfnm{Wenchong}\binits{W.}},
  \bauthor{\bsnm{Jiang},~\bfnm{Zhe}\binits{Z.}},
  \bauthor{\bsnm{Xiao},~\bfnm{Tingsong}\binits{T.}},
  \bauthor{\bsnm{Xu},~\bfnm{Zelin}\binits{Z.}} \AND
  \bauthor{\bsnm{Li},~\bfnm{Yukun}\binits{Y.}}
(\byear{2025}).
\btitle{A Survey on Uncertainty Quantification Methods for Deep Learning}.
\bdoi{10.48550/arXiv.2302.13425}
\end{bmisc}
\endbibitem

\bibitem[\protect\citeauthoryear{Helin et~al.}{2024}]{StuartTeck2}
\begin{binproceedings}[author]
\bauthor{\bsnm{Helin},~\bfnm{Tapio}\binits{T.}},
  \bauthor{\bsnm{Stuart},~\bfnm{Andrew~M.}\binits{A.~M.}},
  \bauthor{\bsnm{Teckentrup},~\bfnm{Aretha~L.}\binits{A.~L.}} \AND
  \bauthor{\bsnm{Zygalakis},~\bfnm{Konstantinos~C.}\binits{K.~C.}}
(\byear{2024}).
\btitle{Introduction to Gaussian Process Regression in Bayesian Inverse
  Problems, with New Results on Experimental Design for Weighted Error
  Measures}.
In \bbooktitle{Monte Carlo and Quasi-Monte Carlo Methods}
(\beditor{\bfnm{Aicke}\binits{A.}~\bsnm{Hinrichs}},
  \beditor{\bfnm{Peter}\binits{P.}~\bsnm{Kritzer}} \AND
  \beditor{\bfnm{Friedrich}\binits{F.}~\bsnm{Pillichshammer}}, eds.)
\bpages{49--79}.
\bpublisher{Springer International Publishing}, \baddress{Cham}.
\bdoi{10.1007/978-3-031-59762-6_3}
\end{binproceedings}
\endbibitem

\bibitem[\protect\citeauthoryear{Higdon et~al.}{2008}]{HigdonBasis}
\begin{barticle}[author]
\bauthor{\bsnm{Higdon},~\bfnm{Dave}\binits{D.}},
  \bauthor{\bsnm{Gattiker},~\bfnm{James}\binits{J.}},
  \bauthor{\bsnm{Williams},~\bfnm{Brian}\binits{B.}} \AND
  \bauthor{\bsnm{Rightley},~\bfnm{Maria}\binits{M.}}
(\byear{2008}).
\btitle{Computer Model Calibration Using High-Dimensional Output}.
\bjournal{Journal of the American Statistical Association}
\bvolume{103}
\bpages{570-583}.
\bdoi{10.1198/016214507000000888}
\end{barticle}
\endbibitem

\bibitem[\protect\citeauthoryear{Higdon et~al.}{2015}]{Higdon_2015}
\begin{barticle}[author]
\bauthor{\bsnm{Higdon},~\bfnm{Dave}\binits{D.}},
  \bauthor{\bsnm{McDonnell},~\bfnm{Jordan~D}\binits{J.~D.}},
  \bauthor{\bsnm{Schunck},~\bfnm{Nicolas}\binits{N.}},
  \bauthor{\bsnm{Sarich},~\bfnm{Jason}\binits{J.}} \AND
  \bauthor{\bsnm{Wild},~\bfnm{Stefan~M}\binits{S.~M.}}
(\byear{2015}).
\btitle{A Bayesian approach for parameter estimation and prediction using a
  computationally intensive model}.
\bjournal{Journal of Physics G: Nuclear and Particle Physics}
\bvolume{42}
\bpages{034009}.
\bdoi{10.1088/0954-3899/42/3/034009}
\end{barticle}
\endbibitem

\bibitem[\protect\citeauthoryear{Huang et~al.}{2016}]{CLMBayesianCalibration}
\begin{barticle}[author]
\bauthor{\bsnm{Huang},~\bfnm{Maoyi}\binits{M.}},
  \bauthor{\bsnm{Ray},~\bfnm{Jaideep}\binits{J.}},
  \bauthor{\bsnm{Hou},~\bfnm{Zhangshuan}\binits{Z.}},
  \bauthor{\bsnm{Ren},~\bfnm{Huiying}\binits{H.}},
  \bauthor{\bsnm{Liu},~\bfnm{Ying}\binits{Y.}} \AND
  \bauthor{\bsnm{Swiler},~\bfnm{Laura}\binits{L.}}
(\byear{2016}).
\btitle{On the applicability of surrogate-based Markov chain Monte
  Carlo-Bayesian inversion to the Community Land Model: Case studies at flux
  tower sites}.
\bjournal{Journal of Geophysical Research: Atmospheres}
\bvolume{121}
\bpages{7548-7563}.
\bdoi{https://doi.org/10.1002/2015JD024339}
\end{barticle}
\endbibitem

\bibitem[\protect\citeauthoryear{Huang et~al.}{2020}]{onSiteCalibration}
\begin{barticle}[author]
\bauthor{\bsnm{Huang},~\bfnm{Jiangeng}\binits{J.}},
  \bauthor{\bsnm{Gramacy},~\bfnm{Robert~B.}\binits{R.~B.}},
  \bauthor{\bsnm{Binois},~\bfnm{Mickaël}\binits{M.}} \AND
  \bauthor{\bsnm{Libraschi},~\bfnm{Mirko}\binits{M.}}
(\byear{2020}).
\btitle{On-site surrogates for large-scale calibration}.
\bjournal{Applied Stochastic Models in Business and Industry}
\bvolume{36}
\bpages{283-304}.
\bdoi{https://doi.org/10.1002/asmb.2523}
\end{barticle}
\endbibitem

\bibitem[\protect\citeauthoryear{Hüllermeier and
  Waegeman}{2021}]{epistemicAleatoric}
\begin{barticle}[author]
\bauthor{\bsnm{Hüllermeier},~\bfnm{E.}\binits{E.}} \AND
  \bauthor{\bsnm{Waegeman},~\bfnm{W.}\binits{W.}}
(\byear{2021}).
\btitle{Aleatoric and epistemic uncertainty in machine learning: an
  introduction to concepts and methods}.
\bjournal{Mach Learn}
\bvolume{110}
\bpages{457-506}.
\bdoi{https://doi.org/10.1007/s10994-021-05946-3}
\end{barticle}
\endbibitem

\bibitem[\protect\citeauthoryear{Iglesias, Law and Stuart}{2013}]{IglesiasEKI}
\begin{barticle}[author]
\bauthor{\bsnm{Iglesias},~\bfnm{Marco~A}\binits{M.~A.}},
  \bauthor{\bsnm{Law},~\bfnm{Kody J~H}\binits{K.~J.~H.}} \AND
  \bauthor{\bsnm{Stuart},~\bfnm{Andrew~M}\binits{A.~M.}}
(\byear{2013}).
\btitle{Ensemble Kalman methods for inverse problems}.
\bjournal{Inverse Problems}
\bvolume{29}
\bpages{045001}.
\bdoi{10.1088/0266-5611/29/4/045001}
\end{barticle}
\endbibitem

\bibitem[\protect\citeauthoryear{Jacob et~al.}{2017}]{moduleModels}
\begin{bmisc}[author]
\bauthor{\bsnm{Jacob},~\bfnm{Pierre~E.}\binits{P.~E.}},
  \bauthor{\bsnm{Murray},~\bfnm{Lawrence~M.}\binits{L.~M.}},
  \bauthor{\bsnm{Holmes},~\bfnm{Chris~C.}\binits{C.~C.}} \AND
  \bauthor{\bsnm{Robert},~\bfnm{Christian~P.}\binits{C.~P.}}
(\byear{2017}).
\btitle{Better together? Statistical learning in models made of modules}.
\bdoi{10.48550/arXiv.1708.08719}
\end{bmisc}
\endbibitem

\bibitem[\protect\citeauthoryear{J{{\"a}}rvenp{{\"a}}{{\"a}} and
  Corander}{2024}]{gpEmMCMC}
\begin{barticle}[author]
\bauthor{\bsnm{J{{\"a}}rvenp{{\"a}}{{\"a}}},~\bfnm{Marko}\binits{M.}} \AND
  \bauthor{\bsnm{Corander},~\bfnm{Jukka}\binits{J.}}
(\byear{2024}).
\btitle{Approximate Bayesian inference from noisy likelihoods with Gaussian
  process emulated MCMC}.
\bjournal{Journal of Machine Learning Research}
\bvolume{25}
\bpages{1--55}.
\bdoi{10.48550/arXiv.2104.03942}
\end{barticle}
\endbibitem

\bibitem[\protect\citeauthoryear{J{\"a}rvenp{\"a}{\"a}
  et~al.}{2021}]{VehtariParallelGP}
\begin{barticle}[author]
\bauthor{\bsnm{J{\"a}rvenp{\"a}{\"a}},~\bfnm{Marko}\binits{M.}},
  \bauthor{\bsnm{Gutmann},~\bfnm{Michael~U.}\binits{M.~U.}},
  \bauthor{\bsnm{Vehtari},~\bfnm{Aki}\binits{A.}} \AND
  \bauthor{\bsnm{Marttinen},~\bfnm{Pekka}\binits{P.}}
(\byear{2021}).
\btitle{{Parallel Gaussian Process Surrogate Bayesian Inference with Noisy
  Likelihood Evaluations}}.
\bjournal{Bayesian Analysis}
\bvolume{16}
\bpages{147 -- 178}.
\bdoi{10.1214/20-BA1200}
\end{barticle}
\endbibitem

\bibitem[\protect\citeauthoryear{Jedhoff et~al.}{2026}]{BurknerTwoStep}
\begin{bmisc}[author]
\bauthor{\bsnm{Jedhoff},~\bfnm{Svenja}\binits{S.}},
  \bauthor{\bsnm{Kutabi},~\bfnm{Hadi}\binits{H.}},
  \bauthor{\bsnm{Meyer},~\bfnm{Anne}\binits{A.}} \AND
  \bauthor{\bsnm{Bürkner},~\bfnm{Paul-Christian}\binits{P.-C.}}
(\byear{2026}).
\btitle{Efficient Uncertainty Propagation in Bayesian Two-Step Procedures}.
\end{bmisc}
\endbibitem

\bibitem[\protect\citeauthoryear{Joseph et~al.}{2015}]{JosephMinEnergy}
\begin{barticle}[author]
\bauthor{\bsnm{Joseph},~\bfnm{V.~Roshan}\binits{V.~R.}},
  \bauthor{\bsnm{Dasgupta},~\bfnm{Tirthankar}\binits{T.}},
  \bauthor{\bsnm{Tuo},~\bfnm{Rui}\binits{R.}} \AND
  \bauthor{\bsnm{and},~\bfnm{C.~F. Jeff~Wu}\binits{C.~F. J.~W.}}
(\byear{2015}).
\btitle{Sequential Exploration of Complex Surfaces Using Minimum Energy
  Designs}.
\bjournal{Technometrics}
\bvolume{57}
\bpages{64--74}.
\bdoi{10.1080/00401706.2014.881749}
\end{barticle}
\endbibitem

\bibitem[\protect\citeauthoryear{Joseph et~al.}{2019}]{JosephMEDSampling}
\begin{barticle}[author]
\bauthor{\bsnm{Joseph},~\bfnm{V.~Roshan}\binits{V.~R.}},
  \bauthor{\bsnm{Wang},~\bfnm{Dianpeng}\binits{D.}},
  \bauthor{\bsnm{Gu},~\bfnm{Li}\binits{L.}},
  \bauthor{\bsnm{Lyu},~\bfnm{Shiji}\binits{S.}} \AND
  \bauthor{\bsnm{and},~\bfnm{Rui~Tuo}\binits{R.~T.}}
(\byear{2019}).
\btitle{Deterministic Sampling of Expensive Posteriors Using Minimum Energy
  Designs}.
\bjournal{Technometrics}
\bvolume{61}
\bpages{297--308}.
\bdoi{10.1080/00401706.2018.1552203}
\end{barticle}
\endbibitem

\bibitem[\protect\citeauthoryear{Järvenpää et~al.}{2019}]{EfficientAcqABC}
\begin{barticle}[author]
\bauthor{\bsnm{Järvenpää},~\bfnm{Marko}\binits{M.}},
  \bauthor{\bsnm{Gutmann},~\bfnm{Michael}\binits{M.}},
  \bauthor{\bsnm{Pleska},~\bfnm{Arijus}\binits{A.}},
  \bauthor{\bsnm{Vehtari},~\bfnm{Aki}\binits{A.}} \AND
  \bauthor{\bsnm{Marttinen},~\bfnm{Pekka}\binits{P.}}
(\byear{2019}).
\btitle{Efficient Acquisition Rules for Model-Based Approximate Bayesian
  Computation}.
\bjournal{Bayesian Analysis}
\bvolume{14}.
\bdoi{10.1214/18-BA1121}
\end{barticle}
\endbibitem

\bibitem[\protect\citeauthoryear{Kandasamy, Schneider and
  P\'{o}czos}{2015}]{KandasamyActiveLearning2015}
\begin{binproceedings}[author]
\bauthor{\bsnm{Kandasamy},~\bfnm{Kirthevasan}\binits{K.}},
  \bauthor{\bsnm{Schneider},~\bfnm{Jeff}\binits{J.}} \AND
  \bauthor{\bsnm{P\'{o}czos},~\bfnm{Barnab\'{a}s}\binits{B.}}
(\byear{2015}).
\btitle{Bayesian Active Learning for Posterior Estimation}.
In \bbooktitle{Proceedings of the 24th International Conference on Artificial
  Intelligence}.
\bseries{IJCAI'15}
\bpages{3605–3611}.
\bpublisher{AAAI Press}.
\end{binproceedings}
\endbibitem

\bibitem[\protect\citeauthoryear{Kandasamy, Schneider and
  Póczos}{2017}]{Kandasamy2017}
\begin{barticle}[author]
\bauthor{\bsnm{Kandasamy},~\bfnm{Kirthevasan}\binits{K.}},
  \bauthor{\bsnm{Schneider},~\bfnm{Jeff}\binits{J.}} \AND
  \bauthor{\bsnm{Póczos},~\bfnm{Barnabás}\binits{B.}}
(\byear{2017}).
\btitle{Query efficient posterior estimation in scientific experiments via
  Bayesian active learning}.
\bjournal{Artificial Intelligence}
\bvolume{243}
\bpages{45–56}.
\bdoi{10.1016/j.artint.2016.11.002}
\end{barticle}
\endbibitem

\bibitem[\protect\citeauthoryear{Kandasamy et~al.}{2018}]{parallelBOThompson}
\begin{binproceedings}[author]
\bauthor{\bsnm{Kandasamy},~\bfnm{Kirthevasan}\binits{K.}},
  \bauthor{\bsnm{Krishnamurthy},~\bfnm{Akshay}\binits{A.}},
  \bauthor{\bsnm{Schneider},~\bfnm{Jeff}\binits{J.}} \AND
  \bauthor{\bsnm{Poczos},~\bfnm{Barnabas}\binits{B.}}
(\byear{2018}).
\btitle{Parallelised Bayesian Optimisation via Thompson Sampling}.
In \bbooktitle{Proceedings of the Twenty-First International Conference on
  Artificial Intelligence and Statistics}
(\beditor{\bfnm{Amos}\binits{A.}~\bsnm{Storkey}} \AND
  \beditor{\bfnm{Fernando}\binits{F.}~\bsnm{Perez-Cruz}}, eds.).
\bseries{Proceedings of Machine Learning Research}
\bvolume{84}
\bpages{133--142}.
\bpublisher{PMLR}.
\end{binproceedings}
\endbibitem

\bibitem[\protect\citeauthoryear{Karabatsos and Leisen}{2018}]{ABCApproxLik}
\begin{barticle}[author]
\bauthor{\bsnm{Karabatsos},~\bfnm{George}\binits{G.}} \AND
  \bauthor{\bsnm{Leisen},~\bfnm{Fabrizio}\binits{F.}}
(\byear{2018}).
\btitle{{An approximate likelihood perspective on ABC methods}}.
\bjournal{Statistics Surveys}
\bvolume{12}
\bpages{66 -- 104}.
\bdoi{10.1214/18-SS120}
\end{barticle}
\endbibitem

\bibitem[\protect\citeauthoryear{Keetz et~al.}{2025}]{FATES_CES}
\begin{barticle}[author]
\bauthor{\bsnm{Keetz},~\bfnm{L.~T.}\binits{L.~T.}},
  \bauthor{\bsnm{Aalstad},~\bfnm{K.}\binits{K.}},
  \bauthor{\bsnm{Fisher},~\bfnm{R.~A.}\binits{R.~A.}},
  \bauthor{\bsnm{Poppe~Terán},~\bfnm{C.}\binits{C.}},
  \bauthor{\bsnm{Naz},~\bfnm{B.}\binits{B.}},
  \bauthor{\bsnm{Pirk},~\bfnm{N.}\binits{N.}},
  \bauthor{\bsnm{Yilmaz},~\bfnm{Y.~A.}\binits{Y.~A.}} \AND
  \bauthor{\bsnm{Skarpaas},~\bfnm{O.}\binits{O.}}
(\byear{2025}).
\btitle{Inferring Parameters in a Complex Land Surface Model by Combining Data
  Assimilation and Machine Learning}.
\bjournal{Journal of Advances in Modeling Earth Systems}
\bvolume{17}
\bpages{e2024MS004542}.
\bnote{e2024MS004542 2024MS004542}.
\bdoi{https://doi.org/10.1029/2024MS004542}
\end{barticle}
\endbibitem

\bibitem[\protect\citeauthoryear{Kennedy and O'Hagan}{2001}]{KOH}
\begin{barticle}[author]
\bauthor{\bsnm{Kennedy},~\bfnm{Marc~C.}\binits{M.~C.}} \AND
  \bauthor{\bsnm{O'Hagan},~\bfnm{Anthony}\binits{A.}}
(\byear{2001}).
\btitle{Bayesian calibration of computer models}.
\bjournal{Journal of the Royal Statistical Society: Series B (Statistical
  Methodology)}
\bvolume{63}
\bpages{425-464}.
\bdoi{https://doi.org/10.1111/1467-9868.00294}
\end{barticle}
\endbibitem

\bibitem[\protect\citeauthoryear{Kim and
  Sanz-Alonso}{2024}]{gp_surrogates_random_exploration}
\begin{bmisc}[author]
\bauthor{\bsnm{Kim},~\bfnm{Hwanwoo}\binits{H.}} \AND
  \bauthor{\bsnm{Sanz-Alonso},~\bfnm{Daniel}\binits{D.}}
(\byear{2024}).
\btitle{Enhancing Gaussian Process Surrogates for Optimization and Posterior
  Approximation via Random Exploration}.
\end{bmisc}
\endbibitem

\bibitem[\protect\citeauthoryear{Kochenderfer
  et~al.}{2015}]{DecisionMakingUncertainty}
\begin{bbook}[author]
\bauthor{\bsnm{Kochenderfer},~\bfnm{Mykel~J.}\binits{M.~J.}},
  \bauthor{\bsnm{Amato},~\bfnm{Christopher}\binits{C.}},
  \bauthor{\bsnm{Chowdhary},~\bfnm{Girish}\binits{G.}},
  \bauthor{\bsnm{How},~\bfnm{Jonathan~P.}\binits{J.~P.}},
  \bauthor{\bsnm{Reynolds},~\bfnm{Hayley J.~Davison}\binits{H.~J.~D.}},
  \bauthor{\bsnm{Thornton},~\bfnm{Jason~R.}\binits{J.~R.}},
  \bauthor{\bsnm{Torres-Carrasquillo},~\bfnm{Pedro~A.}\binits{P.~A.}},
  \bauthor{\bsnm{{\"U}re},~\bfnm{N.~Kemal}\binits{N.~K.}} \AND
  \bauthor{\bsnm{Vian},~\bfnm{John}\binits{J.}}
(\byear{2015}).
\btitle{Decision Making Under Uncertainty: Theory and Application}.
\bpublisher{The MIT Press}, \baddress{Cambridge, MA}.
\bdoi{10.7551/mitpress/10187.001.0001}
\end{bbook}
\endbibitem

\bibitem[\protect\citeauthoryear{Koermer et~al.}{2024}]{Koermer2024}
\begin{barticle}[author]
\bauthor{\bsnm{Koermer},~\bfnm{Scott}\binits{S.}},
  \bauthor{\bsnm{Loda},~\bfnm{Justin}\binits{J.}},
  \bauthor{\bsnm{Noble},~\bfnm{Aaron}\binits{A.}} \AND
  \bauthor{\bsnm{and},~\bfnm{Robert B.~Gramacy}\binits{R.~B.~G.}}
(\byear{2024}).
\btitle{Augmenting a Simulation Campaign for Hybrid Computer Model and Field
  Data Experiments}.
\bjournal{Technometrics}
\bvolume{66}
\bpages{638--650}.
\bdoi{10.1080/00401706.2024.2345139}
\end{barticle}
\endbibitem

\bibitem[\protect\citeauthoryear{Krouglova et~al.}{2025}]{multifidelitySBI}
\begin{bmisc}[author]
\bauthor{\bsnm{Krouglova},~\bfnm{Anastasia~N.}\binits{A.~N.}},
  \bauthor{\bsnm{Johnson},~\bfnm{Hayden~R.}\binits{H.~R.}},
  \bauthor{\bsnm{Confavreux},~\bfnm{Basile}\binits{B.}},
  \bauthor{\bsnm{Deistler},~\bfnm{Michael}\binits{M.}} \AND
  \bauthor{\bsnm{Gonçalves},~\bfnm{Pedro~J.}\binits{P.~J.}}
(\byear{2025}).
\btitle{Multifidelity Simulation-based Inference for Computationally Expensive
  Simulators}.
\end{bmisc}
\endbibitem

\bibitem[\protect\citeauthoryear{Lakshminarayanan, Pritzel and
  Blundell}{2017}]{deepEnsembles}
\begin{binproceedings}[author]
\bauthor{\bsnm{Lakshminarayanan},~\bfnm{Balaji}\binits{B.}},
  \bauthor{\bsnm{Pritzel},~\bfnm{Alexander}\binits{A.}} \AND
  \bauthor{\bsnm{Blundell},~\bfnm{Charles}\binits{C.}}
(\byear{2017}).
\btitle{Simple and Scalable Predictive Uncertainty Estimation using Deep
  Ensembles}.
In \bbooktitle{Advances in Neural Information Processing Systems}
(\beditor{\bfnm{I.}\binits{I.}~\bsnm{Guyon}},
  \beditor{\bfnm{U.~Von}\binits{U.~V.}~\bsnm{Luxburg}},
  \beditor{\bfnm{S.}\binits{S.}~\bsnm{Bengio}},
  \beditor{\bfnm{H.}\binits{H.}~\bsnm{Wallach}},
  \beditor{\bfnm{R.}\binits{R.}~\bsnm{Fergus}},
  \beditor{\bfnm{S.}\binits{S.}~\bsnm{Vishwanathan}} \AND
  \beditor{\bfnm{R.}\binits{R.}~\bsnm{Garnett}}, eds.)
\bvolume{30}.
\bpublisher{Curran Associates, Inc.}
\bdoi{10.48550/arXiv.1612.01474}
\end{binproceedings}
\endbibitem

\bibitem[\protect\citeauthoryear{Lartaud, Humbert and
  Garnier}{2024}]{weightedIVAR}
\begin{bmisc}[author]
\bauthor{\bsnm{Lartaud},~\bfnm{Paul}\binits{P.}},
  \bauthor{\bsnm{Humbert},~\bfnm{Philippe}\binits{P.}} \AND
  \bauthor{\bsnm{Garnier},~\bfnm{Josselin}\binits{J.}}
(\byear{2024}).
\btitle{Sequential design for surrogate modeling in Bayesian inverse problems}.
\bdoi{10.48550/arXiv.2402.16520}
\end{bmisc}
\endbibitem

\bibitem[\protect\citeauthoryear{Lebel et~al.}{2019}]{trainDynamics}
\begin{barticle}[author]
\bauthor{\bsnm{Lebel},~\bfnm{D.}\binits{D.}},
  \bauthor{\bsnm{Soize},~\bfnm{C.}\binits{C.}},
  \bauthor{\bsnm{Fünfschilling},~\bfnm{C.}\binits{C.}} \AND
  \bauthor{\bsnm{Perrin},~\bfnm{G.}\binits{G.}}
(\byear{2019}).
\btitle{Statistical inverse identification for nonlinear train dynamics using a
  surrogate model in a Bayesian framework}.
\bjournal{Journal of Sound and Vibration}
\bvolume{458}
\bpages{158-176}.
\bdoi{https://doi.org/10.1016/j.jsv.2019.06.024}
\end{barticle}
\endbibitem

\bibitem[\protect\citeauthoryear{Li and Marzouk}{2014}]{Li_2014}
\begin{barticle}[author]
\bauthor{\bsnm{Li},~\bfnm{Jinglai}\binits{J.}} \AND
  \bauthor{\bsnm{Marzouk},~\bfnm{Youssef~M.}\binits{Y.~M.}}
(\byear{2014}).
\btitle{Adaptive Construction of Surrogates for the Bayesian Solution of
  Inverse Problems}.
\bjournal{SIAM Journal on Scientific Computing}
\bvolume{36}
\bpages{A1163–A1186}.
\bdoi{10.1137/130938189}
\end{barticle}
\endbibitem

\bibitem[\protect\citeauthoryear{Li, Rudner and Wilson}{2024}]{BayesOptNN}
\begin{binproceedings}[author]
\bauthor{\bsnm{Li},~\bfnm{Yucen}\binits{Y.}}, \bauthor{\bsnm{Rudner},~\bfnm{Tim
  G.~J.}\binits{T.~G.~J.}} \AND
  \bauthor{\bsnm{Wilson},~\bfnm{Andrew~Gordon}\binits{A.~G.}}
(\byear{2024}).
\btitle{A Study of Bayesian Neural Network Surrogates for Bayesian
  Optimization}.
In \bbooktitle{International Conference on Learning Representations}
(\beditor{\bfnm{B.}\binits{B.}~\bsnm{Kim}},
  \beditor{\bfnm{Y.}\binits{Y.}~\bsnm{Yue}},
  \beditor{\bfnm{S.}\binits{S.}~\bsnm{Chaudhuri}},
  \beditor{\bfnm{K.}\binits{K.}~\bsnm{Fragkiadaki}},
  \beditor{\bfnm{M.}\binits{M.}~\bsnm{Khan}} \AND
  \beditor{\bfnm{Y.}\binits{Y.}~\bsnm{Sun}}, eds.)
\bvolume{2024}
\bpages{47003--47041}.
\bdoi{10.48550/arXiv.2305.20028}
\end{binproceedings}
\endbibitem

\bibitem[\protect\citeauthoryear{Li et~al.}{2025}]{AmortizedBayesianWorkflow}
\begin{bmisc}[author]
\bauthor{\bsnm{Li},~\bfnm{Chengkun}\binits{C.}},
  \bauthor{\bsnm{Vehtari},~\bfnm{Aki}\binits{A.}},
  \bauthor{\bsnm{Bürkner},~\bfnm{Paul-Christian}\binits{P.-C.}},
  \bauthor{\bsnm{Radev},~\bfnm{Stefan~T.}\binits{S.~T.}},
  \bauthor{\bsnm{Acerbi},~\bfnm{Luigi}\binits{L.}} \AND
  \bauthor{\bsnm{Schmitt},~\bfnm{Marvin}\binits{M.}}
(\byear{2025}).
\btitle{Amortized Bayesian Workflow}.
\bdoi{10.48550/arXiv.2409.04332}
\end{bmisc}
\endbibitem

\bibitem[\protect\citeauthoryear{Lie, Sullivan and
  Teckentrup}{2018}]{random_fwd_models}
\begin{barticle}[author]
\bauthor{\bsnm{Lie},~\bfnm{H.~C.}\binits{H.~C.}},
  \bauthor{\bsnm{Sullivan},~\bfnm{T.~J.}\binits{T.~J.}} \AND
  \bauthor{\bsnm{Teckentrup},~\bfnm{A.~L.}\binits{A.~L.}}
(\byear{2018}).
\btitle{Random Forward Models and Log-Likelihoods in Bayesian Inverse
  Problems}.
\bjournal{SIAM/ASA Journal on Uncertainty Quantification}
\bvolume{6}
\bpages{1600–1629}.
\bdoi{10.1137/18m1166523}
\end{barticle}
\endbibitem

\bibitem[\protect\citeauthoryear{Liu, Bayarri and
  Berger}{2009}]{modularization}
\begin{barticle}[author]
\bauthor{\bsnm{Liu},~\bfnm{F.}\binits{F.}},
  \bauthor{\bsnm{Bayarri},~\bfnm{M.}\binits{M.}} \AND
  \bauthor{\bsnm{Berger},~\bfnm{J.}\binits{J.}}
(\byear{2009}).
\btitle{Modularization in Bayesian analysis, with emphasis on analysis of
  computer models}.
\bjournal{Bayesian Analysis}
\bvolume{4}.
\bdoi{10.1214/09-BA404}
\end{barticle}
\endbibitem

\bibitem[\protect\citeauthoryear{Liu and
  West}{2009}]{Liu_West_dynamic_emulation}
\begin{barticle}[author]
\bauthor{\bsnm{Liu},~\bfnm{Fei}\binits{F.}} \AND
  \bauthor{\bsnm{West},~\bfnm{Mike}\binits{M.}}
(\byear{2009}).
\btitle{{A dynamic modelling strategy for Bayesian computer model emulation}}.
\bjournal{Bayesian Analysis}
\bvolume{4}
\bpages{393 -- 411}.
\bdoi{10.1214/09-BA415}
\end{barticle}
\endbibitem

\bibitem[\protect\citeauthoryear{Llorente et~al.}{2025}]{noisyMCSurvey}
\begin{barticle}[author]
\bauthor{\bsnm{Llorente},~\bfnm{Fernando}\binits{F.}},
  \bauthor{\bsnm{Martino},~\bfnm{Luca}\binits{L.}},
  \bauthor{\bsnm{Read},~\bfnm{Jesse}\binits{J.}} \AND
  \bauthor{\bsnm{Delgado-Gómez},~\bfnm{David}\binits{D.}}
(\byear{2025}).
\btitle{A Survey of Monte Carlo Methods for Noisy and Costly Densities With
  Application to Reinforcement Learning and ABC}.
\bjournal{International Statistical Review}
\bvolume{93}
\bpages{18-61}.
\bdoi{https://doi.org/10.1111/insr.12573}
\end{barticle}
\endbibitem

\bibitem[\protect\citeauthoryear{Loeppky, Moore and
  Williams}{2010}]{LOEPPKY20101452}
\begin{barticle}[author]
\bauthor{\bsnm{Loeppky},~\bfnm{Jason~L.}\binits{J.~L.}},
  \bauthor{\bsnm{Moore},~\bfnm{Leslie~M.}\binits{L.~M.}} \AND
  \bauthor{\bsnm{Williams},~\bfnm{Brian~J.}\binits{B.~J.}}
(\byear{2010}).
\btitle{Batch sequential designs for computer experiments}.
\bjournal{Journal of Statistical Planning and Inference}
\bvolume{140}
\bpages{1452-1464}.
\bdoi{https://doi.org/10.1016/j.jspi.2009.12.004}
\end{barticle}
\endbibitem

\bibitem[\protect\citeauthoryear{Lu et~al.}{2015}]{PCEBIP}
\begin{barticle}[author]
\bauthor{\bsnm{Lu},~\bfnm{Fei}\binits{F.}},
  \bauthor{\bsnm{Morzfeld},~\bfnm{Matthias}\binits{M.}},
  \bauthor{\bsnm{Tu},~\bfnm{Xuemin}\binits{X.}} \AND
  \bauthor{\bsnm{Chorin},~\bfnm{Alexandre~J.}\binits{A.~J.}}
(\byear{2015}).
\btitle{Limitations of polynomial chaos expansions in the Bayesian solution of
  inverse problems}.
\bjournal{Journal of Computational Physics}
\bvolume{282}
\bpages{138-147}.
\bdoi{https://doi.org/10.1016/j.jcp.2014.11.010}
\end{barticle}
\endbibitem

\bibitem[\protect\citeauthoryear{Lueckmann et~al.}{2019}]{Lueckmann2019}
\begin{binproceedings}[author]
\bauthor{\bsnm{Lueckmann},~\bfnm{Jan-Matthis}\binits{J.-M.}},
  \bauthor{\bsnm{Bassetto},~\bfnm{Giacomo}\binits{G.}},
  \bauthor{\bsnm{Karaletsos},~\bfnm{Theofanis}\binits{T.}} \AND
  \bauthor{\bsnm{Macke},~\bfnm{Jakob~H.}\binits{J.~H.}}
(\byear{2019}).
\btitle{Likelihood-free inference with emulator networks}.
In \bbooktitle{Proceedings of The 1st Symposium on Advances in Approximate
  Bayesian Inference}
(\beditor{\bfnm{Francisco}\binits{F.}~\bsnm{Ruiz}},
  \beditor{\bfnm{Cheng}\binits{C.}~\bsnm{Zhang}},
  \beditor{\bfnm{Dawen}\binits{D.}~\bsnm{Liang}} \AND
  \beditor{\bfnm{Thang}\binits{T.}~\bsnm{Bui}}, eds.).
\bseries{Proceedings of Machine Learning Research}
\bvolume{96}
\bpages{32--53}.
\bpublisher{PMLR}.
\bdoi{doi.org/10.48550/arXiv.1805.09294}
\end{binproceedings}
\endbibitem

\bibitem[\protect\citeauthoryear{MacKay}{1992}]{MacKayLaplaceNN}
\begin{barticle}[author]
\bauthor{\bsnm{MacKay},~\bfnm{David J.~C.}\binits{D.~J.~C.}}
(\byear{1992}).
\btitle{A Practical Bayesian Framework for Backpropagation Networks}.
\bjournal{Neural Computation}
\bvolume{4}
\bpages{448-472}.
\bdoi{10.1162/neco.1992.4.3.448}
\end{barticle}
\endbibitem

\bibitem[\protect\citeauthoryear{Maddox et~al.}{2021}]{BOHighDimOutputs}
\begin{binproceedings}[author]
\bauthor{\bsnm{Maddox},~\bfnm{Wesley~J}\binits{W.~J.}},
  \bauthor{\bsnm{Balandat},~\bfnm{Maximilian}\binits{M.}},
  \bauthor{\bsnm{Wilson},~\bfnm{Andrew~G}\binits{A.~G.}} \AND
  \bauthor{\bsnm{Bakshy},~\bfnm{Eytan}\binits{E.}}
(\byear{2021}).
\btitle{Bayesian Optimization with High-Dimensional Outputs}.
In \bbooktitle{Advances in Neural Information Processing Systems}
(\beditor{\bfnm{M.}\binits{M.}~\bsnm{Ranzato}},
  \beditor{\bfnm{A.}\binits{A.}~\bsnm{Beygelzimer}},
  \beditor{\bfnm{Y.}\binits{Y.}~\bsnm{Dauphin}},
  \beditor{\bfnm{P.~S.}\binits{P.~S.}~\bsnm{Liang}} \AND
  \beditor{\bfnm{J.~Wortman}\binits{J.~W.}~\bsnm{Vaughan}}, eds.)
\bvolume{34}
\bpages{19274--19287}.
\bpublisher{Curran Associates, Inc.}
\end{binproceedings}
\endbibitem

\bibitem[\protect\citeauthoryear{Marinescu et~al.}{2023}]{PCEGPWind}
\begin{barticle}[author]
\bauthor{\bsnm{Marinescu},~\bfnm{Marius}\binits{M.}},
  \bauthor{\bsnm{Olivares},~\bfnm{Alberto}\binits{A.}},
  \bauthor{\bsnm{Staffetti},~\bfnm{Ernesto}\binits{E.}} \AND
  \bauthor{\bsnm{Sun},~\bfnm{Junzi}\binits{J.}}
(\byear{2023}).
\btitle{Polynomial Chaos Expansion-Based Enhanced Gaussian Process Regression
  for Wind Velocity Field Estimation from Aircraft-Derived Data}.
\bjournal{Mathematics}
\bvolume{11}.
\bdoi{10.3390/math11041018}
\end{barticle}
\endbibitem

\bibitem[\protect\citeauthoryear{Marzouk and Najm}{2009}]{dimRedPolyChaos}
\begin{barticle}[author]
\bauthor{\bsnm{Marzouk},~\bfnm{Youssef~M.}\binits{Y.~M.}} \AND
  \bauthor{\bsnm{Najm},~\bfnm{Habib~N.}\binits{H.~N.}}
(\byear{2009}).
\btitle{Dimensionality reduction and polynomial chaos acceleration of Bayesian
  inference in inverse problems}.
\bjournal{Journal of Computational Physics}
\bvolume{228}
\bpages{1862-1902}.
\bdoi{https://doi.org/10.1016/j.jcp.2008.11.024}
\end{barticle}
\endbibitem

\bibitem[\protect\citeauthoryear{McClarren}{2018}]{UQpredCompSci}
\begin{bbook}[author]
\bauthor{\bsnm{McClarren},~\bfnm{Ryan~G.}\binits{R.~G.}}
(\byear{2018}).
\btitle{Uncertainty Quantification and Predictive Computational Science: A
  Foundation for Physical Scientists and Engineers}.
\bpublisher{Springer Cham}
\bnote{\url{https://doi.org/10.1007/978-3-319-99525-0}}.
\end{bbook}
\endbibitem

\bibitem[\protect\citeauthoryear{Medina-Aguayo, Lee and
  Roberts}{2016}]{stabilityNoisyMH}
\begin{barticle}[author]
\bauthor{\bsnm{Medina-Aguayo},~\bfnm{Felipe~J.}\binits{F.~J.}},
  \bauthor{\bsnm{Lee},~\bfnm{Anthony}\binits{A.}} \AND
  \bauthor{\bsnm{Roberts},~\bfnm{Gareth~O.}\binits{G.~O.}}
(\byear{2016}).
\btitle{Stability of noisy {M}etropolis--{H}astings}.
\bjournal{Statistics and Computing}
\bvolume{26}
\bpages{1187--1211}.
\bdoi{10.1007/s11222-015-9604-3}
\end{barticle}
\endbibitem

\bibitem[\protect\citeauthoryear{Meeds and Welling}{2014}]{GPSABC}
\begin{binproceedings}[author]
\bauthor{\bsnm{Meeds},~\bfnm{Edward}\binits{E.}} \AND
  \bauthor{\bsnm{Welling},~\bfnm{Max}\binits{M.}}
(\byear{2014}).
\btitle{GPS-ABC: Gaussian process surrogate approximate Bayesian computation}.
In \bbooktitle{Proceedings of the Thirtieth Conference on Uncertainty in
  Artificial Intelligence}.
\bseries{UAI'14}
\bpages{593–602}.
\bpublisher{AUAI Press}, \baddress{Arlington, Virginia, USA}.
\end{binproceedings}
\endbibitem

\bibitem[\protect\citeauthoryear{Mena, Pujol and Vitri\`{a}}{2021}]{BNNSurvey}
\begin{barticle}[author]
\bauthor{\bsnm{Mena},~\bfnm{Jos\'{e}}\binits{J.}},
  \bauthor{\bsnm{Pujol},~\bfnm{Oriol}\binits{O.}} \AND
  \bauthor{\bsnm{Vitri\`{a}},~\bfnm{Jordi}\binits{J.}}
(\byear{2021}).
\btitle{A Survey on Uncertainty Estimation in Deep Learning Classification
  Systems from a Bayesian Perspective}.
\bjournal{ACM Comput. Surv.}
\bvolume{54}.
\bdoi{10.1145/3477140}
\end{barticle}
\endbibitem

\bibitem[\protect\citeauthoryear{Miller and Mak}{2025}]{MakTargetedVar}
\begin{barticle}[author]
\bauthor{\bsnm{Miller},~\bfnm{John~J.}\binits{J.~J.}} \AND
  \bauthor{\bsnm{Mak},~\bfnm{Simon}\binits{S.}}
(\byear{2025}).
\btitle{Targeted Variance Reduction: Effective Bayesian Optimization of
  Black-Box Simulators with Noise Parameters}.
\bjournal{Technometrics}
\bvolume{67}
\bpages{617--631}.
\bdoi{10.1080/00401706.2025.2495298}
\end{barticle}
\endbibitem

\bibitem[\protect\citeauthoryear{Mohammadi, Challenor and
  Goodfellow}{2019}]{dynamic_nonlinear_simulators_GP}
\begin{barticle}[author]
\bauthor{\bsnm{Mohammadi},~\bfnm{Hossein}\binits{H.}},
  \bauthor{\bsnm{Challenor},~\bfnm{Peter}\binits{P.}} \AND
  \bauthor{\bsnm{Goodfellow},~\bfnm{Marc}\binits{M.}}
(\byear{2019}).
\btitle{Emulating dynamic non-linear simulators using Gaussian processes}.
\bjournal{Computational Statistics \& Data Analysis}
\bvolume{139}
\bpages{178-196}.
\bdoi{https://doi.org/10.1016/j.csda.2019.05.006}
\end{barticle}
\endbibitem

\bibitem[\protect\citeauthoryear{Naesseth, Lindsten and
  Sch\"{o}n}{2019}]{ElementsSMC}
\begin{barticle}[author]
\bauthor{\bsnm{Naesseth},~\bfnm{Christian~A.}\binits{C.~A.}},
  \bauthor{\bsnm{Lindsten},~\bfnm{Fredrik}\binits{F.}} \AND
  \bauthor{\bsnm{Sch\"{o}n},~\bfnm{Thomas~B.}\binits{T.~B.}}
(\byear{2019}).
\btitle{Elements of Sequential Monte Carlo}.
\bjournal{Found. Trends Mach. Learn.}
\bvolume{12}
\bpages{307–392}.
\bdoi{10.1561/2200000074}
\end{barticle}
\endbibitem

\bibitem[\protect\citeauthoryear{Oakley and Youngman}{2017}]{OakleyllikEm}
\begin{barticle}[author]
\bauthor{\bsnm{Oakley},~\bfnm{Jeremy~E.}\binits{J.~E.}} \AND
  \bauthor{\bsnm{Youngman},~\bfnm{Benjamin~D.}\binits{B.~D.}}
(\byear{2017}).
\btitle{Calibration of Stochastic Computer Simulators Using Likelihood
  Emulation}.
\bjournal{Technometrics}
\bvolume{59}
\bpages{80--92}.
\bdoi{10.1080/00401706.2015.1125391}
\end{barticle}
\endbibitem

\bibitem[\protect\citeauthoryear{Osband et~al.}{2023a}]{epistemicNN}
\begin{binproceedings}[author]
\bauthor{\bsnm{Osband},~\bfnm{Ian}\binits{I.}},
  \bauthor{\bsnm{Wen},~\bfnm{Zheng}\binits{Z.}},
  \bauthor{\bsnm{Asghari},~\bfnm{Seyed~Mohammad}\binits{S.~M.}},
  \bauthor{\bsnm{Dwaracherla},~\bfnm{Vikranth}\binits{V.}},
  \bauthor{\bsnm{IBRAHIMI},~\bfnm{MORTEZA}\binits{M.}},
  \bauthor{\bsnm{Lu},~\bfnm{Xiuyuan}\binits{X.}} \AND
  \bauthor{\bsnm{Van~Roy},~\bfnm{Benjamin}\binits{B.}}
(\byear{2023}a).
\btitle{Epistemic Neural Networks}.
In \bbooktitle{Advances in Neural Information Processing Systems}
(\beditor{\bfnm{A.}\binits{A.}~\bsnm{Oh}},
  \beditor{\bfnm{T.}\binits{T.}~\bsnm{Naumann}},
  \beditor{\bfnm{A.}\binits{A.}~\bsnm{Globerson}},
  \beditor{\bfnm{K.}\binits{K.}~\bsnm{Saenko}},
  \beditor{\bfnm{M.}\binits{M.}~\bsnm{Hardt}} \AND
  \beditor{\bfnm{S.}\binits{S.}~\bsnm{Levine}}, eds.)
\bvolume{36}
\bpages{2795--2823}.
\bpublisher{Curran Associates, Inc.}
\bdoi{10.48550/arXiv.2107.08924}
\end{binproceedings}
\endbibitem

\bibitem[\protect\citeauthoryear{Osband et~al.}{2023b}]{BayesOptEpistemicNN}
\begin{binproceedings}[author]
\bauthor{\bsnm{Osband},~\bfnm{Ian}\binits{I.}},
  \bauthor{\bsnm{Wen},~\bfnm{Zheng}\binits{Z.}},
  \bauthor{\bsnm{Asghari},~\bfnm{Seyed~Mohammad}\binits{S.~M.}},
  \bauthor{\bsnm{Dwaracherla},~\bfnm{Vikranth}\binits{V.}},
  \bauthor{\bsnm{Ibrahimi},~\bfnm{Morteza}\binits{M.}},
  \bauthor{\bsnm{Lu},~\bfnm{Xiuyuan}\binits{X.}} \AND
  \bauthor{\bsnm{Van~Roy},~\bfnm{Benjamin}\binits{B.}}
(\byear{2023}b).
\btitle{Approximate {T}hompson Sampling via Epistemic Neural Networks}.
In \bbooktitle{Proceedings of the Thirty-Ninth Conference on Uncertainty in
  Artificial Intelligence}
(\beditor{\bfnm{Robin~J.}\binits{R.~J.}~\bsnm{Evans}} \AND
  \beditor{\bfnm{Ilya}\binits{I.}~\bsnm{Shpitser}}, eds.).
\bseries{Proceedings of Machine Learning Research}
\bvolume{216}
\bpages{1586--1595}.
\bpublisher{PMLR}.
\bdoi{10.48550/arXiv.2302.09205}
\end{binproceedings}
\endbibitem

\bibitem[\protect\citeauthoryear{Owen}{2013}]{OwenMCBook}
\begin{bbook}[author]
\bauthor{\bsnm{Owen},~\bfnm{Art~B.}\binits{A.~B.}}
(\byear{2013}).
\btitle{Monte Carlo theory, methods and examples}.
\bpublisher{\url{https://artowen.su.domains/mc/}}.
\end{bbook}
\endbibitem

\bibitem[\protect\citeauthoryear{Pan et~al.}{2025}]{ESMReview}
\begin{barticle}[author]
\bauthor{\bsnm{Pan},~\bfnm{Xiaoduo}\binits{X.}},
  \bauthor{\bsnm{Chen},~\bfnm{Deliang}\binits{D.}},
  \bauthor{\bsnm{Pan},~\bfnm{Baoxiang}\binits{B.}},
  \bauthor{\bsnm{Huang},~\bfnm{Xiaozhong}\binits{X.}},
  \bauthor{\bsnm{Yang},~\bfnm{Kun}\binits{K.}},
  \bauthor{\bsnm{Piao},~\bfnm{Shilong}\binits{S.}},
  \bauthor{\bsnm{Zhou},~\bfnm{Tianjun}\binits{T.}},
  \bauthor{\bsnm{Dai},~\bfnm{Yongjiu}\binits{Y.}},
  \bauthor{\bsnm{Chen},~\bfnm{Fahu}\binits{F.}} \AND
  \bauthor{\bsnm{Li},~\bfnm{Xin}\binits{X.}}
(\byear{2025}).
\btitle{Evolution and prospects of Earth system models: Challenges and
  opportunities}.
\bjournal{Earth-Science Reviews}
\bvolume{260}
\bpages{104986}.
\bdoi{https://doi.org/10.1016/j.earscirev.2024.104986}
\end{barticle}
\endbibitem

\bibitem[\protect\citeauthoryear{Papamakarios and Murray}{2016}]{murrayNPE}
\begin{binproceedings}[author]
\bauthor{\bsnm{Papamakarios},~\bfnm{George}\binits{G.}} \AND
  \bauthor{\bsnm{Murray},~\bfnm{Iain}\binits{I.}}
(\byear{2016}).
\btitle{Fast epsilon-free Inference of Simulation Models with Bayesian
  Conditional Density Estimation}.
In \bbooktitle{Advances in Neural Information Processing Systems}
(\beditor{\bfnm{D.}\binits{D.}~\bsnm{Lee}},
  \beditor{\bfnm{M.}\binits{M.}~\bsnm{Sugiyama}},
  \beditor{\bfnm{U.}\binits{U.}~\bsnm{Luxburg}},
  \beditor{\bfnm{I.}\binits{I.}~\bsnm{Guyon}} \AND
  \beditor{\bfnm{R.}\binits{R.}~\bsnm{Garnett}}, eds.)
\bvolume{29}.
\bpublisher{Curran Associates, Inc.}
\bdoi{10.48550/arXiv.1605.06376}
\end{binproceedings}
\endbibitem

\bibitem[\protect\citeauthoryear{Papamakarios, Sterratt and
  Murray}{2019}]{SNLE}
\begin{binproceedings}[author]
\bauthor{\bsnm{Papamakarios},~\bfnm{George}\binits{G.}},
  \bauthor{\bsnm{Sterratt},~\bfnm{David}\binits{D.}} \AND
  \bauthor{\bsnm{Murray},~\bfnm{Iain}\binits{I.}}
(\byear{2019}).
\btitle{Sequential Neural Likelihood: Fast Likelihood-free Inference with
  Autoregressive Flows}.
In \bbooktitle{Proceedings of the Twenty-Second International Conference on
  Artificial Intelligence and Statistics}
(\beditor{\bfnm{Kamalika}\binits{K.}~\bsnm{Chaudhuri}} \AND
  \beditor{\bfnm{Masashi}\binits{M.}~\bsnm{Sugiyama}}, eds.).
\bseries{Proceedings of Machine Learning Research}
\bvolume{89}
\bpages{837--848}.
\bpublisher{PMLR}.
\end{binproceedings}
\endbibitem

\bibitem[\protect\citeauthoryear{Paun and Husmeier}{2021}]{MCMC_GP_proposal}
\begin{barticle}[author]
\bauthor{\bsnm{Paun},~\bfnm{L.~Mihaela}\binits{L.~M.}} \AND
  \bauthor{\bsnm{Husmeier},~\bfnm{Dirk}\binits{D.}}
(\byear{2021}).
\btitle{Markov chain Monte Carlo with Gaussian processes for fast parameter
  estimation and uncertainty quantification in a 1D fluid-dynamics model of the
  pulmonary circulation}.
\bjournal{International Journal for Numerical Methods in Biomedical
  Engineering}
\bvolume{37}
\bpages{e3421}.
\bdoi{https://doi.org/10.1002/cnm.3421}
\end{barticle}
\endbibitem

\bibitem[\protect\citeauthoryear{Picchini, Simola and
  Corander}{2023}]{PicchiniABCEKI}
\begin{barticle}[author]
\bauthor{\bsnm{Picchini},~\bfnm{Umberto}\binits{U.}},
  \bauthor{\bsnm{Simola},~\bfnm{Umberto}\binits{U.}} \AND
  \bauthor{\bsnm{Corander},~\bfnm{Jukka}\binits{J.}}
(\byear{2023}).
\btitle{{Sequentially Guided MCMC Proposals for Synthetic Likelihoods and
  Correlated Synthetic Likelihoods}}.
\bjournal{Bayesian Analysis}
\bvolume{18}
\bpages{1099 -- 1129}.
\bdoi{10.1214/22-BA1305}
\end{barticle}
\endbibitem

\bibitem[\protect\citeauthoryear{Plummer}{2015}]{PlummerCut}
\begin{barticle}[author]
\bauthor{\bsnm{Plummer},~\bfnm{Martyn}\binits{M.}}
(\byear{2015}).
\btitle{Cuts in Bayesian graphical models}.
\bvolume{25}
\bpages{37–43}.
\bdoi{10.1007/s11222-014-9503-z}
\end{barticle}
\endbibitem

\bibitem[\protect\citeauthoryear{Price et~al.}{2018}]{PriceSL}
\begin{barticle}[author]
\bauthor{\bsnm{Price},~\bfnm{L.~F.}\binits{L.~F.}},
  \bauthor{\bsnm{Drovandi},~\bfnm{C.~C.}\binits{C.~C.}},
  \bauthor{\bsnm{Lee},~\bfnm{A.}\binits{A.}} \AND
  \bauthor{\bsnm{Nott},~\bfnm{D.~J.}\binits{D.~J.}}
(\byear{2018}).
\btitle{Bayesian Synthetic Likelihood}.
\bjournal{Journal of Computational and Graphical Statistics}
\bvolume{27}
\bpages{1--11}.
\bdoi{10.1080/10618600.2017.1302882}
\end{barticle}
\endbibitem

\bibitem[\protect\citeauthoryear{Pritam~Ranjan and
  Michailidis}{2008}]{contourEstimation}
\begin{barticle}[author]
\bauthor{\bsnm{Pritam~Ranjan},~\bfnm{Derek~Bingham}\binits{D.~B.}} \AND
  \bauthor{\bsnm{Michailidis},~\bfnm{George}\binits{G.}}
(\byear{2008}).
\btitle{Sequential Experiment Design for Contour Estimation From Complex
  Computer Codes}.
\bjournal{Technometrics}
\bvolume{50}
\bpages{527--541}.
\bdoi{10.1198/004017008000000541}
\end{barticle}
\endbibitem

\bibitem[\protect\citeauthoryear{Radev et~al.}{2020}]{bayesflow_2020_original}
\begin{barticle}[author]
\bauthor{\bsnm{Radev},~\bfnm{Stefan~T.}\binits{S.~T.}},
  \bauthor{\bsnm{Mertens},~\bfnm{Ulf~K.}\binits{U.~K.}},
  \bauthor{\bsnm{Voss},~\bfnm{Andreas}\binits{A.}},
  \bauthor{\bsnm{Ardizzone},~\bfnm{Lynton}\binits{L.}} \AND
  \bauthor{\bsnm{K{\"o}the},~\bfnm{Ullrich}\binits{U.}}
(\byear{2020}).
\btitle{{BayesFlow}: Learning complex stochastic models with invertible neural
  networks}.
\bjournal{IEEE transactions on neural networks and learning systems}
\bvolume{33}
\bpages{1452--1466}.
\end{barticle}
\endbibitem

\bibitem[\protect\citeauthoryear{Ranjan et~al.}{2016}]{scalarization}
\begin{bmisc}[author]
\bauthor{\bsnm{Ranjan},~\bfnm{Pritam}\binits{P.}},
  \bauthor{\bsnm{Thomas},~\bfnm{Mark}\binits{M.}},
  \bauthor{\bsnm{Teismann},~\bfnm{Holger}\binits{H.}} \AND
  \bauthor{\bsnm{Mukhoti},~\bfnm{Sujay}\binits{S.}}
(\byear{2016}).
\btitle{Inverse problem for time-series valued computer model via
  scalarization}.
\bdoi{10.48550/arXiv.1605.09503}
\end{bmisc}
\endbibitem

\bibitem[\protect\citeauthoryear{Raoult et~al.}{2025}]{paramLSM}
\begin{barticle}[author]
\bauthor{\bsnm{Raoult},~\bfnm{Nina}\binits{N.}},
  \bauthor{\bsnm{Douglas},~\bfnm{Natalie}\binits{N.}},
  \bauthor{\bsnm{MacBean},~\bfnm{Natasha}\binits{N.}},
  \bauthor{\bsnm{Kolassa},~\bfnm{Jana}\binits{J.}},
  \bauthor{\bsnm{Quaife},~\bfnm{Tristan}\binits{T.}},
  \bauthor{\bsnm{Roberts},~\bfnm{Andrew~G.}\binits{A.~G.}},
  \bauthor{\bsnm{Fisher},~\bfnm{Rosie}\binits{R.}},
  \bauthor{\bsnm{Fer},~\bfnm{Istem}\binits{I.}},
  \bauthor{\bsnm{Bacour},~\bfnm{Cédric}\binits{C.}},
  \bauthor{\bsnm{Dagon},~\bfnm{Katherine}\binits{K.}},
  \bauthor{\bsnm{Hawkins},~\bfnm{Linnia}\binits{L.}},
  \bauthor{\bsnm{Carvalhais},~\bfnm{Nuno}\binits{N.}},
  \bauthor{\bsnm{Cooper},~\bfnm{Elizabeth}\binits{E.}},
  \bauthor{\bsnm{Dietze},~\bfnm{Michael~C.}\binits{M.~C.}},
  \bauthor{\bsnm{Gentine},~\bfnm{Pierre}\binits{P.}},
  \bauthor{\bsnm{Kaminski},~\bfnm{Thomas}\binits{T.}},
  \bauthor{\bsnm{Kennedy},~\bfnm{Daniel}\binits{D.}},
  \bauthor{\bsnm{Liddy},~\bfnm{Hannah~M.}\binits{H.~M.}},
  \bauthor{\bsnm{Moore},~\bfnm{David J.~P.}\binits{D.~J.~P.}},
  \bauthor{\bsnm{Peylin},~\bfnm{Philippe}\binits{P.}},
  \bauthor{\bsnm{Pinnington},~\bfnm{Ewan}\binits{E.}},
  \bauthor{\bsnm{Sanderson},~\bfnm{Benjamin}\binits{B.}},
  \bauthor{\bsnm{Scholze},~\bfnm{Marko}\binits{M.}},
  \bauthor{\bsnm{Seiler},~\bfnm{Christian}\binits{C.}},
  \bauthor{\bsnm{Smallman},~\bfnm{T.~Luke}\binits{T.~L.}},
  \bauthor{\bsnm{Vergopolan},~\bfnm{Noemi}\binits{N.}},
  \bauthor{\bsnm{Viskari},~\bfnm{Toni}\binits{T.}},
  \bauthor{\bsnm{Williams},~\bfnm{Mathew}\binits{M.}} \AND
  \bauthor{\bsnm{Zobitz},~\bfnm{John}\binits{J.}}
(\byear{2025}).
\btitle{Parameter Estimation in Land Surface Models: Challenges and
  Opportunities With Data Assimilation and Machine Learning}.
\bjournal{Journal of Advances in Modeling Earth Systems}
\bvolume{17}
\bpages{e2024MS004733}.
\bnote{e2024MS004733 2024MS004733}.
\bdoi{https://doi.org/10.1029/2024MS004733}
\end{barticle}
\endbibitem

\bibitem[\protect\citeauthoryear{Rasmussen}{2003}]{RasmussenGPHMC}
\begin{bincollection}[author]
\bauthor{\bsnm{Rasmussen},~\bfnm{Carl~Edward}\binits{C.~E.}}
(\byear{2003}).
\btitle{Gaussian Processes to Speed up Hybrid {M}onte {C}arlo for Expensive
  {B}ayesian Integrals}.
In \bbooktitle{Bayesian Statistics 7: Proceedings of the Seventh Valencia
  International Meeting}
(\beditor{\bfnm{J.~M.}\binits{J.~M.}~\bsnm{Bernardo}},
  \beditor{\bfnm{M.~J.}\binits{M.~J.}~\bsnm{Bayarri}},
  \beditor{\bfnm{J.~O.}\binits{J.~O.}~\bsnm{Berger}},
  \beditor{\bfnm{A.~P.}\binits{A.~P.}~\bsnm{Dawid}},
  \beditor{\bfnm{D.}\binits{D.}~\bsnm{Heckerman}},
  \beditor{\bfnm{A.~F.~M.}\binits{A.~F.~M.}~\bsnm{Smith}} \AND
  \beditor{\bfnm{M.}\binits{M.}~\bsnm{West}}, eds.)
\bpages{651--659}.
\bpublisher{Oxford University Press}, \baddress{Oxford, UK}.
\bdoi{10.1093/oso/9780198526155.003.0045}
\end{bincollection}
\endbibitem

\bibitem[\protect\citeauthoryear{Rasmussen}{2004}]{gpML}
\begin{binbook}[author]
\bauthor{\bsnm{Rasmussen},~\bfnm{Carl~Edward}\binits{C.~E.}}
(\byear{2004}).
\btitle{Gaussian Processes in Machine Learning}
In \bbooktitle{Advanced Lectures on Machine Learning: ML Summer Schools 2003,
  Canberra, Australia, February 2 - 14, 2003, T{\"u}bingen, Germany, August 4 -
  16, 2003, Revised Lectures}
\bpages{63--71}.
\bpublisher{Springer Berlin Heidelberg}, \baddress{Berlin, Heidelberg}.
\bdoi{10.1007/978-3-540-28650-9_4}
\end{binbook}
\endbibitem

\bibitem[\protect\citeauthoryear{Ray et~al.}{2015}]{CLMSurrogates}
\begin{barticle}[author]
\bauthor{\bsnm{Ray},~\bfnm{J.}\binits{J.}},
  \bauthor{\bsnm{Hou},~\bfnm{Z.}\binits{Z.}},
  \bauthor{\bsnm{Huang},~\bfnm{M.}\binits{M.}},
  \bauthor{\bsnm{Sargsyan},~\bfnm{K.}\binits{K.}} \AND
  \bauthor{\bsnm{Swiler},~\bfnm{L.}\binits{L.}}
(\byear{2015}).
\btitle{Bayesian Calibration of the Community Land Model Using Surrogates}.
\bjournal{SIAM/ASA Journal on Uncertainty Quantification}
\bvolume{3}
\bpages{199-233}.
\bdoi{10.1137/140957998}
\end{barticle}
\endbibitem

\bibitem[\protect\citeauthoryear{Reiser et~al.}{2025}]{BurknerSurrogate}
\begin{barticle}[author]
\bauthor{\bsnm{Reiser},~\bfnm{Philipp}\binits{P.}},
  \bauthor{\bsnm{Aguilar},~\bfnm{Javier~Enrique}\binits{J.~E.}},
  \bauthor{\bsnm{Guthke},~\bfnm{Anneli}\binits{A.}} \AND
  \bauthor{\bsnm{Bürkner},~\bfnm{Paul-Christian}\binits{P.-C.}}
(\byear{2025}).
\btitle{Uncertainty quantification and propagation in surrogate-based Bayesian
  inference}.
\bjournal{Statistics and Computing}
\bvolume{35}.
\bdoi{10.1007/s11222-025-10597-8}
\end{barticle}
\endbibitem

\bibitem[\protect\citeauthoryear{Riccius et~al.}{2026}]{ActiveLearningMCMC}
\begin{barticle}[author]
\bauthor{\bsnm{Riccius},~\bfnm{Leon}\binits{L.}},
  \bauthor{\bsnm{Rocha},~\bfnm{Iuri B. C.~M.}\binits{I.~B. C.~M.}},
  \bauthor{\bsnm{Bierkens},~\bfnm{Joris}\binits{J.}},
  \bauthor{\bsnm{Kekkonen},~\bfnm{Hanne}\binits{H.}} \AND \bauthor{\bsnm{{van
  der Meer}},~\bfnm{Frans~P.}\binits{F.~P.}}
(\byear{2026}).
\btitle{Integration of active learning and MCMC sampling for efficient Bayesian
  calibration of mechanical properties}.
\bjournal{European Journal of Mechanics - A/Solids}
\bvolume{117}
\bpages{106015}.
\bdoi{https://doi.org/10.1016/j.euromechsol.2026.106015}
\end{barticle}
\endbibitem

\bibitem[\protect\citeauthoryear{Riis et~al.}{2022}]{fullyBayesianGPs}
\begin{binproceedings}[author]
\bauthor{\bsnm{Riis},~\bfnm{Christoffer}\binits{C.}},
  \bauthor{\bsnm{Antunes},~\bfnm{Francisco}\binits{F.}},
  \bauthor{\bsnm{H\"{u}ttel},~\bfnm{Frederik~Boe}\binits{F.~B.}},
  \bauthor{\bsnm{Azevedo},~\bfnm{Carlos~Lima}\binits{C.~L.}} \AND
  \bauthor{\bsnm{Pereira},~\bfnm{Francisco~C\^{a}mara}\binits{F.~C.}}
(\byear{2022}).
\btitle{Bayesian active learning with fully Bayesian Gaussian processes}.
In \bbooktitle{Advances in Neural Information Processing Systems}.
\bseries{NeurIPS '22}
\bvolume{35}.
\bpublisher{Curran Associates Inc.}, \baddress{Red Hook, NY, USA}.
\end{binproceedings}
\endbibitem

\bibitem[\protect\citeauthoryear{Roberts, Dietze and
  Huggins}{2026}]{RobertsUncProp}
\begin{bmisc}[author]
\bauthor{\bsnm{Roberts},~\bfnm{Andrew~Gerard}\binits{A.~G.}},
  \bauthor{\bsnm{Dietze},~\bfnm{Michael}\binits{M.}} \AND
  \bauthor{\bsnm{Huggins},~\bfnm{Jonathan~H.}\binits{J.~H.}}
(\byear{2026}).
\btitle{Propagating Surrogate Uncertainty in Bayesian Inverse Problems}.
\bdoi{10.48550/arXiv.2601.03532}
\end{bmisc}
\endbibitem

\bibitem[\protect\citeauthoryear{Sacks
  et~al.}{1989}]{design_analysis_computer_experiments}
\begin{barticle}[author]
\bauthor{\bsnm{Sacks},~\bfnm{Jerome}\binits{J.}},
  \bauthor{\bsnm{Welch},~\bfnm{William~J.}\binits{W.~J.}},
  \bauthor{\bsnm{Mitchell},~\bfnm{Toby~J.}\binits{T.~J.}} \AND
  \bauthor{\bsnm{Wynn},~\bfnm{Henry~P.}\binits{H.~P.}}
(\byear{1989}).
\btitle{{Design and Analysis of Computer Experiments}}.
\bjournal{Statistical Science}
\bvolume{4}
\bpages{409 -- 423}.
\bdoi{10.1214/ss/1177012413}
\end{barticle}
\endbibitem

\bibitem[\protect\citeauthoryear{Satner, Williams and
  Notz}{2018}]{SanterCompExp}
\begin{barticle}[author]
\bauthor{\bsnm{Satner},~\bfnm{Thomas~J.}\binits{T.~J.}},
  \bauthor{\bsnm{Williams},~\bfnm{Brian~J.}\binits{B.~J.}} \AND
  \bauthor{\bsnm{Notz},~\bfnm{William~I.}\binits{W.~I.}}
(\byear{2018}).
\btitle{{The Design and Analysis of Computer Experiments}}.
\bnote{\url{https://doi.org/10.1007/978-1-4939-8847-1}}.
\end{barticle}
\endbibitem

\bibitem[\protect\citeauthoryear{Sauer, Cooper and
  Gramacy}{2023}]{deepGPVecchia}
\begin{barticle}[author]
\bauthor{\bsnm{Sauer},~\bfnm{Annie}\binits{A.}},
  \bauthor{\bsnm{Cooper},~\bfnm{Andrew}\binits{A.}} \AND
  \bauthor{\bsnm{Gramacy},~\bfnm{Robert~B.}\binits{R.~B.}}
(\byear{2023}).
\btitle{Vecchia-Approximated Deep Gaussian Processes for Computer Experiments}.
\bjournal{Journal of Computational and Graphical Statistics}
\bvolume{32}
\bpages{824--837}.
\bdoi{10.1080/10618600.2022.2129662}
\end{barticle}
\endbibitem

\bibitem[\protect\citeauthoryear{Sauer, Gramacy and Higdon}{2023}]{deepGPAL}
\begin{barticle}[author]
\bauthor{\bsnm{Sauer},~\bfnm{Annie}\binits{A.}},
  \bauthor{\bsnm{Gramacy},~\bfnm{Robert~B.}\binits{R.~B.}} \AND
  \bauthor{\bsnm{Higdon},~\bfnm{David}\binits{D.}}
(\byear{2023}).
\btitle{Active Learning for Deep Gaussian Process Surrogates}.
\bjournal{Technometrics}
\bvolume{65}
\bpages{4--18}.
\bdoi{10.1080/00401706.2021.2008505}
\end{barticle}
\endbibitem

\bibitem[\protect\citeauthoryear{Scheurer
  et~al.}{2026}]{BurknerAmortizedSurrogate}
\begin{barticle}[author]
\bauthor{\bsnm{Scheurer},~\bfnm{Stefania}\binits{S.}},
  \bauthor{\bsnm{Reiser},~\bfnm{Philipp}\binits{P.}},
  \bauthor{\bsnm{Br{\"u}nnette},~\bfnm{Tim}\binits{T.}},
  \bauthor{\bsnm{Nowak},~\bfnm{Wolfgang}\binits{W.}},
  \bauthor{\bsnm{Guthke},~\bfnm{Anneli}\binits{A.}} \AND
  \bauthor{\bsnm{B{\"u}rkner},~\bfnm{Paul-Christian}\binits{P.-C.}}
(\byear{2026}).
\btitle{Uncertainty-Aware Surrogate-based Amortized Bayesian Inference for
  Computationally Expensive Models}.
\bjournal{Transactions on Machine Learning Research}.
\bdoi{10.48550/arXiv.2505.0868}
\end{barticle}
\endbibitem

\bibitem[\protect\citeauthoryear{Schneider et~al.}{2017}]{ESM_modeling_2pt0}
\begin{barticle}[author]
\bauthor{\bsnm{Schneider},~\bfnm{Tapio}\binits{T.}},
  \bauthor{\bsnm{Lan},~\bfnm{Shiwei}\binits{S.}},
  \bauthor{\bsnm{Stuart},~\bfnm{Andrew}\binits{A.}} \AND
  \bauthor{\bsnm{Teixeira},~\bfnm{João}\binits{J.}}
(\byear{2017}).
\btitle{Earth System Modeling 2.0: A Blueprint for Models That Learn From
  Observations and Targeted High-Resolution Simulations}.
\bjournal{Geophysical Research Letters}
\bvolume{44}
\bpages{12,396-12,417}.
\bdoi{https://doi.org/10.1002/2017GL076101}
\end{barticle}
\endbibitem

\bibitem[\protect\citeauthoryear{Shahriari et~al.}{2016}]{reviewBayesOpt}
\begin{barticle}[author]
\bauthor{\bsnm{Shahriari},~\bfnm{Bobak}\binits{B.}},
  \bauthor{\bsnm{Swersky},~\bfnm{Kevin}\binits{K.}},
  \bauthor{\bsnm{Wang},~\bfnm{Ziyu}\binits{Z.}},
  \bauthor{\bsnm{Adams},~\bfnm{Ryan~P.}\binits{R.~P.}} \AND
  \bauthor{\bparticle{de} \bsnm{Freitas},~\bfnm{Nando}\binits{N.}}
(\byear{2016}).
\btitle{Taking the Human Out of the Loop: A Review of Bayesian Optimization}.
\bjournal{Proceedings of the IEEE}
\bvolume{104}
\bpages{148-175}.
\bdoi{10.1109/JPROC.2015.2494218}
\end{barticle}
\endbibitem

\bibitem[\protect\citeauthoryear{Sharma et~al.}{2023}]{partialBNN}
\begin{binproceedings}[author]
\bauthor{\bsnm{Sharma},~\bfnm{Mrinank}\binits{M.}},
  \bauthor{\bsnm{Farquhar},~\bfnm{Sebastian}\binits{S.}},
  \bauthor{\bsnm{Nalisnick},~\bfnm{Eric}\binits{E.}} \AND
  \bauthor{\bsnm{Rainforth},~\bfnm{Tom}\binits{T.}}
(\byear{2023}).
\btitle{Do Bayesian Neural Networks Need To Be Fully Stochastic?}
In \bbooktitle{Proceedings of The 26th International Conference on Artificial
  Intelligence and Statistics}
(\beditor{\bfnm{Francisco}\binits{F.}~\bsnm{Ruiz}},
  \beditor{\bfnm{Jennifer}\binits{J.}~\bsnm{Dy}} \AND
  \beditor{\bfnm{Jan-Willem}\binits{J.-W.}~\bparticle{van~de} \bsnm{Meent}},
  eds.).
\bseries{Proceedings of Machine Learning Research}
\bvolume{206}
\bpages{7694--7722}.
\bpublisher{PMLR}.
\end{binproceedings}
\endbibitem

\bibitem[\protect\citeauthoryear{Sinsbeck and Nowak}{2017}]{SinsbeckNowak}
\begin{barticle}[author]
\bauthor{\bsnm{Sinsbeck},~\bfnm{Michael}\binits{M.}} \AND
  \bauthor{\bsnm{Nowak},~\bfnm{Wolfgang}\binits{W.}}
(\byear{2017}).
\btitle{Sequential Design of Computer Experiments for the Solution of Bayesian
  Inverse Problems}.
\bjournal{SIAM/ASA Journal on Uncertainty Quantification}
\bvolume{5}
\bpages{640-664}.
\bdoi{10.1137/15M1047659}
\end{barticle}
\endbibitem

\bibitem[\protect\citeauthoryear{Stuart}{2010}]{Stuart_BIP}
\begin{barticle}[author]
\bauthor{\bsnm{Stuart},~\bfnm{A.~M.}\binits{A.~M.}}
(\byear{2010}).
\btitle{Inverse problems: A Bayesian perspective}.
\bjournal{Acta Numerica}
\bvolume{19}
\bpages{451-559}.
\bdoi{10.1017/S0962492910000061}
\end{barticle}
\endbibitem

\bibitem[\protect\citeauthoryear{Stuart and Teckentrup}{2018}]{StuartTeck1}
\begin{barticle}[author]
\bauthor{\bsnm{Stuart},~\bfnm{Andrew~M.}\binits{A.~M.}} \AND
  \bauthor{\bsnm{Teckentrup},~\bfnm{Aretha~L.}\binits{A.~L.}}
(\byear{2018}).
\btitle{Posterior consistency for Gaussian process approximations of Bayesian
  posterior distributions}.
\bjournal{Mathematics of Computation}
\bvolume{87}
\bpages{721--753}.
\bdoi{10.1090/mcom/3244}
\end{barticle}
\endbibitem

\bibitem[\protect\citeauthoryear{Sung and
  Tuo}{2024}]{computerModelCalibrationReview}
\begin{barticle}[author]
\bauthor{\bsnm{Sung},~\bfnm{Chih-Li}\binits{C.-L.}} \AND
  \bauthor{\bsnm{Tuo},~\bfnm{Rui}\binits{R.}}
(\byear{2024}).
\btitle{A review on computer model calibration}.
\bjournal{WIREs Computational Statistics}
\bvolume{16}
\bpages{e1645}.
\bdoi{https://doi.org/10.1002/wics.1645}
\end{barticle}
\endbibitem

\bibitem[\protect\citeauthoryear{Sunn{\aa}ker et~al.}{2013}]{ABCPrimer}
\begin{barticle}[author]
\bauthor{\bsnm{Sunn{\aa}ker},~\bfnm{Mikael}\binits{M.}},
  \bauthor{\bsnm{Busetto},~\bfnm{Alberto~Giovanni}\binits{A.~G.}},
  \bauthor{\bsnm{Numminen},~\bfnm{Elina}\binits{E.}},
  \bauthor{\bsnm{Corander},~\bfnm{Jukka}\binits{J.}},
  \bauthor{\bsnm{Foll},~\bfnm{Matthieu}\binits{M.}} \AND
  \bauthor{\bsnm{Dessimoz},~\bfnm{Christophe}\binits{C.}}
(\byear{2013}).
\btitle{Approximate {B}ayesian computation}.
\bjournal{PLOS Computational Biology}
\bvolume{9}
\bpages{e1002803}.
\bdoi{10.1371/journal.pcbi.1002803}
\end{barticle}
\endbibitem

\bibitem[\protect\citeauthoryear{Takhtaganov and
  M{\"u}ller}{2018}]{Takhtaganov2018}
\begin{barticle}[author]
\bauthor{\bsnm{Takhtaganov},~\bfnm{Timur}\binits{T.}} \AND
  \bauthor{\bsnm{M{\"u}ller},~\bfnm{Juliane}\binits{J.}}
(\byear{2018}).
\btitle{Adaptive Gaussian process surrogates for Bayesian inference}.
\bjournal{ArXiv}
\bvolume{abs/1809.10784}.
\bdoi{10.48550/arXiv.1809.10784}
\end{barticle}
\endbibitem

\bibitem[\protect\citeauthoryear{Teckentrup}{2020}]{TeckHyperpar}
\begin{barticle}[author]
\bauthor{\bsnm{Teckentrup},~\bfnm{Aretha~L.}\binits{A.~L.}}
(\byear{2020}).
\btitle{Convergence of Gaussian Process Regression with Estimated
  Hyper-Parameters and Applications in Bayesian Inverse Problems}.
\bjournal{SIAM/ASA Journal on Uncertainty Quantification}
\bvolume{8}
\bpages{1310-1337}.
\bdoi{10.1137/19M1284816}
\end{barticle}
\endbibitem

\bibitem[\protect\citeauthoryear{Tran et~al.}{2019}]{BayesLastLayer}
\begin{binproceedings}[author]
\bauthor{\bsnm{Tran},~\bfnm{Dustin}\binits{D.}},
  \bauthor{\bsnm{Dusenberry},~\bfnm{Mike}\binits{M.}},
  \bauthor{\bparticle{van~der} \bsnm{Wilk},~\bfnm{Mark}\binits{M.}} \AND
  \bauthor{\bsnm{Hafner},~\bfnm{Danijar}\binits{D.}}
(\byear{2019}).
\btitle{Bayesian Layers: A Module for Neural Network Uncertainty}.
In \bbooktitle{Advances in Neural Information Processing Systems}
(\beditor{\bfnm{H.}\binits{H.}~\bsnm{Wallach}},
  \beditor{\bfnm{H.}\binits{H.}~\bsnm{Larochelle}},
  \beditor{\bfnm{A.}\binits{A.}~\bsnm{Beygelzimer}},
  \beditor{\bfnm{F.}\binits{F.}~\bparticle{d\textquotesingle}
  \bsnm{Alch\'{e}-Buc}}, \beditor{\bfnm{E.}\binits{E.}~\bsnm{Fox}} \AND
  \beditor{\bfnm{R.}\binits{R.}~\bsnm{Garnett}}, eds.)
\bvolume{32}.
\bpublisher{Curran Associates, Inc.}
\bdoi{10.48550/arXiv.1812.03973}
\end{binproceedings}
\endbibitem

\bibitem[\protect\citeauthoryear{Varma et~al.}{2019}]{VarmaBlackHole2019}
\begin{barticle}[author]
\bauthor{\bsnm{Varma},~\bfnm{Vijay}\binits{V.}},
  \bauthor{\bsnm{Field},~\bfnm{Scott~E.}\binits{S.~E.}},
  \bauthor{\bsnm{Scheel},~\bfnm{Mark~A.}\binits{M.~A.}},
  \bauthor{\bsnm{Blackman},~\bfnm{Jonathan}\binits{J.}},
  \bauthor{\bsnm{Gerosa},~\bfnm{Davide}\binits{D.}},
  \bauthor{\bsnm{Stein},~\bfnm{Leo~C.}\binits{L.~C.}},
  \bauthor{\bsnm{Kidder},~\bfnm{Lawrence~E.}\binits{L.~E.}} \AND
  \bauthor{\bsnm{Pfeiffer},~\bfnm{Harald~P.}\binits{H.~P.}}
(\byear{2019}).
\btitle{Surrogate models for precessing binary black hole simulations with
  unequal masses}.
\bjournal{Physical Review Research}
\bvolume{1}.
\bdoi{10.1103/physrevresearch.1.033015}
\end{barticle}
\endbibitem

\bibitem[\protect\citeauthoryear{Villani, Unger and
  Weiser}{2024}]{VillaniAdaptiveGP}
\begin{bmisc}[author]
\bauthor{\bsnm{Villani},~\bfnm{Paolo}\binits{P.}},
  \bauthor{\bsnm{Unger},~\bfnm{Jörg}\binits{J.}} \AND
  \bauthor{\bsnm{Weiser},~\bfnm{Martin}\binits{M.}}
(\byear{2024}).
\btitle{Adaptive Gaussian Process Regression for Bayesian inverse problems}.
\bdoi{10.48550/arXiv.2404.19459}
\end{bmisc}
\endbibitem

\bibitem[\protect\citeauthoryear{Wang, Bui-Thanh and Ghattas}{2018}]{randMAP}
\begin{barticle}[author]
\bauthor{\bsnm{Wang},~\bfnm{Kainan}\binits{K.}},
  \bauthor{\bsnm{Bui-Thanh},~\bfnm{Tan}\binits{T.}} \AND
  \bauthor{\bsnm{Ghattas},~\bfnm{Omar}\binits{O.}}
(\byear{2018}).
\btitle{A Randomized Maximum A Posteriori Method for Posterior Sampling of High
  Dimensional Nonlinear Bayesian Inverse Problems}.
\bjournal{SIAM Journal on Scientific Computing}
\bvolume{40}
\bpages{A142-A171}.
\bdoi{10.1137/16M1060625}
\end{barticle}
\endbibitem

\bibitem[\protect\citeauthoryear{Wang and Li}{2018}]{wang2018adaptive}
\begin{barticle}[author]
\bauthor{\bsnm{Wang},~\bfnm{Hongqiao}\binits{H.}} \AND
  \bauthor{\bsnm{Li},~\bfnm{Jinglai}\binits{J.}}
(\byear{2018}).
\btitle{Adaptive gaussian process approximation for bayesian inference with
  expensive likelihood functions}.
\bjournal{Neural Comput.}
\bvolume{30}
\bpages{3072–3094}.
\bdoi{10.1162/neco_a_01127}
\end{barticle}
\endbibitem

\bibitem[\protect\citeauthoryear{Wilkinson}{2014}]{WilkinsonABCGP}
\begin{binproceedings}[author]
\bauthor{\bsnm{Wilkinson},~\bfnm{Richard}\binits{R.}}
(\byear{2014}).
\btitle{{Accelerating ABC methods using Gaussian processes}}.
In \bbooktitle{Proceedings of the Seventeenth International Conference on
  Artificial Intelligence and Statistics}
(\beditor{\bfnm{Samuel}\binits{S.}~\bsnm{Kaski}} \AND
  \beditor{\bfnm{Jukka}\binits{J.}~\bsnm{Corander}}, eds.).
\bseries{Proceedings of Machine Learning Research}
\bvolume{33}
\bpages{1015--1023}.
\bpublisher{PMLR}, \baddress{Reykjavik, Iceland}.
\end{binproceedings}
\endbibitem

\bibitem[\protect\citeauthoryear{Wilson, Hutter and Deisenroth}{2018}]{maxAcq}
\begin{binproceedings}[author]
\bauthor{\bsnm{Wilson},~\bfnm{James}\binits{J.}},
  \bauthor{\bsnm{Hutter},~\bfnm{Frank}\binits{F.}} \AND
  \bauthor{\bsnm{Deisenroth},~\bfnm{Marc}\binits{M.}}
(\byear{2018}).
\btitle{Maximizing acquisition functions for Bayesian optimization}.
In \bbooktitle{Advances in Neural Information Processing Systems}
(\beditor{\bfnm{S.}\binits{S.}~\bsnm{Bengio}},
  \beditor{\bfnm{H.}\binits{H.}~\bsnm{Wallach}},
  \beditor{\bfnm{H.}\binits{H.}~\bsnm{Larochelle}},
  \beditor{\bfnm{K.}\binits{K.}~\bsnm{Grauman}},
  \beditor{\bfnm{N.}\binits{N.}~\bsnm{Cesa-Bianchi}} \AND
  \beditor{\bfnm{R.}\binits{R.}~\bsnm{Garnett}}, eds.)
\bvolume{31}.
\bpublisher{Curran Associates, Inc.}
\end{binproceedings}
\endbibitem

\bibitem[\protect\citeauthoryear{Wilson et~al.}{2020}]{EfficientSampGPPost}
\begin{binproceedings}[author]
\bauthor{\bsnm{Wilson},~\bfnm{James~T.}\binits{J.~T.}},
  \bauthor{\bsnm{Borovitskiy},~\bfnm{Viacheslav}\binits{V.}},
  \bauthor{\bsnm{Terenin},~\bfnm{Alexander}\binits{A.}},
  \bauthor{\bsnm{Mostowsky},~\bfnm{Peter}\binits{P.}} \AND
  \bauthor{\bsnm{Deisenroth},~\bfnm{Marc~Peter}\binits{M.~P.}}
(\byear{2020}).
\btitle{Efficiently sampling functions from Gaussian process posteriors}.
In \bbooktitle{Proceedings of the 37th International Conference on Machine
  Learning}.
\bseries{ICML'20}.
\bpublisher{JMLR.org}.
\end{binproceedings}
\endbibitem

\bibitem[\protect\citeauthoryear{Wilson et~al.}{2021}]{pathwiseCond}
\begin{barticle}[author]
\bauthor{\bsnm{Wilson},~\bfnm{James~T.}\binits{J.~T.}},
  \bauthor{\bsnm{Borovitskiy},~\bfnm{Viacheslav}\binits{V.}},
  \bauthor{\bsnm{Terenin},~\bfnm{Alexander}\binits{A.}},
  \bauthor{\bsnm{Mostowsky},~\bfnm{Peter}\binits{P.}} \AND
  \bauthor{\bsnm{Deisenroth},~\bfnm{Marc~Peter}\binits{M.~P.}}
(\byear{2021}).
\btitle{Pathwise Conditioning of Gaussian Processes}.
\bjournal{Journal of Machine Learning Research}
\bvolume{22}
\bpages{1--47}.
\end{barticle}
\endbibitem

\bibitem[\protect\citeauthoryear{Wood}{2010}]{WoodSL}
\begin{barticle}[author]
\bauthor{\bsnm{Wood},~\bfnm{Simon~N.}\binits{S.~N.}}
(\byear{2010}).
\btitle{Statistical inference for noisy nonlinear ecological dynamic systems}.
\bjournal{Nature}
\bvolume{466}
\bpages{1102--1104}.
\bdoi{10.1038/nature09319}
\end{barticle}
\endbibitem

\bibitem[\protect\citeauthoryear{Wynn}{1970}]{WynnDiscreteExchange}
\begin{barticle}[author]
\bauthor{\bsnm{Wynn},~\bfnm{Henry~P.}\binits{H.~P.}}
(\byear{1970}).
\btitle{The Sequential Generation of D-Optimum Experimental Designs}.
\bjournal{The Annals of Mathematical Statistics}
\bvolume{41}
\bpages{1655--1664}.
\end{barticle}
\endbibitem

\bibitem[\protect\citeauthoryear{Xu et~al.}{2024}]{adaptiveMultimodal}
\begin{bmisc}[author]
\bauthor{\bsnm{Xu},~\bfnm{Zhihang}\binits{Z.}},
  \bauthor{\bsnm{Zhu},~\bfnm{Xiaoyu}\binits{X.}},
  \bauthor{\bsnm{Li},~\bfnm{Daoji}\binits{D.}} \AND
  \bauthor{\bsnm{Liao},~\bfnm{Qifeng}\binits{Q.}}
(\byear{2024}).
\btitle{An adaptive Gaussian process method for multi-modal Bayesian inverse
  problems}.
\end{bmisc}
\endbibitem

\bibitem[\protect\citeauthoryear{Yeh et~al.}{2017}]{PODemulation}
\begin{bmisc}[author]
\bauthor{\bsnm{Yeh},~\bfnm{Shiang-Ting}\binits{S.-T.}},
  \bauthor{\bsnm{Wang},~\bfnm{Xingjian}\binits{X.}},
  \bauthor{\bsnm{Sung},~\bfnm{Chih-Li}\binits{C.-L.}},
  \bauthor{\bsnm{Mak},~\bfnm{Simon}\binits{S.}},
  \bauthor{\bsnm{Chang},~\bfnm{Yu-Hung}\binits{Y.-H.}},
  \bauthor{\bsnm{Zhang},~\bfnm{Liwei}\binits{L.}},
  \bauthor{\bsnm{Wu},~\bfnm{C.~F.~Jeff}\binits{C.~F.~J.}} \AND
  \bauthor{\bsnm{Yang},~\bfnm{Vigor}\binits{V.}}
(\byear{2017}).
\btitle{Data-Driven Analysis and Common Proper Orthogonal Decomposition
  (CPOD)-Based Spatio-Temporal Emulator for Design Exploration}.
\bdoi{10.48550/arXiv.1709.07841}
\end{bmisc}
\endbibitem

\bibitem[\protect\citeauthoryear{Zammit-Mangion, Sainsbury-Dale and
  Huser}{2025}]{NeuralAmortizedReview}
\begin{barticle}[author]
\bauthor{\bsnm{Zammit-Mangion},~\bfnm{Andrew}\binits{A.}},
  \bauthor{\bsnm{Sainsbury-Dale},~\bfnm{Matthew}\binits{M.}} \AND
  \bauthor{\bsnm{Huser},~\bfnm{Raphaël}\binits{R.}}
(\byear{2025}).
\btitle{Neural Methods for Amortized Inference}.
\bjournal{Annual Review of Statistics and Its Application}
\bvolume{12}
\bpages{311-335}.
\bdoi{https://doi.org/10.1146/annurev-statistics-112723-034123}
\end{barticle}
\endbibitem

\bibitem[\protect\citeauthoryear{Zeng
  et~al.}{2022}]{turbulenceModelAdaptiveLHS}
\begin{barticle}[author]
\bauthor{\bsnm{Zeng},~\bfnm{Fanzhi}\binits{F.}},
  \bauthor{\bsnm{Zhang},~\bfnm{Wei}\binits{W.}},
  \bauthor{\bsnm{Li},~\bfnm{Jinping}\binits{J.}},
  \bauthor{\bsnm{Zhang},~\bfnm{Tianxin}\binits{T.}} \AND
  \bauthor{\bsnm{Yan},~\bfnm{Chao}\binits{C.}}
(\byear{2022}).
\btitle{Adaptive Model Refinement Approach for Bayesian Uncertainty
  Quantification in Turbulence Model}.
\bjournal{AIAA Journal}
\bvolume{60}
\bpages{3502-3516}.
\bdoi{10.2514/1.J060889}
\end{barticle}
\endbibitem

\bibitem[\protect\citeauthoryear{Zhang et~al.}{2016}]{hydrologicalModel2}
\begin{barticle}[author]
\bauthor{\bsnm{Zhang},~\bfnm{Jiangjiang}\binits{J.}},
  \bauthor{\bsnm{Li},~\bfnm{Weixuan}\binits{W.}},
  \bauthor{\bsnm{Zeng},~\bfnm{Lingzao}\binits{L.}} \AND
  \bauthor{\bsnm{Wu},~\bfnm{Laosheng}\binits{L.}}
(\byear{2016}).
\btitle{An adaptive Gaussian process-based method for efficient Bayesian
  experimental design in groundwater contaminant source identification
  problems}.
\bjournal{Water Resources Research}
\bvolume{52}
\bpages{5971-5984}.
\bdoi{https://doi.org/10.1002/2016WR018598}
\end{barticle}
\endbibitem

\bibitem[\protect\citeauthoryear{Zhang et~al.}{2020}]{hydrologicalModel}
\begin{barticle}[author]
\bauthor{\bsnm{Zhang},~\bfnm{Jiangjiang}\binits{J.}},
  \bauthor{\bsnm{Zheng},~\bfnm{Qiang}\binits{Q.}},
  \bauthor{\bsnm{Chen},~\bfnm{Dingjiang}\binits{D.}},
  \bauthor{\bsnm{Wu},~\bfnm{Laosheng}\binits{L.}} \AND
  \bauthor{\bsnm{Zeng},~\bfnm{Lingzao}\binits{L.}}
(\byear{2020}).
\btitle{Surrogate-Based Bayesian Inverse Modeling of the Hydrological System:
  An Adaptive Approach Considering Surrogate Approximation Error}.
\bjournal{Water Resources Research}
\bvolume{56}
\bpages{e2019WR025721}.
\bnote{e2019WR025721 2019WR025721}.
\bdoi{https://doi.org/10.1029/2019WR025721}
\end{barticle}
\endbibitem

\bibitem[\protect\citeauthoryear{Özge Sürer}{2026}]{SurerBatchStochastic}
\begin{barticle}[author]
\bauthor{\bparticle{Özge} \bsnm{Sürer}}
(\byear{2026}).
\btitle{Batch Sequential Experimental Design for Calibration of Stochastic
  Simulation Models}.
\bjournal{Technometrics}
\bvolume{68}
\bpages{1--13}.
\bdoi{10.1080/00401706.2025.2520860}
\end{barticle}
\endbibitem

\bibitem[\protect\citeauthoryear{Özge Sürer, Plumlee and
  Wild}{2024}]{Surer2023sequential}
\begin{barticle}[author]
\bauthor{\bparticle{Özge} \bsnm{Sürer}},
  \bauthor{\bsnm{Plumlee},~\bfnm{Matthew}\binits{M.}} \AND
  \bauthor{\bsnm{Wild},~\bfnm{Stefan~M.}\binits{S.~M.}}
(\byear{2024}).
\btitle{Sequential Bayesian Experimental Design for Calibration of Expensive
  Simulation Models}.
\bjournal{Technometrics}
\bvolume{66}
\bpages{157--171}.
\bdoi{10.1080/00401706.2023.2246157}
\end{barticle}
\endbibitem

\end{thebibliography}

\end{document}